\documentclass[conference]{IEEEtran}
\IEEEoverridecommandlockouts

\def \ifanon {\if0}

\usepackage{cite}
\usepackage{amsmath,amssymb,amsfonts}
\usepackage{cases}

\usepackage{graphicx}
\usepackage{textcomp}
\usepackage{xcolor}
\def\BibTeX{{\rm B\kern-.05em{\sc i\kern-.025em b}\kern-.08em
    T\kern-.1667em\lower.7ex\hbox{E}\kern-.125emX}}

\usepackage{multirow}
\usepackage{graphicx}
\usepackage{color}
\usepackage{wrapfig}

\usepackage{amsthm}
\usepackage{mathtools,mathrsfs}
\usepackage{tikz}
\usetikzlibrary{shapes.geometric}

\usepackage{silence}
\usepackage{caption}
\usepackage{subcaption}

\usepackage{verbatim}
\usepackage{diagbox}
\usepackage{float}
\usepackage{amsfonts} 

\usepackage{enumitem}
\newlist{enumthm}{enumerate}{1}
\setlist[enumthm]{label=(\alph*)}
\usepackage{amsmath}
\usepackage{varwidth}
\usepackage{tasks}
\usepackage{makecell}
\usepackage{multirow}
\usepackage{colortbl}

\usepackage{algorithm}
\usepackage[noend]{algpseudocode}
\algnewcommand\algorithmicforeach{\textbf{for each}}
\algnewcommand{\IfThenElse}[3]{
  \State \algorithmicif\ #1\ \algorithmicthen\ #2\ \algorithmicelse\ #3}
\algnewcommand{\IfThen}[2]{
  \State \algorithmicif\ #1\ \algorithmicthen\ #2}

\algrenewtext{ForEach}[1]{\algorithmicforeach\ #1\ \algorithmicfordo}
\algrenewtext{EndForEach}{}

\algdef{SnE}[FOREACH]{ForEach}{EndForEach}[1]{\algorithmicforeach\ #1}%
\algnewcommand\algorithmicswitch{\textbf{switch}}
\algnewcommand\algorithmiccase{\textbf{case}}
\algnewcommand\algorithmicassert{\texttt{assert}}
\algnewcommand\Assert[1]{\State \algorithmicassert(#1)}%
\algdef{SE}[SWITCH]{Switch}{EndSwitch}[1]{\algorithmicswitch\ #1\ \algorithmicdo}{\algorithmicend\ \algorithmicswitch}%
\algdef{SE}[CASE]{Case}{EndCase}[1]{\algorithmiccase\ #1}{\algorithmicend\ \algorithmiccase}%
\algtext*{EndSwitch}%
\algtext*{EndCase}%

\theoremstyle{plain}
\newtheorem{theorem}{Theorem}
\numberwithin{theorem}{subsection} 
\newtheorem{corollary}[theorem]{Corollary}
\newtheorem{lemma}[theorem]{Lemma}
\newtheorem{proposition}[theorem]{Proposition}
\theoremstyle{definition}
\newtheorem{definition}[theorem]{Definition}
\newtheorem{example}[theorem]{Example}
\newtheorem{construction}[theorem]{Construction}
\newtheorem{property}[theorem]{Property}

\theoremstyle{definition}

\theoremstyle{remark}
\newtheorem{remark}[theorem]{Remark}

\numberwithin{case}{theorem}

\numberwithin{subcase}{case}

\allowdisplaybreaks[1]
\newcommand{\floor}[1]{\left\lfloor#1\right\rfloor}

\usepackage[createShortEnv]{proof-at-the-end}
\usepackage{thmtools}
\usepackage{thm-restate}

\usepackage[hypertexnames=false,hidelinks]{hyperref}
\usepackage{cleveref}

\usepackage{adjustbox,lipsum}
\usepackage{tablefootnote}

\usepackage{color-edits}
\addauthor[Kartik]{kl}{blue}
\addauthor[Laura]{lm}{red}
\addauthor[Kelly]{ki}{orange}
\addauthor[Ales]{ak}{brown}
\addauthor[Daniel]{dh}{magenta}
\addauthor[Aleyah]{ad}{yellow}

\begin{document}
\ifanon
\title{Edge-Disjoint Spanning Trees on Star Products}
\maketitle

\else

\title{Edge-Disjoint Spanning Trees on Star Products
\thanks{This work was supported in part by the Colgate University Picker Interdisciplinary Science Institute major grant, by the U.S. Department of Energy through Los Alamos National Laboratory, operated by Triad National Security, LLC, for the U.S. DOE (Contract \# 89233218CNA000001), and by LANL LDRD Project \# 20230692ER. The U.S. Government retains an irrevocable, nonexclusive, royalty-free license to publish, translate, reproduce, use, or dispose of the published form of the work and to authorize others to do the same. 
This paper has LANL identification number LA-UR-25-24115.
}
}

\makeatletter
\newcommand{\linebreakand}{%
  \end{@IEEEauthorhalign}
  \hfill\mbox{}\par
  \mbox{}\hfill\begin{@IEEEauthorhalign}
}
\makeatother

\author{\IEEEauthorblockN{
Kelly Isham \textsuperscript{\textsection}}
\IEEEauthorblockA{\textit{Department of Mathematics} \\
\textit{Colgate University}\\
Hamilton, NY, USA \\
kisham@colgate.edu}
\and
\IEEEauthorblockN{
Laura Monroe \textsuperscript{\textsection}}
\IEEEauthorblockA{\textit{Ultrascale Systems Research Center} \\
\textit{Los Alamos National Laboratory}\\
Los Alamos, NM, USA \\
lmonroe@lanl.gov}
\and
\IEEEauthorblockN{
Kartik Lakhotia \textsuperscript{\textsection}}
\IEEEauthorblockA{\textit{Intel Labs} \\
\textit{Intel}\\
Santa Clara, CA, USA \\
kartik.lakhotia@intel.com}
\linebreakand
\IEEEauthorblockN{
Aleyah Dawkins}
\IEEEauthorblockA{\textit{Department of Mathematics} \\
\textit{Carnegie Mellon University}\\
Pittsburgh, PA, USA \\
adawkins@andrew.cmu.edu}
\and
\IEEEauthorblockN{
Daniel Hwang}
\IEEEauthorblockA{\textit{Department of Mathematics} \\
\textit{Georgia Institute of Technology}\\
Atlanta, GA, USA \\
dhwang48@gatech.edu}
\and
\IEEEauthorblockN{
Ales Kubicek}
\IEEEauthorblockA{\textit{Department of Computer Science} \\
\textit{ETH Z\"{u}rich}\\
Z\"{u}rich, Switzerland \\
akubicek@student.ethz.ch}
}

\maketitle
\begingroup\renewcommand\thefootnote{\textsection}
\makeatletter\def\Hy@Warning#1{}\makeatother
\footnotetext{\ The first three authors contributed equally to this work.}
\endgroup
\fi

\IEEEpeerreviewmaketitle
\thispagestyle{plain}
\pagestyle{plain}

\begin{abstract}    A star-product operation may be used to create large graphs from smaller factor graphs.
Network topologies based on star-products demonstrate several advantages including low-diameter, high scalability, modularity and others. 
Many state-of-the-art diameter-2 and -3  topologies~(Slim Fly, Bundlefly, PolarStar etc.) 
can be represented as star products. 

In this paper, we explore constructions of edge-disjoint spanning trees~(EDSTs) in star-product topologies. 
EDSTs expose multiple parallel disjoint 
pathways in the network and can be leveraged to accelerate collective 
communication, enhance fault tolerance and network recovery, and manage congestion. 

Our EDSTs have provably maximum or near-maximum cardinality which amplifies their benefits. 
We further analyze their depths and show that for one of our constructions, all trees have order of the depth of the EDSTs of the factor graphs, and for all other constructions, a large subset of the trees have that depth. 

\end{abstract}

\section{Introduction and Motivation}
\subsection{The Star Product in Network Topologies}
A network topology can be modeled as a graph connecting nodes (switches or endpoints) using links. One way to build a scalable topology is via a graph product, 
composing a large graph out of two smaller graphs called \emph{factor graphs}.

We explore here the star product, a family of graphs that has been used to construct some of the largest known low-diameter network topologies. Several new state-of-the-art network topologies are based on the star product \cite{slim-fly, bundlefly_2020,PolarStar_23}.

The star product network generalizes the Cartesian product network (sometimes known simply as the product network), whose networking benefits are well understood~\cite{cart_prod_networks_1997,Youssef1991CartesianPN} and used in popular topologies such as  HyperX~\cite{HyperX_2009}, Torus~\cite{bluegene_2011}, etc. 
The star product shares these benefits and structural characteristics, and also improves significantly on the Cartesian product with lower diameter and higher scale.

In this paper, we show a construction of maximum or near maximum sets of edge-disjoint spanning trees (EDSTs) in  star product graphs.
Such large sets of EDSTs can be used to improve fault-tolerance in large-scale systems, broadcast reliabilty and parallelism in collective operations. 
Thus, our work adds to the strengths of star-product topologies and makes them more compelling for real-world deployment.

\subsection{Edge-Disjoint Spanning Trees in Networking}
A set of spanning trees in a graph $G$ is \emph{edge-disjoint} if no two contain a common edge. Collective operations, especially in-network Allreduces, 
are typically implemented over spanning trees in the network~\cite{lakhotia2021network,graham2016scalable, graham2020scalable, lakhotia2021accelerating, allreduce_PF_2023, de2024canary}, using multiple EDSTs to parallelize  execution. 
EDSTs provide independent pathways between all nodes and can be used for 
reliable communication under component failures or broadcasting system state~\cite{fragopoulou_fault_tolerance, awerbuch1986reliable, petrini2001hardware, Product_STs_2003}. 
A maximum-sized set of EDSTs  maximizes collective bandwidth and fault-tolerance of a system, which  improves performance of HPC and Machine Learning workloads~\cite{allreduce_PF_2023, petrini2001hardware, graham2020scalable, petrini2002quadrics, Product_STs_2003}. 
Thus, a large number of EDSTs is advantageous for any network and generating a maximum-sized set of EDSTs is of great interest~(as shown for PolarFly in~\cite{allreduce_PF_2023}). 

In this paper, we present constructive methods to generate near-maximum or (under certain conditions) maximum-sized 
sets of EDSTs for star-product graphs using only the EDSTs from smaller and less complex factor graphs.

The depth of EDSTs is also critical, as it directly affects the latency and efficiency of communication over the tree. 
The \emph{depth of a tree} is defined as the maximum path length from any vertex to the root. In this paper, we assume that the \emph{center is the root}.
The explicit constructions in this paper enable analysis and optimization of depth and other structural properties of EDSTs.
A large fraction of our EDSTs have a depth on the order of the depth of EDSTs in factor graphs, 
and also present a trade off between cardinality and EDST depth.
Thus, our approach has an advantage over generic algorithms for EDSTs such as Roskind-Tarjan algorithm that do not leverage structural insights of the underlying graph, making it difficult to predict or control the depth~\cite{roskind_paper, roskind_thesis}.

\subsection{Contributions }

\begin{itemize}[itemsep=0pt,parsep=2pt,leftmargin=*]
    \item \textbf{In-depth discussion of the star product, derived topologies and its generalization of the Cartesian product.} We discuss several state-of-the-art star-product network topologies, one of which is newly identified as a star product (Section~\ref{sec:slimfly}). 
    We discuss the star product as a generalization of the Cartesian product, and open the possibility of extending features of the Cartesian product to any star-product topology. We also list the differences 
    between these two products highlighting the need
    for non-trivial methods to generalize the results on Cartesian products.
    
     \item \textbf{General construction of a maximum or near-maximum number of EDSTs for any star product. The near-maximum is within $\mathbf{1}$ of the maximum}, tightening the bounds we present in \cite{edst_ipdps_2025} by $1$.  Our methods construct EDSTs in star-product networks using those
     in the factor graphs, reducing the problem 
     to that of finding EDSTs in smaller, simpler factor graphs. 
    We develop a Universal solution to construct a near-maximum number of EDSTs for any star product, and give additional constructions that reach maximum size under certain conditions.
    Our Universal construction intuitively generalizes methods in \cite{Product_STs_2003}, but our construction of additional trees is different. 
    Theorems that prove the constructions are listed in Table~\ref{tab:star_ths_all}, and the constructions themselves are given in Section~\ref{sec:star_edsts}. 

    \item \textbf{Analysis of the depth of constructed EDSTs.} We show our Universal solution has depth on the order of the factor-graph EDST depths 
    (Theorems \ref{thm:depth_universal_1}, \ref{thm:depth_universal_2}). 
    Our Maximum constructions have larger worst-case depth but also have a large subset of small-depth EDSTs, as in the Universal solution 
    (Theorem \ref{thm:depth_maximum}, Corollary \ref{cor:depth_maximum}). 
    This leads to a tradeoff between number of EDSTs and their depth. 
    
    \item \textbf{Detailed characterization of EDSTs in existing star-product network topologies.} 
    We apply our constructions over different configurations of emerging star-product topologies: Slim Fly~\cite{slim-fly}, BundleFly~\cite{bundlefly_2020} and PolarStar~\cite{PolarStar_23}. Our constructions achieve the maximum number of EDSTs in most of these existing networks, as shown in~Table~\ref{table:network_graphs}.
\end{itemize}

\section{The Cartesian and the Star Products}

\subsection{Graph, Vertex and Edge Notation}
$G_i=(V_i, E_i)$ denotes a graph with vertex set $V_i$ and edge set $E_i$.
Cartesian and star products are product graphs,  composed of two \emph{factor graphs}, the first called the \emph{structure graph} and the 
second the \emph{supernode}. 

For vertices $x,y$ in $V_i$, we use $\{x,y\}$ to denote an undirected edge and $(x,y)$ to denote a directed edge.
Vertices of a product graph are also represented as ordered pairs of vertex indices denoted $(x,y)$, where the first index $x$ is derived from the structure graph and second index $y$ from the supernode. 
It will be clear from the context whether the notation $(x,y)$ denotes a directed edge or a vertex in a product graph.

\subsection{Intuition and Definitions of Cartesian and Star Products}\label{sec:defs}
Intuitively, the type of product graphs discussed in this paper are made by replacing individual vertices in the structure graph by
instances of the supernode. Thus, each supernode instance corresponds to a vertex in the structure graph whose ID 
is assigned to it. A vertex in the product graph is identified by an ordered pair $(x,y)$ where $x$ is the  ID of the instance in which it resides and $y$ is its ID within the supernode. 
The instances are connected to each other by creating edges between vertices across different instances according to some rule.

In both star and Cartesian products, inter-supernode connectivity follows the edges in the structure graph. Vertices in one supernode instance are bijectively connected to vertices in another supernode instance if and only if the structure graph vertices corresponding to these instances are adjacent.

In a Cartesian product, the edges between supernodes must join equivalent vertices having the same ID within their respective supernodes.
An example of a Cartesian product is shown in Figure~\ref{fig:cart_star_comp_2_cart} and its formal definition is given below.
\begin{definition}\label{def:cart_product}[\textbf{Cartesian Product}]
The Cartesian product $G^\times= G_s\times G_n$ of two graphs $G_s$ and   $G_n$ is defined as follows:
\begin{enumerate}[itemsep=0pt,parsep=0pt]
    \item $V_\times$ is $\{(x, y) \mid x \in V_s, y \in V_n\}$.
    \item Vertices $(x,y)$ and $(x',y')$ are adjacent if and only if
\begin{enumerate}[itemsep=0pt,parsep=0pt]
        \item $x=x'$ and $\{y,y'\} \in E_n$, or 
        \item $\{x,x'\} \in E_s$ and $y=y'$.
    \end{enumerate}
\end{enumerate}
\end{definition}
In the star product, adjacent supernode instances are still joined by bijections between their vertex sets, 
but the Cartesian requirement of having edges only between equivalent vertices in supernodes is relaxed. 
Instead, edges between adjacent supernode instances $x$ and $x'$ may be defined by any convenient bijection $f$, and $f$ can vary across the instance pairs. 
In effect, the Cartesian edges between supernodes are ``twisted'' in a star product. The Cartesian product is then a special case of the star product, where $f$ is the identity for all supernode pairs.

The star product was introduced by Bermond et. al in \cite{bermond82}.
Their definition and a useful property are below.
Examples of non-Cartesian star products are shown in Figures~\ref{fig:cart_star_comp_20_star} and \ref{fig:cart_star_comp_21_star}.
\begin{definition}\label{def:star_product}~\cite{bermond82} [\textbf{Star Product}]
    Let $G_s$ and $G_n$ be two graphs, and let $A_s$ be the set of arcs on $G_s$ resulting from an arbitrary orientation of the edges of $G_s$. For each arc $(x,x')$ in $A_s$, let $f_{(x,x')}: V_n\rightarrow V_n$ be a bijection. 
The star product $G^*= G_s*G_n$ is defined as follows:
\begin{enumerate}[itemsep=0pt,parsep=0pt]
    \item $V_*$ is $\{(x, y) \mid x \in V_s, y \in V_n\}$.
    \item Vertices $(x,y)$ and $(x',y')$ are adjacent if and only if
    \begin{enumerate}[itemsep=0pt,parsep=0pt]
        \item $x=x'$ and $\{y, y'\}\in E_n$, or 
        \item $(x,x')\in A_s$ and $y'=f_{(x,x')}(y)$.
 \end{enumerate}
\end{enumerate}
\end{definition}
\begin{restatable}[]{property_restate}{props}
\label{property:star_props}\cite{bermond82}
    Let $G^*=G_s*G_n$ be a star product.
    \begin{enumerate}[noitemsep]
        \item  The number of vertices in $G^*$ is $\rvert V_* \lvert =\lvert V_s \rvert \lvert V_n \rvert$.\label{th:star_props_order}
        \item When there are no self-loops in $G_s$, the number of edges  in $G^*$ is $\lvert E_*\rvert = \lvert V_s\rvert \lvert E_n\rvert + \lvert V_n\rvert \lvert E_s\rvert$.\label{th:star_props_edges}
    \end{enumerate}
\end{restatable}
\begin{figure}[!ht]
    \centering
    \begin{subfigure}[t]{.22\linewidth}
        \centering      
        \includegraphics[width=.60\linewidth]%
        {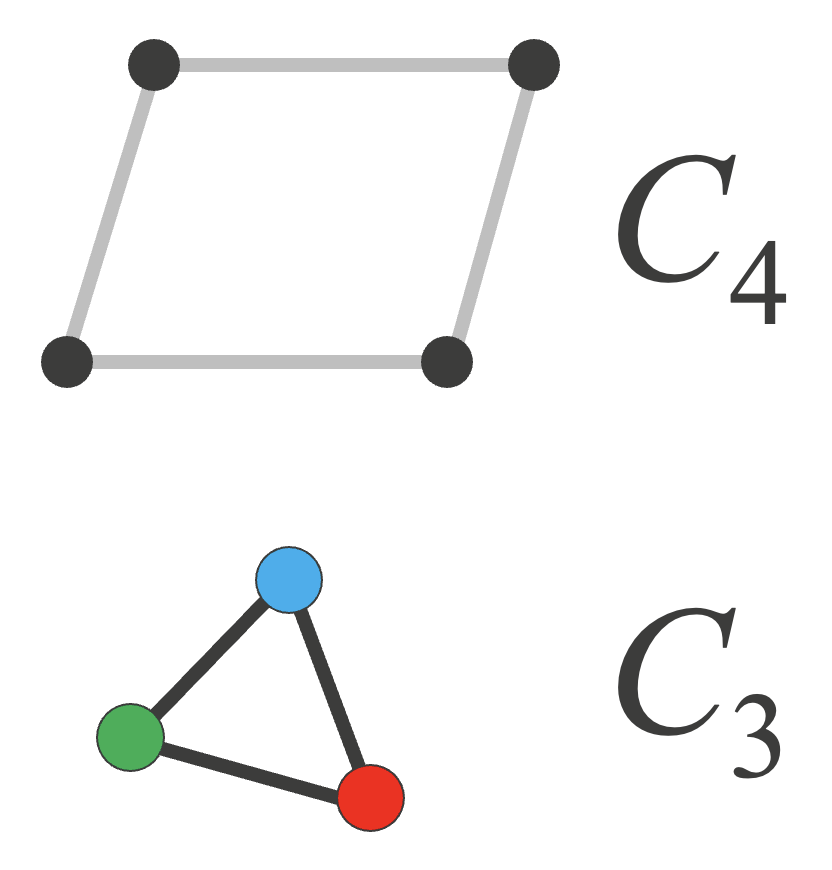}        \caption{Structure graph $C_4$  and supernode $C_3$. }\label{fig:factors_2_cart}
    \end{subfigure}\hfill
    \begin{subfigure}[t]{.22\linewidth}
        \centering      
        \includegraphics[width=.95\linewidth]{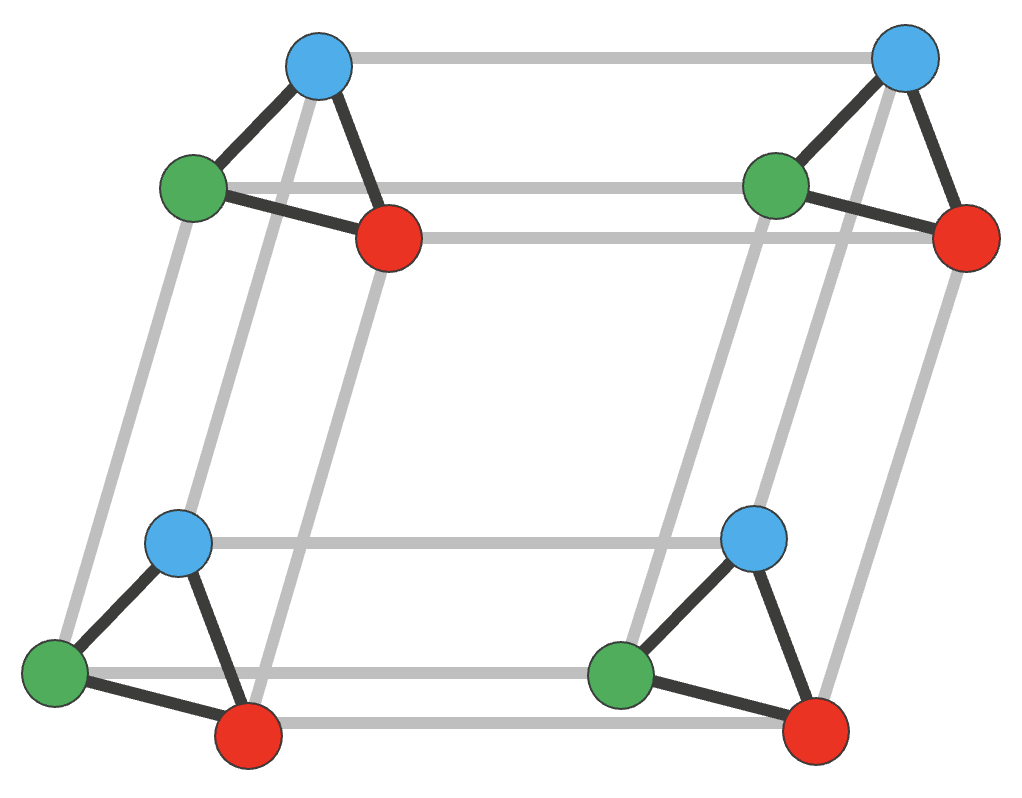}
        \caption{The Cartesian $C_4 \times C_3$. 
        }\label{fig:cart_star_comp_2_cart}
    \end{subfigure}\hfill
    \begin{subfigure}[t]{.22\linewidth}
        \centering
        \includegraphics[width=.95\linewidth]{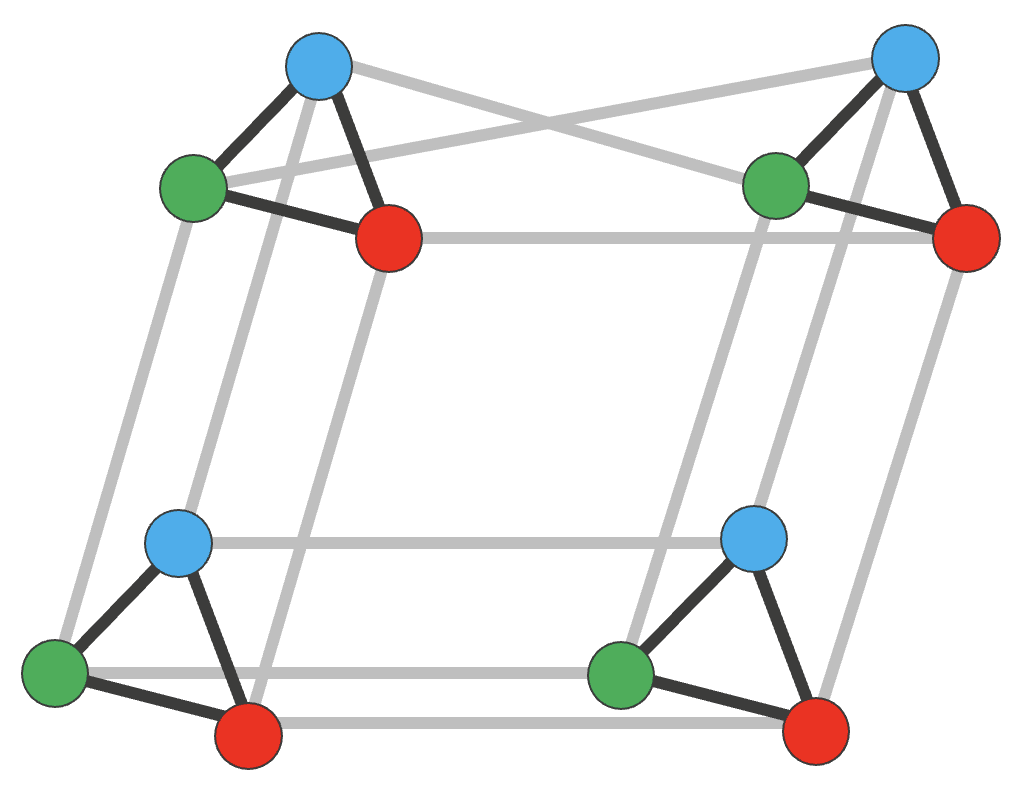}
        \caption{A $C_4 * C_3$ star product. 
       }\label{fig:cart_star_comp_20_star}
    \end{subfigure}\hfill
    \begin{subfigure}[t]{.22\linewidth}
        \centering
        \includegraphics[width=.95\linewidth]{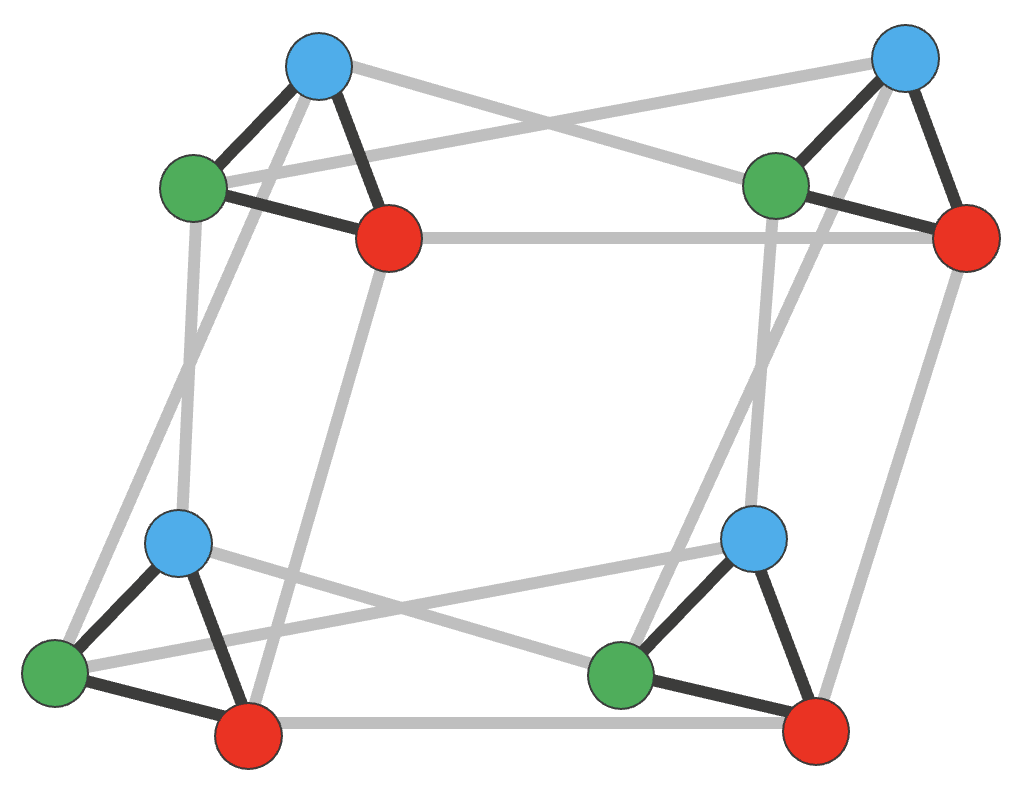}
        \caption{Another $C_4 * C_3$ star product. 
        }
        \label{fig:cart_star_comp_21_star}
    \end{subfigure}
\caption{Comparing Cartesian and star products on structure graph $C_4$ and supernode $C_3$. 
Figure \ref{fig:cart_star_comp_2_cart} is the unique Cartesian product on $C_4$ and  $C_3$, joining equivalent vertices of the same color. Figures \ref{fig:cart_star_comp_20_star} and \ref{fig:cart_star_comp_21_star} are two non-Cartesian star products. 
} \label{fig:cart_star_comp}
\end{figure}
It is interesting to note that the structure graph and supernode are not always interchangeable in the star product,
in contrast to the Cartesian product. For example, the constrained diameter discussed in \cite{bermond82} and \cite{PolarStar_23} depends 
on 
the ordering of the two factor graphs
and changing the order may result in a different product graph with a larger diameter.

A Cartesian product with either a disconnected structure graph or disconnected supernode will be disconnected. 
A star product with a disconnected structure graph is always disconnected. 
However, a star product with a disconnected supernode can be connected. 
This is because a well-chosen $f$ can serve to connect distinct components in a supernode instance, by connecting a vertex from each to a pair of vertices in a single connected component in another instance. 

Non-trivial connected star-product networks having disconnected supernodes exist. One example is the PolarStar network special case $ER_q * IQ(0)$ discussed in \cite{PolarStar_23}. Another example is the quantum Chimera network \cite{chimera_2014, chimera_2016}. 
However, most real-world examples we have seen have connected supernodes.

In this paper, we use EDSTs in the factor graphs to construct EDSTs in product graph, so we consider only star products in which both structure and supernode graphs are connected.
\begin{table*}[!ht]
\renewcommand\arraystretch{1.5}
\centering
\footnotesize
\noindent\adjustbox{max width=\textwidth}
{
\begin{tabular}{|l|l|l|l|l|l|l|l|}
\hline
\rowcolor[HTML]{EFEFEF} 
\textbf{Conditions}  & \textbf{\# EDSTs} & \textbf{Max?} & \textbf{EDST Thm.} &\textbf{Max Depth} &\textbf{Min Depth} 
&\textbf{\begin{tabular}[c]{@{}c@{}}\# Min-Depth \\EDSTs \end{tabular}} 
&\textbf{Depth Thms.} \\
\hline
\begin{tabular}[c]{@{}l@{}}$r_s\ge t_s$ AND $r_n\ge t_n$\end{tabular} & $t_s+ t_n$ & Within $1$ & \ref{cor:r_ge_t} &  $(d_{s_1}+1)(2d_{n_1} + 2r_s) + d_{s_1}$& $2d_{s_i} + d_{n_1}$ & $t_s-1$&\ref{thm:depth_maximum}, \ref{cor:depth_maximum} \\ \hline
\begin{tabular}[c]{@{}l@{}}$r_s=t_s$ AND $r_n=t_n$\end{tabular} & $t_s+ t_n$ & \textbf{YES} & \ref{th:r_eq_t} & $(d_{s_1}+1)(2d_{n_1} + 2r_s) + d_{s_1}$ & $2d_{s_i} + d_{n_1}$&$t_s-1$ &\ref{thm:depth_maximum}, \ref{cor:depth_maximum}\\ \hline
\begin{tabular}[c]{@{}l@{}}$r_s\ge t_s$ OR $r_n\ge t_n$\end{tabular} & $t_s+ t_n-1$ & Within $1$ & \ref{cor:t_r_ge_or} &  $(d_{s_1}+1)(2d_{n_1} + 2r_s) + d_{s_1}$& $2d_{s_i} + d_{n_1}$& $t_s-1$&\ref{thm:depth_maximum}, \ref{cor:depth_maximum}\\ \hline

\begin{tabular}[c]{@{}l@{}}
Property~\ref{property:no_edge_conditions}, $r_s< t_s$ AND $r_n<t_n$\end{tabular} &
$t_s+ t_n-1$&  
\textbf{YES} & 
\ref{th:t_r_lt} & 
$2d_{s_1}d_{n_i}+d_{n_i} + d_{s_1}$& 
$4d_{s_i} + d_{n_1}$&
$t_s-1$&
\ref{th:depth_r_lt_t}\\ \hline

Property~\ref{property:no_edge_conditions} & 
$t_s+ t_n-1$ & 
Within $1$ & 
\ref{th:t_r_lt} & 
$2d_{s_1}d_{n_i}+d_{n_i} + d_{s_1}$&
$4d_{s_i} + d_{n_1}$  &
$t_s-1$&
\ref{th:depth_r_lt_t}\\ \hline

\begin{tabular}[c]{@{}l@{}}\end{tabular}Any star product (Universal) & $t_s+ t_n-2$ & Within $3$ & \ref{th:t_r_generic} 
& $\max(2d_{s_i} + d_{n_1},d_{s_1} + 2d_{n_i})$ 
& $\min(2d_{s_i} + d_{n_1},d_{s_1} + 2d_{n_i})$ &$t_s-1$ or $t_n-1$&\ref{thm:depth_universal_1}, \ref{thm:depth_universal_2}\\ \hline
\end{tabular}
}
\caption{This table summarizes the main results of this paper. It includes the number and depth of EDSTs that are constructed here on star products with connected factor graphs, and how closely these constructions approach maximum cardinality. All have cardinality linear in the cardinality of the EDSTs of the factor graphs. Although worst-case depths are quadratic in the depths of the EDSTs of $G_s$ and $G_n$, all constructions still have a large subset of small depth trees. A comparison of the advantages and disadvantages of the maximal and lower-depth solutions is in Section~\ref{sec:discussion}. 
Star products are analyzed in terms of the conditions on $t_i$ (the number of EDSTs of factor $G_i$) and $r_i$ (the number of remaining edges in $G_i$ that are not in any EDST). We denote by $d_{s_i}$ and $d_{n_i}$ the depth of the $i^{th}$ EDST used from the structure graph $G_s$ and supernode $G_n$, respectively.}
\label{tab:star_ths_all}
\end{table*}

\subsection{Advantages of the Star Product}\label{sec:advantages}
The Cartesian product has many advantages for networking such as modular structure, amenability to link bundling and mathematical structure exploitable for routing algorithms, deadlock avoidance and other use cases~\cite{cart_prod_networks_1997,Youssef1991CartesianPN,bundlefly_2020}. Star products have additional advantages beyond these,
including:
\begin{itemize} 
[itemsep=0pt,parsep=2pt,leftmargin=*]
    \item \textbf{Constrained diameter.} If the star product is constructed carefully, with a proper choice of bijective connections, the diameter of the star product is constrained 
    to be at most $1$ more than that of its first factor graph \cite{bermond82,PolarStar_23, bundlefly_2020}. This is in contrast to the Cartesian graph, where the diameter is the sum of the diameters of the factor graphs \cite{cart_prod_networks_1997}.
    
    \item \textbf{Flexibility in network design.} Any desired bijective function may be used to create edges between supernodes.

    \item \textbf{High scalability.} PolarStar~\cite{PolarStar_23}, based on star products, gives some of the largest known diameter-3 networks for most radixes. 
    Bundlefly~\cite{bundlefly_2020} and Slim Fly~\cite{slim-fly} are other examples of star-product networks with near state-of-the-art scalability.
\end{itemize}

These advantages make the star product a good candidate for 
large HPC and datacenter networks, capable of meeting the growing demands for AI and scientific computing workloads.

Cartesian products are very well-studied, with extensive known results in the area of networking~\cite{cart_prod_networks_1997,Youssef1991CartesianPN, bisect_BW_cart_2014, fault_diam_cart_2000, deadlock_cart_1998, resource_placement_cart_2010}. Star products are  much less studied.
A long-term goal will be to generalize the wealth of knowledge on Cartesian products to star products, in view of impact to the field. 

We take a step in this direction by exploring construction of EDSTs in star product networks. 
EDSTs are important structures from a networking perspective with direct applications in collective communication, system reliability and other tasks. 
Our results generalize similar results on Cartesian product networks~\cite{Product_STs_2003}. 
We also perform an extensive depth analysis and characterization of EDSTs in existing star product networks, highlighting the trade-offs across different solutions.

\subsection{Some Examples of State-of-the-Art Star Product Networks
}\label{sec:examples}
 
In this section, we discuss several examples of state-of-the-art diameter-2 and -3 star product network topologies. 
Applications of theorems in this paper to these networks may be seen in Table~\ref{table:network_graphs}. The maximum number of EDSTs are provably constructed for almost all of these networks.

\begin{example} \textbf{[Cartesian product / Mesh, Torus, HyperX~\cite{HyperX_2009}]}
    These are star products, since they are Cartesian.
\end{example}
\begin{example}\label{sec:slimfly} \textbf{[McKay-Miller-Širáň graph~\cite{MMS_98} / Slim Fly~\cite{slim-fly}]}
    The \emph{McKay-Miller-Širáň (MMS) graph $H_q$}, is a star product parameterized by a prime power $q$.
    Its structure graph is the complete bipartite graph $K_{q,q}$ and its supernode is a Cayley graph.
    Internal supernode edges and the bijective $\rho$ edges between supernodes are defined via constructs over the finite fields $\mathbb{F}_q$,
    given by certain Cayley graphs $C(q)$ (see \cite{MMS_98,slim-fly} for details). $H_q$ is a well-known graph 
    that has not previously been recognized as a star product. It was used for the diameter-2 \emph{Slim Fly} network. 
\end{example}

\begin{example} \textbf{[$P$ star-product graph~\cite{bermond82} / Bundlefly~\cite{bundlefly_2020}]}
    The \emph{$P$ star product} was defined in the same paper as the star product \cite{bermond82}, with properties $P$ that give better constraints on diameter than the generic star product.     
    \emph{Bundlefly} is a diameter-$3$ network based on the $P$ star product, with the MMS graph as the structure graph and the  
    Paley graph \cite{erdosrenyi_paley_1963} as the supernode.    
\end{example}
\begin{example}\label{sec:polarstar} \textbf{[$R$ star-product graph~\cite{PolarStar_23} / PolarStar~\cite{PolarStar_23}]}
    The \emph{$R$ star product} \cite{PolarStar_23} is similar to the $P$ star product, but uses a different set of properties $R$. 
    The \emph{PolarStar} network construction is an $R$ star product with
    the Erd\H os-R\'enyi polarity graph ($ER_q$)~\cite{erdosrenyi1962,brown_1966} as the structure graph,
    and the Paley graph (QR) or Inductive-Quad~(IQ)~\cite{PolarStar_23} as a supernode.
    Both $ER_q$ and IQ are optimally large for graphs with the $R$ (or $P$) properties,
    so PolarStar achieves higher scale than any known diameter-3  network or graph 
    for many radixes. 
\end{example}
\begin{example}\label{sec:petersen} \textbf{[A diameter-2 example asymptotically having 50\% of the upper bound on the number of vertices]}  This is the star product of a complete structure graph and a Paley graph supernode. When the structure graph is the complete graph on $2$ vertices and the Paley graph is the cycle on $5$ vertices, the Petersen graph \cite{petersen_1886, petersen_1898} is produced.
\end{example}
\begin{example}\textbf{[Chimera graph~\cite{chimera_2016} / D-Wave 2000Q~\cite{chimera_2014}]}
The \emph{Chimera graph} was the graph used in production in the first D-Wave adiabatic quantum computers~\cite{chimera_2016}. The $n$-dimensional Chimera ($C_n$) is a building block of D-Wave's Pegasus~\cite{pegasus_2020} and Zephyr~\cite{zephyr_2021} topologies. Chimera is shown to be a star product in Appendix~\ref{sec:appendix_chimera}.
\end{example}

\section{Intuition for EDST Construction}\label{sec:EDSTs_products}
\begin{figure}[!ht]
    \centering
    \begin{subfigure}[t]{.22\linewidth}
        \centering      
        \includegraphics[width=1.0\linewidth]{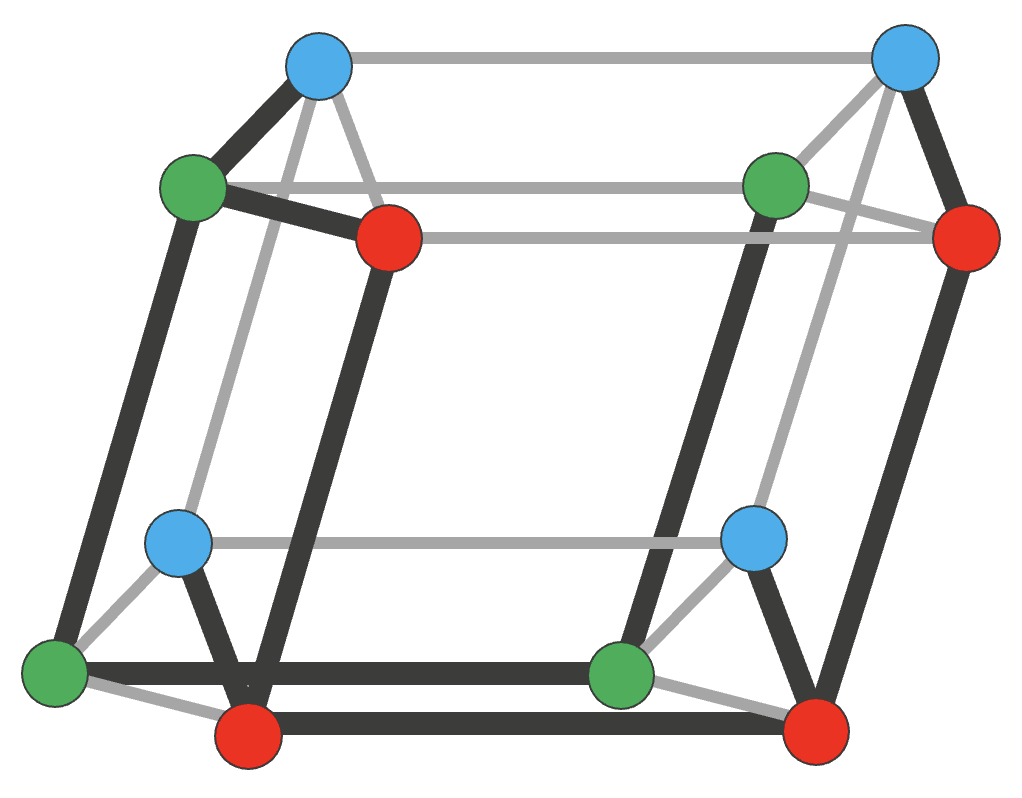}
        \caption{A Cartesian spanning tree (bolded).}\label{fig:const_counterex_cart1}
    \end{subfigure}\hfill
    \begin{subfigure}[t]{.22\linewidth}
        \centering      
        \includegraphics[width=1.0\linewidth]{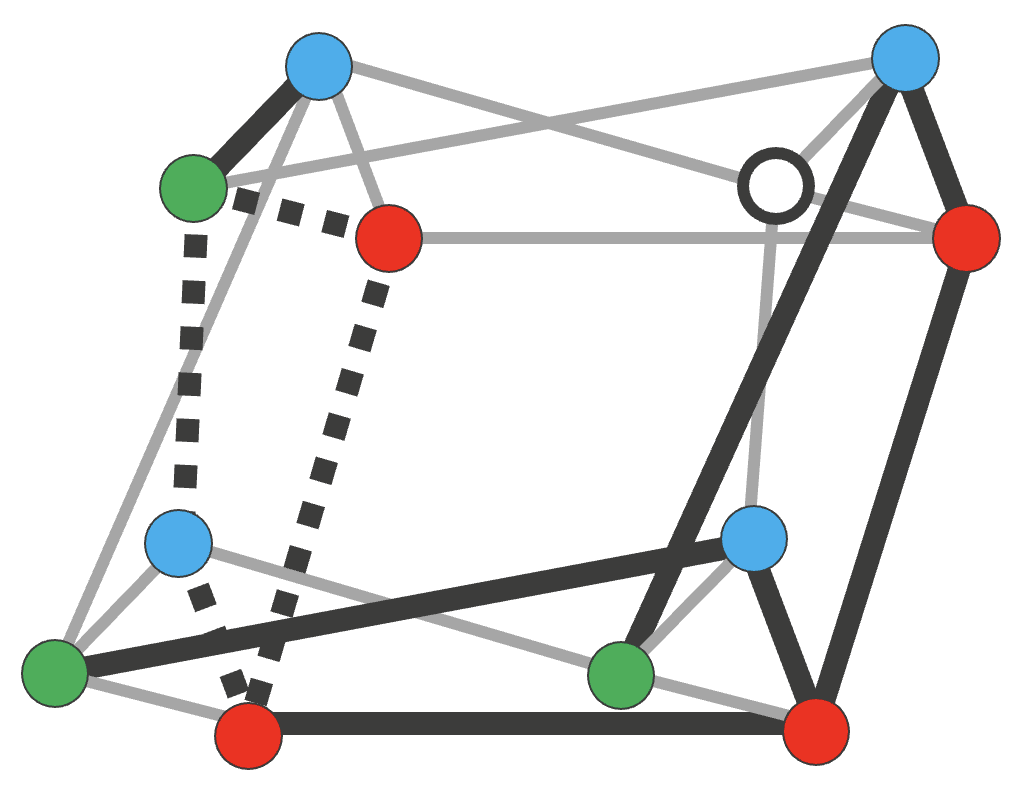}
        \caption{An $f$ transform of tree edges in \ref{fig:const_counterex_cart1}.        }\label{fig:const_counterex_star1}
    \end{subfigure}\hfill
    \begin{subfigure}[t]{.22\linewidth}
        \centering        \includegraphics[width=1.0\linewidth]{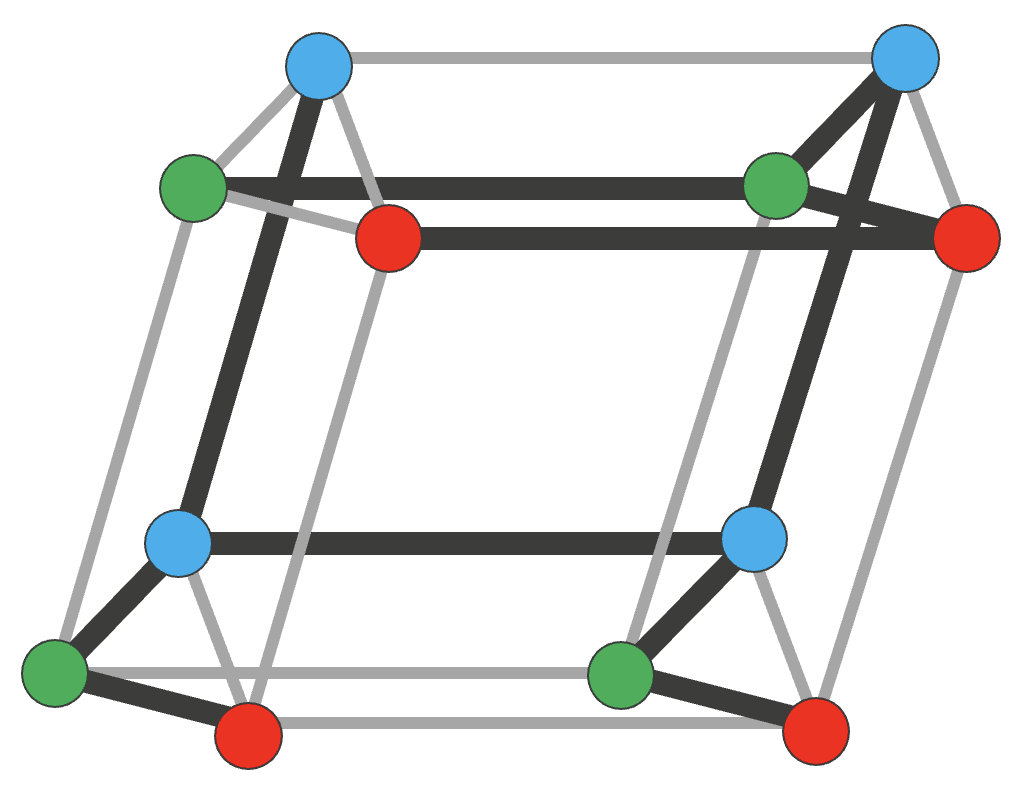}
        \caption{Another   Cartesian spanning tree.
        }\label{fig:const_counterex_cart2}
    \end{subfigure}\hfill
    \begin{subfigure}[t]{.22\linewidth}
        \centering        \includegraphics[width=1.0\linewidth]{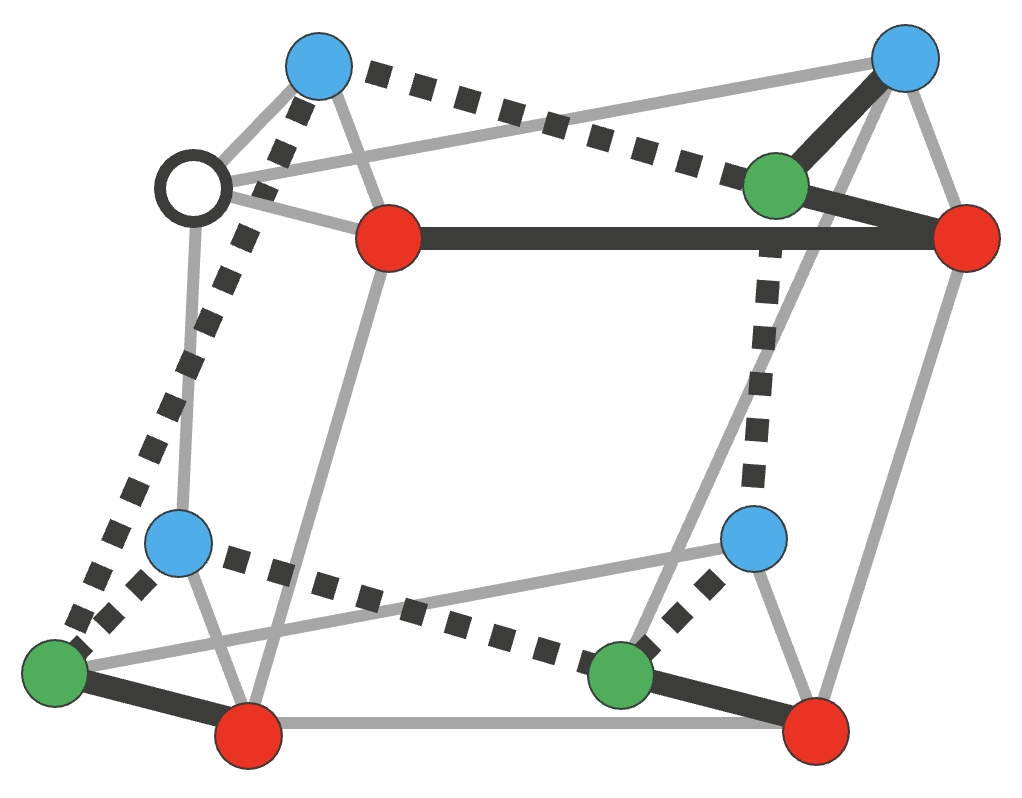}
        \caption{An $f$ transform of tree edges in \ref{fig:const_counterex_cart2}.
       }\label{fig:const_counterex_star2}
    \end{subfigure}
\caption{The $f$-transforms  of Cartesian spanning trees to a star product may not be spanning trees. 
Figures \ref{fig:const_counterex_cart1} and \ref{fig:const_counterex_cart2} show 
the two EDSTs
from \cite{Product_STs_2003} for the Cartesian product in Figure \ref{fig:cart_star_comp_2_cart}. The bijection $f$ from Figure \ref{fig:cart_star_comp_21_star} joins blue and green vertices between supernodes and transforms these spanning trees to get Figures \ref{fig:const_counterex_star1} and \ref{fig:const_counterex_star2}.  These edge sets are not spanning trees: they have dotted-line cycles and exclude white vertices.
}
\label{fig:const_counterex}
\end{figure}

Since the star product generalizes the Cartesian, it is natural to ask if the method for constructing EDSTs in the Cartesian product~\cite{Product_STs_2003} might directly apply to the star product: one might hope to
transform the spanning-tree edges $\{(x,y),(x',y)\}$ in the Cartesian product  
to edges $\{(x,y),(x',f_{(x,x')}(y))\}$ in the star product, to get a set of star-product EDSTs. 
Unfortunately, doing so guarantees neither trees nor spanning edges.
An example of what can go wrong is shown in Figure~\ref{fig:const_counterex}.
Since applying $f$ to the construction from \cite{Product_STs_2003} does not always produce spanning trees, we give a new Cartesian-inspired construction that requires more careful choice of inter-supernode edges. 
We also need extra care in the choice of edges to obtain spanning trees of low-depth. Tree depth was not analyzed in \cite{Product_STs_2003}.

Let $V_s$ and $V_n$ denote the vertices in the structure graph
and supernode, respectively.
In the star product, 
every structure graph vertex is replaced with a supernode instance, giving $\left|V_s\right|$ supernode instances. 
So, there are $\left|V_s\right|$ copies of every supernode EDST in the product graph, one inside each instance.
Also, every edge in the structure graph becomes a set of $\left|V_n\right|$ disjoint edges between neighboring supernodes,  each of which connects a 
distinct vertex in one supernode to a unique vertex in the neighboring supernode. 
Thus, in the product graph, there are $\left|V_n\right|$ disjoint twisted copies of each structure graph EDST.

The $\left|V_n\right|$ copies of one of these structure graph EDSTs form a disconnected spanning forest, each tree connecting one vertex each from all supernodes. 
A copy of any supernode EDST then connects all vertices inside a supernode, effectively connecting all trees in the forest above, producing a spanning tree for the product graph.
One can similarly connect the disconnected spanning forest formed by the copies of a supernode EDST using edges from (perhaps distinct)  copies of a structure graph EDST.
This is the key intuition behind our solutions.

\textbf{Universal Solution:} Let the structure and supernode graphs have $t_s$ and $t_n$ EDSTs, respectively.
We reserve one of the $t_s$ structure graph EDSTs (or $t_n$ supernode EDSTs). 
The copies of each of the remaining $t_s-1$ structure graph EDSTs ($t_n-1$ supernode EDSTs) are connected by the reserved supernode EDST (edges from the copies of reserved structure graph EDST) to construct  EDSTs in the product graph, giving $t_s+t_n-2$ EDSTs. This approach is shown in Figure~\ref{fig:max_solutions_universal} and is formally described in Constructions~\ref{construction:t1_trees} and \ref{construction:t2_trees}.

\textbf{Low-depth Solution:} Consider the $t_n-1$ EDSTs formed by copies of the supernode EDSTs and a reserved structure graph EDST. 
Adjacent supernodes in such trees are joined by a single edge from a copy of the reserved EDST, requiring multiple intra-supernode hops to reach the vertices incident to this edge. 
Our low-depth solution carefully picks the edges from the copies of the reserved EDST so that the total intra-supernode hops in multi-supernode paths are minimized. The remaining $t_s-1$ trees are already low-depth.
This is shown in Figure~\ref{fig:low_depth_solution} with a formal description in Theorem~\ref{thm:depth_universal_2}(2). 

\textbf{Maximum Solution:} The supernode may contain \emph{non-tree} edges not present in any of its $t_n$ EDSTs. Some copies of the reserved supernode EDST or edges from copies of the reserved structure graph EDST are also unused in the $t_s+t_n-2$ EDSTs above. 
Under certain conditions on non-tree edges, with careful selection of edges from reserved EDST copies, the non-tree edges in all supernode instances combined with unused edges from copies of reserved EDSTs give a spanning tree. 
Another EDST is similarly constructed from the structure graph non-tree edges and unused edges from copies of reserved EDSTs, giving a maximum solution.

The star-product constructions here are more complex than those for the Cartesian product in \cite{Product_STs_2003}.
Selection of reserved tree edges for low-depth EDSTs requires explicit alignment with the bijective functions $f_{(x,x')}$ (which impart the twist to structure graph EDST copies) to minimize intra-supernode hops. This is trivial in the Cartesian product where $f_{(x,x')}$ is always the identity. 
The Maximum solution requires the unused edges (not utilized in the first $t_s+t_n-2$ trees) from the reserved structure graph EDST to span a specific set of vertices inside each supernode instance, 
which is also challenging to control with arbitrary bijective connectivity across supernodes.
\section{Upper Bounds on the Number of EDSTs}\label{sec:upper_bounds}
We give here two different bounds on the number of EDSTs that may exist in a star product. Proofs are in the Appendices. We then show that the extended constructions in this paper are nearly or exactly of maximum cardinality.

A spanning tree in a graph $G(V,E)$ with vertex set $V$ and edge set $E$, has $|V|-1$ edges. The number of EDSTs $t$ in a graph must therefore be bounded above by 
\begin{equation}\label{eq:edsts}
    \tau = \floor{\frac{|E|}{|V|-1}},    
\end{equation} 
This is the general combinatorial upper bound on EDSTs on an arbitrary graph. Unfortunately, it may be difficult to simplify.

 The bound in Proposition~\ref{prop:max_bounds_uv} is equivalent to that in Equation~(\ref{eq:edsts}), and is more easily calculated when $|V|$ and $|E|$ have a common factor. (This happens, for example, when at least one of the factor graphs is regular, or more generally, when the $|V_i|$ and $|E_i|$ of one of the factor graphs have a common factor.) All networks discussed in this paper have this property. We have found that in practice, using this proposition has made the calculations for such networks very much easier. 
 
 Proposition~\ref{prop:max_bounds_uv} is used in Table~\ref{table:network_graphs} to show that in almost every network we examine, the maximum number of EDSTs is constructed, even if this is not guaranteed by our theorems. 

\begin{restatable}[]{proposition_restate}{ubsimple}
\label{prop:max_bounds_uv}     
Let $G$ be a simple graph with $|E|$ edges and $|V|>1$ vertices, and let $|E|=m|V|+c,$ where $m$ and $c$ are non-negative integers and $0\le c\le |V|-1$. Then the combinatorial upper bound $\tau$ on the number of  EDSTs in $G$ may be expressed as
$$
    \tau =
    \begin{cases}
        \floor{\frac{\left|E\right|}{\left|V\right|}}=m, &\text{if $m+c< |V|-1$,}\\
        \floor{\frac{\left|E\right|}{\left|V\right|}}+1 = m+1, &\text{otherwise.}
    \end{cases}     
$$
\end{restatable}
\begin{proof}
    The proof is given in Appendix~\ref{sec:appendix_general_star_prod_ths}.
\end{proof}
Corollary~\ref{cor:sptree_identity_regular} simplifies the calculations even more in the special case when the star product itself is regular. This corollary applies to Slim Fly and Bundlefly, but not to PolarStar. 

\begin{restatable}[]{corollary_restate}{ubregular}
\label{cor:sptree_identity_regular}
    Let $G$ be as in Proposition~\ref{prop:max_bounds_uv}. If $G$ is regular with degree $d$ and $|V| \ne 2$, then the combinatorial upper bound $\tau$ on the number of EDSTs in $G$ is $\floor{\frac{d}{2}}$. If $|V| = 2$, then $\tau$ is $\floor{\frac{d}{2}}+1$.
\end{restatable}
The $\sigma$ bound on EDSTs on star products in Proposition~\ref{prop:max_bounds_tr} is less strict than the $\tau$ bound in Equation~(\ref{eq:edsts}) and Proposition~\ref{prop:max_bounds_uv}. However, it establishes a linear relationship between the maximum number of EDSTs in
the product and in the factor graphs. It is used in Theorem~\ref{th:one_less_than_max} to show that our extended constructions are always within $1$ of maximum size, and in Theorem~\ref{th:r_eq_t} and Theorem~\ref{th:t_r_lt} to show maximum size is sometimes attained.
\begin{restatable}[]{proposition_restate}{maxboundstr}
\label{prop:max_bounds_tr}\footnote{All bounds in Proposition~\ref{prop:max_bounds_tr} are tighter than those we presented in \cite{edst_ipdps_2025} by $1$, except for the $\rho_i=\tau_i$ case, which stays the same.}
    Let $G^* = G_s*G_n$ be a star product, with $G^*$ simple. Let $\tau_s$ and $\tau_n$ be the combinatorial upper bounds on the number of EDSTs in $G_s$ and $G_n$, and let $\rho_s$ and $\rho_n$ be the number of unused edges that would be in $G_s$ and $G_n$ if the $\tau_s$ and $\tau_n$ EDSTs were constructed. Then an upper bound $\sigma$ on the number of EDSTs in $G^*$ that may be produced is
    $$
    \sigma = 
        \begin{cases*}
            \tau_s+\tau_n+1, &\text{if $\rho_s \ge \tau_s$ and $\rho_n \ge \tau_n$},\\ 
            \tau_s+\tau_n, &\text{if $\rho_s = \tau_s$ and $\rho_n = \tau_n$}, \\
            \tau_s+\tau_n, &\text{if $\rho_s \ge \tau_s$ or $\rho_n \ge \tau_n$ (but not both)},\\ 
            \tau_s+\tau_n-1, &\text{if $\rho_s < \tau_s$ and $\rho_n < \tau_n$}.
        \end{cases*} 
    $$ 
\end{restatable}
\begin{proof}
    The proof is given in Appendix~\ref{sec:appendix_proofs}.
\end{proof}
\begin{theorem}\label{th:one_less_than_max}\footnote{Theorem~\ref{th:one_less_than_max} shows that the cardinality of our constructions of EDSTs is within $1$ of the upper bound, rather than within $2$, as we presented in \cite{edst_ipdps_2025}.}
Let $G^*$ be a star product whose factor graph EDSTs attain their combinatorial upper bounds $\tau_s$ and $\tau_n$, and let $\mu$ be the maximum possible number of EDSTs in $G^*$. Then: 
\begin{itemize}[itemsep=0pt,parsep=2pt,leftmargin=*]
    \item The extended constructions in Sections~\ref{sec:max_construct} and \ref{sec:different_construction} cover most cases. They always produce at least $\mu-1$ EDSTs, and are guaranteed to produce the maximum $\mu$ when either $\rho_s = \tau_s$ and $\rho_n = \tau_n$, or when $\rho_s < \tau_s$, $\rho_n < \tau_n$ and Property~\ref{property:no_edge_conditions} holds.
    \item The Universal constructions in Section~\ref{subsection:nearly_maximum_set} cover all cases and produce at least $\mu-3$ EDSTs.
\end{itemize} 
\end{theorem}
\begin{proof}
    Let $t$ be the number of EDSTs constructed using the methods in this paper producing the most EDSTs for given conditions on $\rho_s, \rho_n, \tau_s$ and $\tau_n$. By Theorem~\ref{th:r_eq_t}, Corollaries~\ref{cor:r_ge_t} and \ref{cor:t_r_ge_or}, and Theorems~\ref{th:t_r_lt} and \ref{th:t_r_generic}, and as displayed in Table~\ref{tab:star_ths_all}, we construct
    $$
    t = \begin{cases*}
            \tau_s+\tau_n, &\text{if $\rho_s \ge \tau_s$ and $\rho_n \ge \tau_n$},\\ 
            \tau_s+\tau_n, &\text{if $\rho_s = \tau_s$ and $\rho_n = \tau_n$}, \\
            \tau_s+\tau_n-1, &\text{if $\rho_s \ge \tau_s$ or $\rho_n \ge \tau_n$ (but not both)},\\ 
            \tau_s+\tau_n-1, &\text{if $\rho_s < \tau_s$, $\rho_n < \tau_n$,  and Prop.~\ref{property:no_edge_conditions}}\\
            \tau_s+\tau_n-2, &\text{using the Universal construction}.
        \end{cases*} 
    $$

The theorem follows by comparison to the upper bound $\sigma$ described in Proposition~\ref{prop:max_bounds_tr}.
\begin{figure*}[!ht]
\centering
\includegraphics[width=1\textwidth]{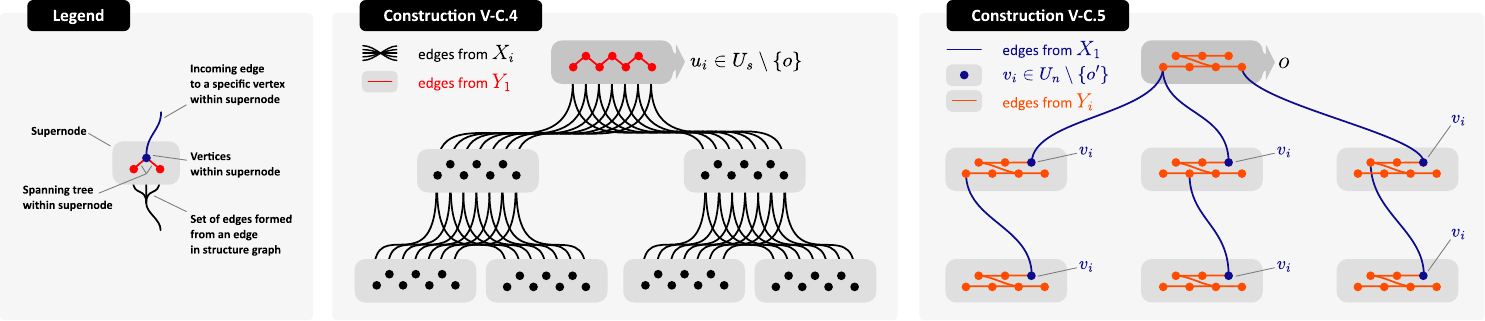}
\caption{The constructions producing $t_s+t_n-2$ of the trees in the Maximum solution when $r_s=t_s$ and $r_n=t_n$. The legend depicts the choices of vertices and edges used to construct these EDSTs coming from the edges in the EDSTS $\mathbf{X} = \{X_1, X_2, \ldots, X_{t_s}\}$ of $G_s$ and $\mathbf{Y} = \{Y_1, Y_2, \ldots, Y_{t_n}\}$ of $G_n$. 
The constructions~\ref{construction:t1_trees} and~\ref{construction:t2_trees} producing the Universal solution of $t_s+t_n-2$ EDSTs can be seen as generalizations of Constructions~\ref{construction:t1_trees_maximum} and \ref{construction:t2_trees_maximum}, where we allow $u_i$ to be an arbitrary vertex in $V_s$ that is unique for each $i$ (corresponds to arbitrary supernode in product graph).
}
\label{fig:max_solutions_universal}
\end{figure*}

\end{proof}

\section{Constructing Star-Product EDSTs}\label{sec:star_edsts}
\subsection{Preliminary Facts for the Constructions}\label{sec:preliminaries}

We consider a star product $G^* = G_s*G_n$ of simple connected graphs $G_s$ and $G_n$ with vertex sets $V_s$ and $V_n$, respectively. Notation used throughout is in Table~\ref{table:notation}.

 We assume there exist sets of EDSTs of maximum size for $G_s$ and $G_n$: $\mathbf{X} = \{X_1, X_2, \ldots, X_{t_s}\}$ and $\mathbf{Y} = \{Y_1, Y_2, \ldots, Y_{t_n}\}$. (The constructions still apply even when $\mathbf{X}$ and $\mathbf{Y}$ are not of maximum size, although in that case, the product graph EDSTs will not be maximum.)

The edges left in $G_s$ and $G_n$ after removing all edges from the trees in $\mathbf{X}$ and $\mathbf{Y}$,
respectively, are called \emph{non-tree} edges. Subgraphs induced by non-tree
edges are called $N_s$ and $N_n$.

$\bar{X}_1$ is a directed version of $X_1$ rooted at $o\in V_s$, so that each
edge goes from parent to child and the parent is the vertex closer to 
$o$. 
$X_1$ is a spanning tree, so for any given $x' \in V_s$, there is at most one edge $(x,x') \in \bar{X}_1$ where $x'$ is the child. 

In the product graph, such an edge becomes a set of $|V_n|$ directed edges with $x'$ as the child supernode. 
Each of these $|V_n|$ edges is incident with a distinct vertex in the $x'$ supernode.
Thus, disjoint edges can be selected
by enforcing uniqueness of incident vertices in the child supernode, as in Construction~\ref{construction:t2_trees}. 
Direction is only used to ensure edge-disjointness; all edges in the final product graph EDSTs are \looseness=-1undirected.
\begin{table}[ht]
\renewcommand\arraystretch{1.35}
\noindent\adjustbox{max width=\columnwidth}{
\begin{tabular}{|p{.8cm}|p{7.2cm}|}
\hline
$t_s/t_n$ &
  the maximum number of spanning trees in $G_s/G_n$ \\ \hline
$\mathbf{X}$&
  a set of EDSTs of maximum cardinality $t_s$ in $G_s$ \\ \hline
$\mathbf{Y}$ &
  a set of EDSTs of maximum cardinality $t_n$ in $G_n$ \\ \hline
$N_s/N_n$ &
  the subgraph of $G_s/G_n$ formed by the non-tree edges \\ \hline
$r_s/r_n$ &
  the number of edges in $N_s/N_n$ \\ \hline
$U_s$ &
  a set of $t_s$ vertices in $N_s$ such that for each $u\in U_s$, there is a path in $N_s$ from $u$ to some $u'\in V_s\setminus U_s$ 
  \\ \hline
$U_n$ &
  a set of $t_n$ vertices in $N_n$ such that for each $v\in U_n$, there is a path in $N_n$ from $v$ to some $v'\in V_n\setminus U_n$ 
  \\ \hline
$o/o'$ &
  a vertex in $V_s/V_n$ (Universal) or $U_s/U_n$ (Maximum) \\ \hline
$X_1$ &
  a spanning tree in $\mathbf{X}$ \\ \hline
$\bar{X}_1$ &
  directed version of $X_1$ rooted at $o\in V_s$ with edges directed away from the root \\ \hline
\end{tabular}
}
\caption{Notation used in the constructions.
\vspace{-.4cm}}
\label{table:notation}
\end{table}

\subsection{A Universal Construction of EDSTs on Star Products}\label{subsection:nearly_maximum_set}

We show here an intuitively constructed set of EDSTs on any star product having connected factor graphs. 
This Universal solution uses the factor graph EDSTs
directly, as graphically illustrated in Figure~\ref{fig:max_solutions_universal}, and produces a total of $t_s+t_n-2$ EDSTs in the star product.

\begin{theorem}
\label{th:t_r_generic}
    Any star product $G^*$ has $t_s+t_n-2$ EDSTs with no conditions on the connected graphs $G_s$ and $G_n$. \end{theorem}

To prove Theorem \ref{th:t_r_generic}, we describe two constructions that collectively produce a set of $t_s+t_n-2$ spanning trees. Then we show that these two sets of trees are edge-disjoint. 

\begin{construction}\label{construction:t1_trees} Construct spanning trees $T_{i|2\leq i\leq t_s}$ using $X_i$ and $Y_1$.
For each $2 \le i \le t_s$, choose a unique $u_i \in V_s$ such that $u_i=u_j$ iff $i=j$. 
 We construct the spanning tree $T_i$ from the union of the following edge sets:
     \vspace{-.06in}
    \begin{align}
  &\left\{ \left\{(u_i, y),(u_i,y')\right\} \} \ | \ \{y, y'\} \in E(Y_1)\right\}\label{eq:Y1_in_ui} \\
  &\left\{ \left\{ (x, y), (x', f_{(x,x')}(y)) \right\} \ | \ \{x, x'\} \in E(X_i), y \in V_n\right\}, \label{eq:Xi_across_all}
    \end{align}
    where $2\leq i\leq t_s$, $u_i\in V_s$ and $u_i=u_j$ if and only if $i=j$.
\end{construction}

Edge set (\ref{eq:Y1_in_ui}) is the copy of supernode EDST $Y_1$ in the supernode corresponding to structure-graph vertex $u_i$. 
Edge set (\ref{eq:Xi_across_all}) is all the copies of structure-graph EDST $X_i$.   
\begin{construction}\label{construction:t2_trees} Construct spanning trees $T_{i|2\leq i\leq t_n}'$ using $X_1$ and $Y_i$. 
For each $2 \le i \le t_n$, choose a unique $v_i \in V_n$ such that $v_i=v_j$ iff $i=j$.
We construct the spanning tree $T_i'$ from the union of the following edge sets:
\begin{align}
    &\left\{\left\{(x,f_{(x,x')}^{-1}(v_i)),(x',v_i)\right\} \ | \ (x,x') \in \bar{X}_1\right\} \ \text{and} \label{eq:one_Xi_edge}\\
 &\left\{ \left\{(x, y), (x, y')\right\}  | \ x\in V_s, \{y, y'\} \in E(Y_i)\right\},\label{eq:Yi_in_all}
\end{align}
where $2\leq i\leq t_n$, $v_i\in V_n$ and $v_i=v_j$ if and only if $i=j$.
\end{construction}
Edge set (\ref{eq:one_Xi_edge}) is the set of edges in copies of the directed structure graph EDST $\bar X_1$ incident with $v_i$ in the child supernodes. 
Edge set (\ref{eq:Yi_in_all}) is the set of the copies of supernode EDST $Y_i$, one per supernode.
\begin{algorithm}
\caption{Constructions~\ref{construction:t1_trees} and \ref{construction:t1_trees_maximum}}
    \label{alg:t1_trees}
    \begin{algorithmic}[1]
        \Statex{Output: spanning tree $T_i$ for $2\leq i \leq t_s$}
        \State{Select a unique $u_i\in V_s$ for Construction~\ref{construction:t1_trees}} 
        \State {Select a unique $u_i\in U_s\setminus \{o\}$ for Construction~\ref{construction:t1_trees_maximum}}
        \State{$\delta(u)\leftarrow $ children of vertex $u$ in $X_i$, rooted at $u_{i}$}
        \State{Initialize a FIFO queue $Q$ with $u_{i}$}
        \ForEach {$(v,v')\in E(Y_1)$}  \Comment{Intra-supernode spanning tree}
            \State{Add edge $\left(\left(u,v\right),\left(u,v'\right)\right)$ to $T_i$}
        \EndForEach

        \ForEach{$(u,u')\in E(X_i)$}
                \ForEach{$v\in V_n$}\Comment{Inter-supernode edges}
                    \State{Add edge $\left(\left(u, v\right), \left(u', f_{(u,u')}(v)\right)\right)$ to $T_i$}
                \EndForEach
        \EndForEach   
    \end{algorithmic}
\end{algorithm}

\begin{proof}[Proof of Theorem \ref{th:t_r_generic}]
First, we show that Construction \ref{construction:t1_trees} produces $t_s-1$ spanning trees in $G^*$.
It is clear that $T_i$ is a spanning tree as the root supernode is connected and the subsequent edges use all copies of $X_i$. Each $T_i$ is disjoint since a unique root supernode $u_i$ and a unique $X_i$ are used.

Next, we show that Construction \ref{construction:t2_trees} produces $t_n-1$ spanning trees in $G^*$. Clearly, the subgraphs $T_i'$ are spanning trees: each supernode is connected internally via the spanning tree $Y_i$ and adjacent supernodes in $\bar{X_1}$ are connected by a single edge in $T_i'$. The latter is always incident with vertex $v_i$ in the child supernode for that edge. Thus, the trees are edge-disjoint since $v_i$ and $Y_i$ are chosen to be distinct for each $i$.
  
Finally, we show that the two constructions are edge-disjoint with respect to each other. 
Within supernodes, the $T_i$ trees use edges from $Y_1$ and then $T_i'$ trees use edges from $Y_i$ ($i\ne 1)$. Thus these are disjoint. Between supernodes, the $T_i$ trees use edges from $X_i$ ($i \ne 1)$ and the $T_i'$ trees use edges from $X_1$. Thus the constructions must be edge-disjoint. 
\end{proof}

Note that when $t_s=t_n=1$,  the Universal construction builds no spanning trees, although clearly the star product must have at least one. The constructions from Sections~\ref{sec:max_construct} and \ref{sec:different_construction} apply to this special case, and build one or two trees. 
\begin{algorithm}
    \caption{Constructions~\ref{construction:t2_trees} and \ref{construction:t2_trees_maximum}}
    \label{alg:t2_trees}
    \begin{algorithmic}[1]
        \Statex{Output: spanning tree $T'_i$ for $2\leq i \leq t_n$}
        \State{Select a unique $v_i\in V_n$ for Construction~\ref{construction:t2_trees}} 
        \State {Select a unique $v_i\in U_n\setminus \{o'\}$ for Construction~\ref{construction:t2_trees_maximum}}
   \State{$\delta(u)\leftarrow $ children of vertex $u$ in $X_1$, rooted at $o$} 
        \State{Initialize a FIFO queue $Q$ with $o$}
        
        \While{$Q$ is not empty}
            \State {$u\leftarrow $ pop($Q$)}
            \ForEach{$(v',v'')\in E(Y_i)$}\Comment{Intra-supernode trees}
                \State{Add edge $\left(\left(u, v'\right), \left(u, v''\right)\right)$ to $T_i'$}       
            \EndForEach
            \ForEach{$u'\in \delta(u)$}\Comment{Inter-supernode edges}
                \State{Add edge $\left(\left(u, f_{(u,u')}^{-1}(v_i)\right), \left(u', v_i\right)\right)$ to $T_i'$}
                \State{Push $u'$ in $Q$}
            \EndForEach
        \EndWhile
    \end{algorithmic}
\end{algorithm}
\subsubsection{Depth of the EDSTs of the Universal Solution}
The depth of the constructed trees are in terms of depths $d_{s_i}$ and $d_{n_i}$ of tree $\bar{X}_i$ in $G_s$ and tree $Y_i$ in $G_n$, where centers are the roots.

In both constructions, the depth of the EDSTs in the star product depends on the depths of the EDSTs in the factor graphs. If the factor graphs are chosen carefully to have low-depth trees, then the star product will also have a low depth. 

If we are not interested in constructing more than $t_s + t_n-2$ trees, Construction \ref{construction:t2_trees} can be modified to give trees with depth roughly the sum of the factor graphs' depths rather than the product, as discussed in Theorem~\ref{thm:depth_universal_2}.

\begin{theorem} \label{thm:depth_universal_1}
    The trees in Construction \ref{construction:t1_trees} have depth at most $2d_{s_{i}} + d_{n_{1}}$ for each $i = 2, \ldots, t_s$.
\end{theorem}

\begin{proof}
Let $v_1$ be the root of $Y_1$ and consider the distance from vertex $(u_i, v_1)$ to an arbitrary vertex $(u, v)$ in $G^*$. 
Each of the $\left|V_n\right|$ copies of $X_i$ in $T_i$ is a tree containing exactly one vertex from every supernode in $G^*$. 
Let $(u_i, v')$ be the vertex in supernode $u_i$ contained in the same copy of $X_i$ as $(u, v)$.  From root $(u_i,v_1)$, we first traverse the copy of $Y_1$ within supernode $u_i$ to reach $(u_i, v')$, using at most $d_{n_1}$ hops. Then from $(u_i,v')$, we traverse the copy of $X_i$ to $(u,v)$. This takes at most $2d_{s_i}$ hops since $u_i$ may not be the root of $X_i$. 
\end{proof}
Construction \ref{construction:t2_trees} allows for an arbitrary distinct edge from the copies of $X_1$ to be used between any adjacent supernodes in each EDST $T_i'$. 
However, the depth can be optimized with a careful selection of edges, as in Theorem~\ref{thm:depth_universal_2}.

\begin{theorem}\label{thm:depth_universal_2}\leavevmode
\begin{enumerate}
\item The trees in Construction \ref{construction:t2_trees} have depth at most $2d_{s_1}d_{n_i}+d_{n_i} + d_{s_1}$ for each $i = 2, \ldots, t_n$. 

\item A low-depth modification of Construction \ref{construction:t2_trees}: For each $i = 2, \ldots, t_n$, choose a unique $v_i \in V_n$. 
Let $x_1, x_2, \ldots, x_{d_{k}}$ be a path in $\bar{X}_1$ starting from the root of $\bar{X}_1$ defined as $o = x_1$. Define inter-supernode edges using in Construction \ref{construction:t2_trees} to be the edges between between vertices $(o, v_i), (x_2, f_{(o,x_2)}(v_i)),$  $(x_3, f_{(x_2, x_3)}\circ f_{(o,x_2)}(v_i)),$ 
$\ldots, (x_{k}, f_{(x_{k-1}, x_{k})} \circ \cdots \circ f_{(o,x_2)}(v_i)),$ for every path from the root in $\bar{X}_{1}.$ Then the depth for each of the trees in Construction \ref{construction:t2_trees} is $d_{s_1} + 2d_{n_i}$.
\end{enumerate}
\end{theorem}

\begin{figure}[!ht]
\centering
\includegraphics[width=1\linewidth]{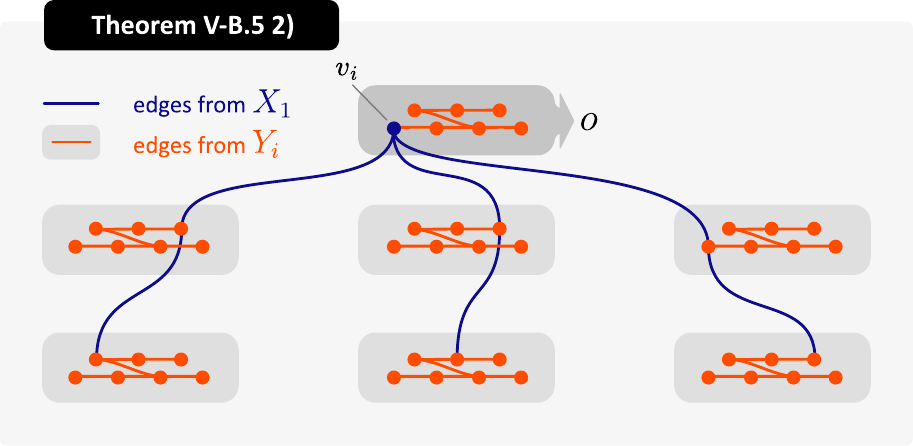}
\caption{To construct low-depth EDSTs, we modify the Universal Construction~\ref{construction:t2_trees} which is a generalization of Construction~\ref{construction:t2_trees_maximum} shown in Figure~\ref{fig:max_solutions_universal}. 
Instead of using edges from different copies of $\bar{X}_1$,  
we trace out a single copy of $\bar{X}_1$.  
Inter-supernode traversal is done along this copy and thus, avoids hopping within intermediate supernodes.
\vspace*{-.5cm}
}
\label{fig:low_depth_solution}
\end{figure}

\begin{proof}
First we prove (1). Let $o$ be the root of $\bar{X}_1$  and $o'$ be the root of $Y_i$ and consider the path from $(o, o')$ to any other vertex in $G^*$. This path goes through
at most $d_{s_1}$ distinct supernodes (because $o$ is root of $\bar{X}_1$) and requires as many inter-supernode hops. Within each supernode,
we may require intra-supernode hops to reach the vertex adjacent to the next supernode in the path. Specifically, 
we may traverse at most $d_{n_i}$ hops in supernode $o$ (because $o'$ is the root of $Y_i$) and at most $2d_{n_i}$ hops in the subsequent $d_{s_1}$ supernodes (because the vertices
adjacent to other supernodes may not be the root of $Y_i$). Thus, the overall path length is at most $d_{s_1}+d_{n_i}+2d_{s_1}d_{n_i}$.

To minimize travel within supernodes, we carefully choose the edge connecting each supernode in (2). Consider the distance from $(o, v_i)$ to any other vertex in $G^*$. By taking the path described in (2), we omit traversal within all intermediate supernodes and reach the  supernode of the destination vertex within $d_{s_1}$ inter-supernode hops from $(o, v_i)$. Within the final supernode, we may need at most $2d_{n_i}$ hops to reach the destination vertex, resulting in a distance of at most $d_{s_1}+2d_{n_i}$.
\end{proof}
\vspace*{-.3cm}
\begin{algorithm}
    \caption{Low-depth construction from Theorem~\ref{thm:depth_universal_2}.2}
    \label{alg:t2_trees_low_depth}
    \begin{algorithmic}[1]
        \Statex{Output: spanning tree $T'_i$ for $2\leq i \leq t_n$}
        \State{$u_{min}\leftarrow$ any center of tree $X_1$}
        \State{Select a unique $v_i\in V_n$}
        \State{Initialize a FIFO queue $Q$ with $\left(u_{min}, v_i\right)$}
        \State{$\delta(u)\leftarrow $ children of vertex $u$ in $X_1$, rooted at $u_{min}$}
     
        \While{$Q$ is not empty}
            \State {$(u, v)\leftarrow $ pop($Q$)}
            \ForEach{$(v',v'')\in E(Y_i)$}\Comment{Intra-supernode trees}
                \State{Add edge $\left(\left(u, v'\right), \left(u, v''\right)\right)$ to $T_i'$}         
            \EndForEach
            \ForEach{$u'\in \delta(u)$}\Comment{Inter-supernode edges}
                \State{Add edge $\left(\left(u, v\right), \left(u', f_{(u,u')}(v)\right)\right)$ to $T_i'$}
                \State{Push $\left(u', f_{(u,u')}(v)\right)$ in $Q$}
            \EndForEach
        \EndWhile        
    \end{algorithmic}
\end{algorithm}

\subsection{Additional EDSTs When \texorpdfstring{$r_s \ge t_s$}{} or \texorpdfstring{$r_n \ge t_n$}{}}\label{sec:max_construct}
\begin{figure*}[!ht]
\centering
\includegraphics[width=1\textwidth]{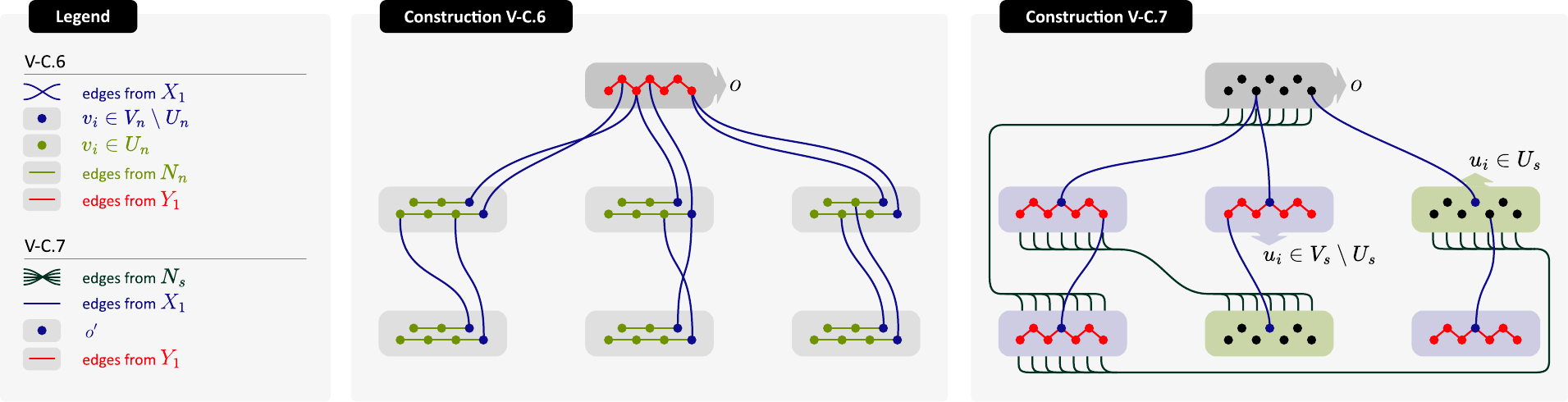}
\caption{The constructions producing the additional two trees in the Maximum solution. The legend depicts the choices of vertices and edges used to construct these EDSTs. As in Figure~\ref{fig:max_solutions_universal}, the gray boxes represent supernodes. Supernode $o$ is distinct from the supernodes where $Y_1$ is instantiated in Construction \ref{construction:t1_trees_maximum}~(Figure~\ref{fig:max_solutions_universal}). 
\vspace{-.18in}}
\label{fig:max_solutions_additional}
\end{figure*}
In this section we construct additional trees, under certain conditions, beyond the Universal construction~(Section~\ref{subsection:nearly_maximum_set}). 

This construction will be of maximum size when the number of non-tree edges in factor graphs are exactly the number of their EDSTs. In that case, we get two more trees. Some star products do in fact satisfy this constraint, as shown in Table~\ref{table:network_graphs}.

We then get two corollaries from this theorem which address construction of additional trees when non-tree edges in factor graphs are more than their EDSTs. 
In this scenario, while maximum size is not guaranteed in general, our comparison of the number of EDSTs obtained to the upper bound from Proposition~\ref{prop:max_bounds_uv} 
shows that we do in fact achieve maximum cardinality for almost all star product networks discussed in this paper.
A detailed comparison is presented in Table~\ref{table:network_graphs}.

\subsubsection{Preliminary Lemmas for our Maximum Solution}\label{sec:additional_trees}

Lemma~\ref{lemma:ngConnect} from \cite{Product_STs_2003} shows that factor-graph EDSTs may be chosen to satisfy a needed condition for our construction. 

\begin{lemma}\label{lemma:ngConnect}\cite{Product_STs_2003}
    Let $G=(V,E)$ be a graph with a maximum number $t$ of EDSTs and $r$ non-tree edges. 
    There exists a set $\mathbf{X}$ of $t$ EDSTs in $G$ for which the subgraph of $N_G$ induced by non-tree edges satisfies the following: there 
    exists a subset $U \subseteq V$ of $r$ vertices such that every vertex in $U$ has a path in $N_G$ to a vertex in $V \setminus U$.
\end{lemma}

Lemma~\ref{lemma:ngConnect} means that we can pick subsets of vertices $U_s$ so that $N_s$ has paths from any vertex within $U_s$ to a vertex outside of $U_s$. 
This will be vital to our construction: if we can join supernodes $V_s\setminus U_s$ 
to a connected component, we can use edges formed from $N_s$  to join all vertices in $U_s$ to the same connected component. 
Thus, edges incident with $U_s$ in the copies of $X_i$ can be used to construct the first $t_1 + t_2 - 2$ trees (as in the Universal solution), 
and $N_s$ can be used to connect $U_s$ in the additional tree. Similarly, non-tree edges in 
$G_n$ can also be used to obtain another additional tree.

Not every set of maximum size of EDSTs in $G_s$ and $G_n$ will allow such special subsets. However, we can construct the required EDSTs satisfying Lemma~\ref{lemma:ngConnect}, starting from any maximum-size set of EDSTs, as shown in \cite{Product_STs_2003}. For brevity, 
we refer the reader to \cite{Product_STs_2003}
for this construction.

Lemma~\ref{lemma:t_r_eq} is a guarantee of maximum cardinality when $r\le t$. This lemma permits us to make use of Lemma~\ref{lemma:ngConnect} to construct the EDSTs for Theorem~\ref{th:r_eq_t}.

\begin{lemma}
    \label{lemma:t_r_eq}
   Let $G$ be a simple graph with $|E|$ edges and $|V| > 1$ vertices. Let $t$ be the number of EDSTs in some set of EDSTs, with $r$ the number of remaining non-tree edges. Then
   $t$ is the maximum number of EDSTs when $r\le t$.  
\end{lemma}
\begin{proof}
    This is trivial for $|V|\le 2$.  Let $|V|>2$. Then $t \le \floor{\frac{|E|}{|V|-1}} \le \floor{\frac{|V|}{2}} < |V|-1$, since $G$ is simple. So $r\le t<|V|-1$, and $t$ must be maximum.
\end{proof}

\subsubsection{A Maximum Construction When \texorpdfstring{$r_s = t_s$}{} and \texorpdfstring{$r_n = t_n$}{}}
\begin{theorem}
\label{th:r_eq_t}
    If $r_s = t_s$ and $r_n = t_n$ in connected graphs $G_s$ and $G_n$, then we can construct a set of $t_s +t_n$ EDSTs in $G_s*G_n$ and this set is of maximum size. 
\end{theorem}

Any construction of $t_s+t_n$ EDSTs will be of maximum size, by Proposition~\ref{prop:max_bounds_tr}. We show here such a construction.

By Lemma~\ref{lemma:t_r_eq}, $r_s=t_s$ and $r_n=t_n$ implies that
$t_s$ and $t_n$ are of maximum cardinality in $G_s$ and $G_n$, respectively, so we can invoke Lemma~\ref{lemma:ngConnect} in the following constructions.

\begin{construction} \label{construction:t1_trees_maximum}

This is a special case of Construction \ref{construction:t1_trees}, where we carefully choose the supernodes $u_i$ to be in $U_{s} \setminus \{o\}$ so that $u_i = u_j$ iff $i=j$.
\end{construction}

\begin{construction} \label{construction:t2_trees_maximum}
This is a special case of Construction \ref{construction:t2_trees}, where we carefully select the incident vertices $v_i$ in child supernodes within $U_{n} \setminus \{o'\}$ so that $v_i = v_j$ iff $i=j$.
\end{construction}

For the remaining two trees, Construction \ref{construction:extra_tree1} requires $|E(N_n)|\ge t_n$ and Construction \ref{construction:extra_tree2} requires  $|E(N_s)| \ge t_s$. 

\begin{construction} \label{construction:extra_tree1} Construction via $X_1$, $Y_1$, $o$, and $N_n$ 

	   \emph{Intuition:} 
     In the previous constructions, we have used $Y_1$ in all supernodes in $U_s \setminus \{o\}$ and have used all edges formed from $\bar{X}_1$ incident with a vertex in set $U_n \setminus \{o'\}$ in the child supernode. 
 This construction will use the instance of $Y_1$ in supernode $o$, instances of $N_n$ in all other supernodes, and all edges formed from $\bar{X}_1$ incident with the set $V_n\setminus{U_n}$ in child supernodes. 
 Thus, we will ensure edge disjointness from previous constructions. Lemma \ref{lemma:ngConnect} and our careful selection of inter-supernode edges will ensure that this construction spans $G^*$. This construction is shown in Figure \ref{fig:max_solutions_additional}.

    \emph{Formal Construction:} We construct the subgraph of $G^*$ formed by the union of the following edge sets:
    \vspace{-.05in}
    \begin{flalign}
    &\left\{\left\{(o, y), (o, y')\right\} \ | \ \{y, y'\} \in E(Y_1)\right\} \ 
    \label{eq:Y1_in_o}&&\\
    &\left\{\left\{(x, n), (x, n')\right\} \ | \ \{n, n'\} \in E(N_n), x \in V_s \setminus \{o\}\right\} \  \label{eq:NH_in_most}&&\\
	&\{\{(x,f_{(x,x')}^{-1}(v)),(x',v)\} \ | \ (x,x') \in E(\bar{X}_1), v \in V_n \setminus U_n\}.\label{eq:X1_in_VH_minus_UH}&&
    \end{flalign}
The edges between supernodes given in \eqref{eq:X1_in_VH_minus_UH} ensure both that the construction spans $G^*$ edge-disjointness from the other trees. As shown in Figure \ref{fig:const_counterex}, applying the Cartesian \cite{Product_STs_2003} solution directly cannot guarantee these properties.
	\end{construction}

\begin{algorithm}
    \caption{Construction~\ref{construction:extra_tree1}}
    \label{alg:extra_tree1}
    \begin{algorithmic}[1]
        \Statex{Output: spanning tree $T_{e_1}$}
        
        \State{Initialize empty subgraph $S_{e_1}$}
        \State{$\delta(u)\leftarrow $ children of vertex $u$ in $X_1$, rooted at $o$}

        \State{Initialize a FIFO queue $Q$ with $o\in V_s$}
        \ForEach{$(v,v')\in Y_1$}\Comment{Intra-supernode spanning tree}
            \State{Add edge $\left(\left(o,v\right), \left(o, v'\right)\right)$ to $S_{e_1}$}
        \EndForEach
            
        \While{$Q$ is not empty}
            \State {$u\leftarrow $ pop($Q$)}
            \ForEach{$u'\in \delta(u)$}
                \ForEach{$(v,v')\in N_n$}\Comment{Intra-supernode non-tree edges}
                    \State{Add edge $\left(\left(u, v\right), \left(u, v'\right)\right)$ to $S_{e_1}$}
                \EndForEach
                \ForEach{$v\in V_n\setminus U_n$}\Comment{Inter-supernode edges}
                    \State{Add edge $\left(\left(u, f_{(u, u')}^{-1}(v)\right), (u', v)\right)$ to $S_{e_1}$}
                \EndForEach
            \EndForEach
        \EndWhile

        \State{$T_{e_1}\leftarrow$ any spanning tree of $S_{e_1}$}
        
    \end{algorithmic}
\end{algorithm}

\begin{construction}\label{construction:extra_tree2} Construction via $X_1$, $Y_1$, $o'$, and $N_s$  

 \emph{Intuition:} 
  In the previous constructions, we have used $Y_1$ in all supernodes in $U_s$ and have used all edges formed from $\bar{X}_1$ incident with a vertex in set $V_n \setminus \{o'\}$ in the child supernodes. 
  For this last construction, we use an instance of $Y_1$ in all supernodes $u \in V_s \setminus U_s$. 
 To connect supernodes, we will use all edges formed from $N_s$ and remaining edges from $\bar{X}_1$ that are incident with vertex $o'$ in child supernodes. 

 All vertices in supernodes without an instance of $Y_1$ will be connected to supernodes in $V_s\setminus U_s$ via edges formed from $N_s$. The edges formed from $\bar{X}_1$ connect all supernodes that contain $Y_1$, which
 ensures that we produce a connected spanning subgraph. 
 This is depicted in  Figure~\ref{fig:max_solutions_additional}.
 
 \emph{Formal Construction:} We build this subgraph of $G^*$ using a union of the following edge sets:
 \begin{flalign}
    &\left\{\left\{(x, y), (x, y')\right\} \ | \ \{y, y'\} \in E(Y_1), x\in V_s \setminus U_s\right\} \ \label{eq:Y1_in_VG_minus_UG}&&\\
    &\left\{\left\{(n, y), (n', f_{(n,n')}(y))\right\} \ | \ (n,n') \in E(N_s), y \in V_n\right\}
    \label{eq:NG_edges}\hspace{-.1cm}&&\\ 
     &\left\{\left\{(x,f_{(x,x')}^{-1}(o')),(x',o')\right\} \ | \ (x,x') \in E(\bar{X}_1)\right\}.\label{eq:X1_to_o_prime}&&
 \end{flalign}
 
In \cite{Product_STs_2003}, their last construction does not use all copies of $N_s$ to connect the supernodes. However, to ensure the subgraph is spanning, we must use all copies. As in our other constructions, we choose disjoint edges from different copies of $\bar{X}_1$ to connect the supernodes rather than a single instance.
	\end{construction}

\begin{algorithm}
    \caption{Construction~\ref{construction:extra_tree2}}
    \label{alg:extra_tree2}
    \begin{algorithmic}[1]
        \Statex{Output: spanning tree $T_{e_2}$}
        
        \State{Initialize a FIFO queue $Q$ with $o\in V_s$}
        \State{Initialize empty subgraph $S_{e_2}$}\State{$\delta(u)\leftarrow $ children of vertex $u$ in $X_1$, rooted at $o$}

        \ForEach{$u\in V_s\setminus U_s$}\Comment{Intra-supernode spanning trees}
            \ForEach{$(v,v')\in E(Y_1)$}
                \State{Add edge $((u, v), (u, v'))$ to $S_{e_2}$}
            \EndForEach
        \EndForEach

        \ForEach{$(u, u')\in N_s$}\Comment{Inter-supernode non-tree edges}
            \ForEach{$v\in V_n$}
                \State{Add edge $\left((u, v), \left(u', f_{(u,u')}(v)\right)\right)$ to $S_{e_2}$}
            \EndForEach
        \EndForEach
            
        \While{$Q$ is not empty}
            \State {$u\leftarrow $ pop($Q$)}
            \ForEach{$u'\in \delta(u)$} \Comment{Inter-supernode edges}
                \State{Add edge $\left(\left(u, o'\right), \left(u', f_{(u,u')}(o')\right)\right)$ to $S_{e_2}$}
                \State{Push $u'$ in $Q$}
            \EndForEach
        \EndWhile

        \State{$T_{e_2}\leftarrow$ any spanning tree of $S_{e_2}$}
        
    \end{algorithmic}
\end{algorithm}
    
 \begin{remark} \label{rmk:cycles_okay} Constructions~\ref{construction:extra_tree1} and \ref{construction:extra_tree2} may not be acyclic. However, a breadth-first search on the subgraphs can give a spanning tree if these subgraphs are connected and spanning in $G^*$. 
\end{remark}

 Constructions~\ref{construction:t1_trees_maximum} and \ref{construction:t2_trees_maximum} clearly form spanning trees in $G^*$, as shown in the proof for 
 Theorem~\ref{th:t_r_generic}.
 Next, we formally prove that additional Constructions~\ref{construction:extra_tree1} and \ref{construction:extra_tree2} also form spanning subgraphs of $G^*$.

\begin{lemma}\label{lemma:additional_trees}
    Constructions \ref{construction:extra_tree1} and \ref{construction:extra_tree2} produce connected subgraphs of $G^*$ with the same vertex set as $G^*$. Moreover, these two graphs are edge-disjoint.
\end{lemma}
\begin{proof}
First, consider Construction \ref{construction:extra_tree1}. 
From the edges in \eqref{eq:Y1_in_o}, all vertices in supernode $o$ are connected by spanning tree $Y_1$. Now consider a neighbor $x'$ of $o$ in $\bar{X}_1$ and pick some $y \in V_n$. We will show that $(x', y)$ is connected to supernode $o$.
If $y \in V_n \setminus U_n$, then there exists a direct edge from the set $\eqref{eq:X1_in_VH_minus_UH}$ connecting supernode $o$ to $(x', y)$. If $y \in U_n$, by Lemma~\ref{lemma:ngConnect}, there exists a path within graph $N_n$ from $y$ to $w \in V_n \setminus U_n$. Since $w \in V_n \setminus U_n$, there exists a direct edge from supernode $o$ to $(x', w)$. We then concatenate these paths. The claim follows from induction on the distance from $o$.

In Construction \ref{construction:extra_tree2}, we show that all supernodes are internally connected. 
Clearly, supernodes in $V_s\setminus U_s$ are internally connected because of the $Y_1$ instance within them.
Consider a supernode $u\in U_s$. It is connected to some supernode $u'\in V_s\setminus U_s$ via a path in $N_s$. 
 Since $u'\in V_s\setminus U_s$, it is connected by $Y_1$ internally. Hence, vertices in supernode $u$ are connected internally. 
Since all supernodes are internally connected and all supernodes are externally connected by edges from the copies of $\bar{X}_1$, the entire subgraph is connected.

Finally, the two graphs are edge-disjoint: this reduces to showing i) \eqref{eq:Y1_in_o} and \eqref{eq:Y1_in_VG_minus_UG} are edge-disjoint and ii) \eqref{eq:X1_in_VH_minus_UH} and \eqref{eq:X1_to_o_prime} are edge-disjoint. Since $o \in U_s$ and $o'\in U_n$, this is clear.
\end{proof}

\begin{proof}[Proof of Theorem \ref{th:r_eq_t}]
It remains to show that when $N_s$ and $N_n$ have at least $t_s$ and $t_n$ edges respectively, then we can produce $t_s+t_n$ EDSTs on $G^*$.

Constructions \ref{construction:t1_trees_maximum} and \ref{construction:t2_trees_maximum} produce $t_s-1$ and $t_n-1$ EDSTs respectively. 
From Lemma~\ref{lemma:ngConnect}, if $r_s \ge t_s$ and $r_n \ge t_n$, then $|U_s|\geq t_s$ and $|U_n|\geq t_n$. Hence, there are sufficient
choices for unique supernodes $u_i\in U_s\setminus\{o\}$ in Construction \ref{construction:t1_trees_maximum}, and sufficient choices for unique vertices $v_i\in U_n\setminus\{o'\}$ in Construction \ref{construction:t2_trees_maximum}.

From Remark \ref{rmk:cycles_okay}, we can produce spanning trees of $G^*$ from Constructions \ref{construction:extra_tree1} and \ref{construction:extra_tree2}. It remains to show that the additional two trees share no edges with the trees from Constructions \ref{construction:t1_trees_maximum} or \ref{construction:t2_trees_maximum}. A comparison of the edge sets across these constructions completes the argument.
\end{proof}

\subsubsection{Additional EDSTs When \texorpdfstring{$r_s \ge t_s$}{} or \texorpdfstring{$r_n \ge t_n$}{}}

We complete our constructions by stating two corollaries of Theorem~\ref{th:r_eq_t}, addressing the cases where one or both of the $r_i$ are greater than the $t_i$. These are more general, but for these cases, there is no guarantee of maximum cardinality.
\begin{corollary}
\label{cor:r_ge_t}
    If $r_s \ge t_s$ and $r_n \ge t_n$ in connected graphs $G_s$ and $G_n$, we can construct $t_s +t_n$ EDSTs in $G_s*G_n$.  \end{corollary}
\begin{proof}
    This follows from Theorem~\ref{th:r_eq_t} by using $t_s$ of the $r_s$ remaining edges and $t_n$ of the $r_n$ remaining edges. 
\end{proof}

\begin{corollary}
\label{cor:t_r_ge_or}
    If $r_s \ge t_s$ or $r_n \ge t_n$ in connected graphs $G_s$ and $G_n$, we can construct $t_s+t_n-1$ EDSTs in $G_s*G_n$.     \end{corollary}
\begin{proof}
    In the proof of Theorem~\ref{th:r_eq_t}, Construction \ref{construction:extra_tree1} only requires that $r_s \ge t_s$, with no constraints on $r_n$, and Construction \ref{construction:extra_tree2} only requires that $r_n \ge t_n$, with no requirements for $r_s$. Thus if only one of $r_s \ge t_s$ or $r_n \ge t_n$, we can obtain $t_s+t_n-1$ EDSTs.
    \end{proof}

\subsubsection{Depth of the EDSTs of this Maximum Solution}
The construction of the extra one or two trees comes at cost of quadratic worst-case depth (in terms of depth of the factor trees), in contrast to the linear depth seen in the Universal \looseness=-1case.

\begin{theorem}\label{thm:depth_maximum}
   Let $d_{s_i}$ be the depth of the tree $X_i$ and $d_{n_i}$ be the depth of $Y_i$. Let $m_s$ be the largest path in $N_s$ and $m_n$ the largest path in $N_n$. The depths of the trees in the Maximum construction are at most:
    \begin{enumerate}
[itemsep=0pt,parsep=2pt,leftmargin=*]
        \item Construction \ref{construction:t1_trees_maximum}: $2d_{s_i} + d_{n_1}$, $i \in \{2, \ldots, t_s\}$
        \item Construction \ref{construction:t2_trees_maximum}: $2d_{s_1}d_{n_i} + d_{n_i} + d_{s_1}$, $i \in \{2, \ldots, t_n\}$
        \item Construction \ref{construction:extra_tree1}: $d_{n_1} + d_{s_1}(1+m_n)$
        \item Construction \ref{construction:extra_tree2}: $(d_{s_1}+1)(2d_{n_1} + 2m_s) + d_{s_1}$.
    \end{enumerate}
\end{theorem}

The depth can also be calculated in terms of $d_{s_i}, d_{n_i}$, $r_s$, and $r_n$ since $m_n \le r_n$ and $m_s \le r_s$. 
\begin{corollary}\label{cor:depth_maximum}\leavevmode Let $d$ be tree depth in a construction.
    \begin{itemize}
[itemsep=0pt,parsep=2pt,leftmargin=*] 
    \item In Construction \ref{construction:extra_tree1}, $d \le d_{n_1} + d_{s_1} (1+r_n).$
    \item In Construction \ref{construction:extra_tree2}, $d \le (d_{s_1}+1)(2d_{n_1} + 2r_s) + d_{s_1}.$
    \end{itemize}
\end{corollary}
\begin{proof}[Proof of Theorem \ref{thm:depth_maximum}]
    The depths in Constructions \ref{construction:t1_trees_maximum} and \ref{construction:t2_trees_maximum} follow from Theorems \ref{thm:depth_universal_1} and \ref{thm:depth_universal_2} respectively. Note that we cannot use the optimized version of Construction \ref{construction:t2_trees} since we need to carefully choose edges to allow for the additional trees.

    Next, we prove the depth of Construction \ref{construction:extra_tree1}. Let $o$ be the root of $\bar{X}_1$ and $o'$ be the root of $Y_1$. 
Since supernode $o$ is connected, any vertex within this supernode is at most $d_{n_1}$ hops away from $(o,o')$. To traverse to an adjacent supernode $x$, we first traverse to the correct vertex in supernode $o$, then traverse down using the edge between $o$ and $x$.

     Since $N_n$ may be disconnected, we need to ensure that this vertex is chosen so that we reach the correct component $C$ of $N_n$. Once we hop to supernode $x$, we take at most $|E(C)|$ hops to $(x,y)$. Since $|E(C)| \leq m_n$, the maximum number of hops taken within a supernode visited is $m_n.$ Since we visit $d_{s_{1}}$ supernodes and use one edge per supernode, the depth is at most $d_{n_{1}} + d_{s_{1}}(1 + m_{n})$.

Finally, we prove the depth in Construction \ref{construction:extra_tree2}. 
    Start from $(o,o')$. In order to reach another vertex $(o,y)$ for some $y \in V_n$, we must first traverse via $N_s$ edges to a supernode $u\in V_s \setminus U_s.$ This is possible by Lemma~\ref{lemma:ngConnect}, and requires at most $m_s$ hops. From here, we can move within supernode $u_i$ to any vertex within the supernode via $Y_1$. Finally, we traverse back through the $N_s$ path to supernode $o$. This requires at most $2d_{n_1} +2m_{s}$ hops.
    We then move down the copy of $\bar{X}_1$ to a vertex in another supernode $u$. In order to connect this vertex to $(u,y)$ for any $y\in V_n$, we may need to use the path described above again. Since there are at most $d_{s_1}$ edges from root supernode $o$ to any leaf, the depth in our construction is at worst $(d_{s_1}+1)(2d_{n_1}+2m_s) + d_{s_1}.$  \end{proof}

\subsection{A Maximum Construction For \texorpdfstring{$r_s<t_s$}{} and \texorpdfstring{$r_n<t_n$}{}
}\label{sec:different_construction}
We discuss here an alternative general construction that applies for all graphs having a certain Property. This construction builds a maximum number of EDSTs when $r_s<t_s$ and $r_n<t_n$. 

The Universal construction (Theorem \ref{th:t_r_generic}) builds at least $t_s+t_n-2$ EDSTs for any graphs $G_s$ and $G_n$. 
This solution is tweaked and augmented in Section~\ref{sec:max_construct} 
 to get an extra tree for a total of $t_s+t_n-1$ EDSTs when $r_s \ge t_s$ or $r_n \ge t_n$. 
However, when $r_s < t_s$ and $r_n < t_n$, 
neither the Universal construction nor the approaches in Section~\ref{sec:max_construct}
save enough edges in the right places to build an extra tree. 

A construction with $t_s+t_n-1$ EDSTs can be shown in this case, but to do so, we must impose Property~\ref{property:no_edge_conditions}, which partitions the edges of some spanning tree in some set of supernode EDSTs into two large-enough sets: one connected and both fixed under $f$. 
This construction gives more EDSTs than the Universal construction, and gives a maximum construction when $r_s < t_s$ and $r_n < t_n$. 
\begin{property}\label{property:no_edge_conditions} A star product $G^*$ has Property \ref{property:no_edge_conditions} if there exists a set of maximum size $\mathbf{Y}$ of EDSTs in $G_n$ containing a spanning tree  $Y_i \in \mathbf{Y}$ with edge partition $S_1$ and $S_2$ having vertices $V(S_1)$ and $V(S_2)$ respectively, where $S_2$ is connected and
\begin{enumerate}
[itemsep=0pt,parsep=1pt,leftmargin=*]
    \item $|S_1| \ge  t_n -2 +|V(S_1) \cap V(S_2)|$\footnote{This is a slightly different and tighter bound from that in ~\cite{edst_ipdps_2025}.},
    \item $|S_2| \ge t_n-2+|V(S_1) \cap V(S_2)|$, and
    \item $f_{(x,x')}(V(S_j))=V(S_j)$  for
    $X_i \in \mathbf{X}$, $(x,x') \in E(\bar{X}_i)$, and $i \in \{1,2\}$.\label{property:no_edge_conditions:fixed}
\end{enumerate}
\end{property}

Intuitively, $G^*$ has Property \ref{property:no_edge_conditions} if we can partition some spanning tree in the set of EDSTs of $G_n$ so that edges between supernodes connect all vertices in one of the partitions to the vertices in the other. 
The idea for the construction is then to locate $S_1$ and $S_2$ in pairs of supernodes, connect them, and thus guarantee $S_2$ and $S_1$ edges are in place after the two basic constructions to connect an additional tree.

The following are consequences of Property~\ref{property:no_edge_conditions}.  
\begin{itemize}
[itemsep=0pt,parsep=1pt,leftmargin=*]
    \item Every vertex in $V(S_1)$ is connected to some vertex in $V(S_2)$ by some sequence of edges in $S_1$. This is key to the construction of additional EDSTs.
    \item No vertex in $V(S_1)$ is an ancestor of a vertex in $V(S_2)$, and thus  $V(S_1) \cap V(S_2) \neq \emptyset$. ($S_1$ and $S_2$ are of course nonempty since $S_1 \cup S_2$ is a partition.)
    \item The edges of $S_2$ form a subtree of $Y_1$: $S_2$ is then a ``top'' subtree of $Y_1$ and $S_1$ is a collection of ``bottom'' branches. (Branches may be separated from the tree at different levels.) 
    \item For all $(x,x') \in E(\bar{X}_i)$, $f_{(x,x')}(V(S_1) \cap V(S_2))=V(S_1) \cap V(S_2)$. 
\end{itemize}
Cartesian products trivially have Property~\ref{property:no_edge_conditions} (since in this case, $f$ is the identity).
It is implicitly used but not stated in \cite{Product_STs_2003} in the partition for Cartesian products. 
Some (but not all) non-Cartesian star products have Property~\ref{property:no_edge_conditions} as well. 
Clearly, given a supernode and a structure graph, one can explicitly construct a non-Cartesian star product with $f$ satisfying Property~\ref{property:no_edge_conditions}.
There are also star products that do not have Property~\ref{property:no_edge_conditions}: consider a star product such that for some directed edge $(x,x')$, $f_{(x,x')}$ is the cyclic permutation on all the vertices of the supernode. This will violate condition \ref{property:no_edge_conditions:fixed} in Property~\ref{property:no_edge_conditions}.

\subsubsection{The Construction}

In Theorem~\ref{th:t_r_lt}, we present a construction that generalizes the Cartesian result in \cite{Product_STs_2003}. Since all Cartesian products satisfy Property~\ref{property:no_edge_conditions}, Theorem~\ref{th:t_r_lt} applies there, as well as to any non-Cartesian star product having the property.
\begin{theorem}
\label{th:t_r_lt}\footnote{In this paper, the conditions for a maximum construction are extended from the $r_s=0,r_n\le t_n$ or $r_n=0,r_s\le t_s$ conditions presented in \cite{edst_ipdps_2025}.}
If $G^*$ satisfies Property~\ref{property:no_edge_conditions}, we can construct $t_s + t_n - 1$ EDSTs without any constraints on $r_s$ or $r_n$. This construction of $t_s + t_n - 1$ EDSTs is of maximum cardinality when $r_s < t_s$ and $r_n < t_n$.
\end{theorem}

Let $Y_1$ be a spanning tree of a maximal set of EDSTs $\mathbf{Y}$ of the supernode such that $Y_1$ satisfies the conditions of \ref{property:no_edge_conditions}, and let $o' \in S_2$. Let $I=V(S_1) \cap V(S_2)$.

To prove Theorem \ref{th:t_r_lt}, we first describe modifications on Constructions \ref{construction:t1_trees} and \ref{construction:t2_trees} that produce a set of $t_s+t_n-2$ EDSTs, while leaving enough remaining edges to construct one additional EDST without relying on the use of non-tree edges of the factor graphs. We then construct the additional tree. Then we show that these sets of trees can share no edges. As in Section
\ref{subsection:nearly_maximum_set}, 
it is important to establish direction in the EDSTs in $\mathbf{X}$. We choose an arbitrary vertex $o \in V_s$. In our constructions, we will define directed trees rooted at $o$ for each $X_i \in \mathbf{X}$, which we will label $\bar{X}_i$. Partition the vertices of $V_s \setminus \{o\}$ into balanced sets $R_1$ and $R_2$. 

To complete the construction, we need a partition that allows us to choose some vertices $c_i$ and $d_i$ from outside the intersection $V(S_1) \cap V(S_2)$ so that we have edges between them to connect supernodes. 

For each $2 \le i \le t_s$, construct the triples $(i,a_i,b_i)$ where $a_i \in R_1$ and $b_i \in R_2$ such that no two triples have the same $i,a_i,$ or $b_i$. For each $2 \le i \le t_n$, construct the triples $(i,c_i,d_i)$ where $c_i \in V(S_1) \setminus I$ and $d_i \in V(S_2) \setminus I$ such that no two triples have the same $i,c_i,$ or $d_i$. 

It is here that Property~\ref{property:no_edge_conditions} comes in. 
To guarantee existence of these triples, we need $|V(S_1)| \ge t_n -1+|I|$ and $|V(S_2)| \ge t_n-1+|I|$.
We know that $|V(S_i)|> |S_i|$, so by Property~ \ref{property:no_edge_conditions}, $|V(S_i)| \ge t_n-1+|I|.$ This shows that there are sufficient choices for the triples $(i,c_i, d_i)$.

In this section, we continue to use notation from Table~\ref{table:notation}. Table~\ref{table:notation_additional}  gives additional notation used only here. 
\begin{table}[ht]
\centering
\renewcommand\arraystretch{1.35}
\begin{tabular}{|p{1.5cm}|p{6.5cm}|}
\hline
$R_1$ and $R_2$ & balanced partition of the vertices $V_s \setminus \{o\}$ \\ \hline
$S_1$ and $S_2$ & partition of $Y_1$ edges: $|S_1| \ge t_n-2+|I|$, $|S_2| \ge t_n-2+|I|$, $S_2$ is connected, and both $V(S_1)$ and $V(S_2)$ are fixed under $f$.
\\ \hline
$o$ & an arbitrary vertex in $V_s$ \\ \hline
$o'$ &  an arbitrary vertex in $S_2$ \\ \hline
$I$ & the set of vertices $V(S_1) \cap V(S_2)$ \\ \hline
$(i,a_i,b_i)$ & unique triples from vertices in $R_1$ and $R_2$ \\ \hline
$(i,c_i,d_i)$ & unique triples from vertices in $V(S_1)\setminus I$ and $V(S_2)\setminus I$ \\ \hline
\end{tabular}
\caption{Notation used in Section~\ref{sec:different_construction}.}
\label{table:notation_additional} 
\end{table}
\begin{figure*}[ht]
\centering
\includegraphics[width=1\textwidth]{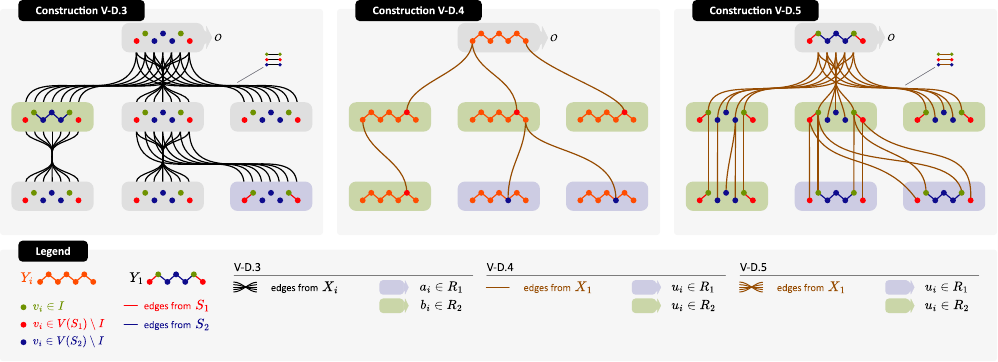}
\caption{The constructions producing the maximal solution when $r_s=0$ and $r_n=0$ is shown above. The legend depicts the choices of vertices and edges used to construct these EDSTs in a star product $G^*$ that satisfies Property~\ref{property:no_edge_conditions}.}
\label{fig:no_edge_maximal}
\end{figure*}

Our first Construction~\ref{construction:no_edge_t1_trees} is similar to Construction \ref{construction:t1_trees}, but is altered to leave enough connected edges to form a final subtree in the final construction.

\begin{construction} \label{construction:no_edge_t1_trees}
Construction of $T_i$ is via $X_i$ and $Y_1$ for $2 \le i \le t_s$. 

    \emph{Intuition:} 
    For each tree, instead of selecting one supernode to build an entire copy of $Y_1$, we select two supernodes and build $S_1$ in the first and $S_2$ in the second. We connect these edges using the $i^{th}$ structure graph spanning tree, and since $S_2$ is connected and $S_1 \cap S_2 \ne \emptyset$, we are able to connect the subgraphs.

Note that in each pair of supernodes, using the edges of $S_1$ in the first and the edges of $S_2$ in the second in Construction~\ref{construction:no_edge_t1_trees} leaves untouched the edges of $S_2$ in the first and $S_1$ in the second. We will use these edges and their connectivity properties in the final construction of the $(t_s+t_n-1)^{th}$ spanning tree.

    \emph{Formal Construction:} The following edge set defines the subgraph. Let $2\leq i\leq t_s$.
    \begin{align}
    &\{\{(a_i, s_1), (a_i,s_1')\} \ | \ \{s_1,s_1'\} \in S_1\}, \label{no_edge_t1_1} \\
    &\{\{(b_i, s_2), (b_i, s_2')\} \ | \ \{s_2,s_2'\} \in S_2\}, \label{no_edge_t1_2} \\
    &\{\{(x,y),(x',f_{(x,x')}(y))\} \ | \ (x,x') \in E(\bar{X}_i), y \in V_n\}. \label{no_edge_t1_3}
    \end{align}
\end{construction}
Our second Construction~\ref{construction:no_edge_t2_trees} is essentially the same as Construction \ref{construction:t2_trees}, and does not use Property~\ref{property:no_edge_conditions}.
\begin{construction}\label{construction:no_edge_t2_trees}
Construction of $T_i'$ via $Y_i$ and $X_1$ for all $2 \le i \le t_n$.

 \emph{Intuition:} For tree $T_i'$, we build $Y_i$ in each supernode, then connect all supernodes using edges formed from $X_1$ similar to Construction \ref{construction:t2_trees}. The uniqueness of the triples $(i,c_i,d_i)$ allows for edge-disjointness amongst these trees, as well as the final tree constructed. We connect  sink nodes of $\bar{X}_1$ in $R_1$ by a single edge defined by $d_i$ and sink nodes in $R_2$ by a single edge defined by $c_i$. We depict this in Figure \ref{fig:no_edge_maximal}.

\emph{Formal Construction:} Formally, we use the following edge sets to define the subgraph, where $2 \le i \le t_n.$
\begin{flalign}
    &\{\{(x,f_{(x,x')}^{-1}(c_i)),(x',c_i)\} \ | \ (x,x') \in E(\bar{X}_1), x' \in R_2 \}, \label{no_edge_t2_1}&&\\
    &\{\{(x,f_{(x,x')}^{-1}(d_i)),(x',d_i)\} \ | \ (x,x') \in E(\bar{X}_1), x' \in R_1 \}, \label{no_edge_t2_2}&&\\
    &\{\{(x, y), (x, y')\} \ | \ x \in V_s, \{y, y'\} \in E(Y_i)\}. \label{no_edge_t2_3}&&
\end{flalign}

\end{construction}
At this point, we have built $t_s+t_n-2$ EDSTs. We now use the $S_1$ and $S_2$ edges left untouched in Construction~\ref{construction:no_edge_t1_trees} to construct the $(t_s+t_n-1)^{th}$ tree.
\begin{construction}\label{construction:no_edge_final_tree} Construction of $T$ via $X_1$ and $Y_1$.

\emph{Intuition:} For tree $T$, we build $Y_1$ in supernode $o$ and build $S_2$ and $S_1$ in each $r_1 \in R_1$ and $r_2 \in R_2$ supernode, respectively. supernodes.
For each edge $(x,x')$ in $\bar{X}_1$, if the sink supernode $x'$ is in $R_1$, we use all edges between supernodes $x$ and $x'$ that are incident with $V(S_1)$ in supernode $x'$. If the sink supernode is in $R_2$, we use all edges between $x$ and $x'$ that are incident with $V(S_2)$.
Since every vertex in $V(S_1)$ is connected to a vertex in $V(S_2)$ by edges in $S_1$ and the edges of $S_2$ form a subtree of $Y_1$, Property~\ref{property:no_edge_conditions} guarantees connectedness by mapping vertices in the edge partition $S_1$ and $S_2$ back to their respective sets. We depict this in Figure \ref{fig:no_edge_maximal}.

\emph{Formal Construction:} Formally, we use the following edge sets to define the subgraph.
\begin{align}
    &\{\{(o, y), (o, y')\} \ | \ \{y, y'\} \in E(Y_1)\}, \label{no_edge_1_1} \\
    &\{\{(r_1, s_2), (r_1, s_2')\} \ | \ r_1 \in R_1, \{s_2, s_2'\} \in S_2\}, \label{no_edge_1_2} \\
    &\{\{(r_2, s_1), (r_2, s_1')\} \ | \ r_2 \in R_2, \{s_1, s_1'\} \in S_1\}, \label{no_edge_1_3} \\
    \begin{split}&\{\{(x,f_{(x,x')}^{-1}(s_1)),(x',s_1)\} \ | \ (x,x') \in E(\bar{X}_1),\\
    &x' \in R_1,  s_1 \in V(S_1)\}, \label{no_edge_1_4} 
    \end{split}\\
    \begin{split}
    &\{\{(x,f_{(x,x')}^{-1}(s_2)),(x',s_2)\} \ | \ (x,x') \in E(\bar{X}_1),\label{no_edge_1_5}\\
    & x' \in R_2, s_2 \in V(S_2) \}. 
    \end{split}
\end{align}
\end{construction}

Construction \ref{construction:no_edge_t2_trees} does not require the use of Property~\ref{property:no_edge_conditions} and, by the same proof as in Construction~\ref{construction:t2_trees_maximum}, can be seen to produce $t_n-1$ EDSTs. For Constructions \ref{construction:no_edge_t1_trees} and \ref{construction:no_edge_final_tree}, we need Property~\ref{property:no_edge_conditions}. In Lemma~\ref{lemma:no_edge_additional_trees}, we prove that these constructions produce spanning subgraphs.

\begin{lemma}\label{lemma:no_edge_additional_trees}
    Constructions \ref{construction:no_edge_t1_trees}, \ref{construction:no_edge_t2_trees}, and \ref{construction:no_edge_final_tree} are connected spanning subgraphs of $G^*$.
\end{lemma}

\begin{proof}
    First, we show Construction \ref{construction:no_edge_t1_trees} produces $t_s-1$ connected spanning subgraphs in $G^*$. Consider the subgraph $T_i$ with $S_1$ in supernode $a_i$ and $S_2$ in supernode $b_i$ via \eqref{no_edge_t1_1} and \eqref{no_edge_t1_2}, respectively. 
    It is clear from \eqref{no_edge_t1_3} that $T_i$ spans the vertices of $G^*$.

    It remains to show that $T_i$ is connected. Let $(x,y)$ be an arbitrary vertex in $G^*$. We will show that $(x,y)$ is connected to $(a_i, s_2)$ for some $s_2 \in S_2.$ First, if $y \in S_2$, then by Property \ref{property:no_edge_conditions} and the edges in \eqref{no_edge_t1_3}, there exists a path to $(a_i, s_2)$ for some $s_2 \in S_2$. Next, if $y \not \in S_2$, then $y \in S_1$. Thus by Property \ref{property:no_edge_conditions} and \eqref{no_edge_t1_3}, there exists a path from $(x,y)$ to $(b_i, s_1)$ for some $s_1 \in S_1$. There exists a path from $(b_i, s_1)$ to $(b_i, s_2)$ for some $s_2 \in S_2$ by the construction of the partition $S_1 \cup S_2$. Thus from the above argument, there exists a path from $(b_i, s_2)$ to $(a_1, s_2')$ for some $s_2' \in S_2$. Concatenating the paths gives the result. 

    Since every vertex $(x,y)$ is connected to some $(a_i, s_2)$ with $s_2 \in S_2$ and since $S_2$ is connected, then $T_i$ must be connected.

    Construction \ref{construction:no_edge_t2_trees} is very similar to Construction \ref{construction:t2_trees}, and a similar inductive proof shows that each is a spanning subgraph.

    Lastly, we show Construction \ref{construction:no_edge_final_tree} produces a connected spanning subgraph $T$ in $G^*$. Consider a neighbor $x'$ of $o$ and pick some $v \in V_n$. We will show that $(x',v)$ is connected to supernode $o$. If $x' \in R_1$ then in the supernode $x'$ we have constructed $S_2$.
    If $v \in V(S_1)$, then there is a direct edge in \eqref{no_edge_1_4} between supernode $o$ and $(x',v)$ by Property~\ref{property:no_edge_conditions}. If $v\in V(S_2)$, then since every vertex in $V(S_1)$ has a path to a vertex in $V(S_2)$, we can first follow this path to $(x', w)$ for some $w \in V(S_1)$, then use the direct edge to supernode $o$. Now suppose $x' \in R_2$. Then the supernode contains a copy of $S_1$. If $v\in V(S_2)$, then there is a direct edge in \eqref{no_edge_1_5} to supernode $o$ by Property~\ref{property:no_edge_conditions}. Else, by the same logic as before, we can follow a path to $(x', w)$ for some $w \in V(S_2)$ and then use the direct edge to supernode $o$. We conclude using induction on the level of the tree $\bar{X}_1$. 
\end{proof}

As in Section \ref{sec:additional_trees}, some of these constructions may not be acyclic. By Remark \ref{rmk:cycles_okay}, we can take these constructions to produce spanning trees of $G^*$. We now complete the proof of Theorem \ref{th:t_r_lt} by showing the spanning trees produces are edge-disjoint.

\begin{proof}[Proof of Theorem \ref{th:t_r_lt}]
    By Lemma~\ref{lemma:no_edge_additional_trees} and Remark~\ref{rmk:cycles_okay}, we have $t_s+t_n-1$ spanning trees of $G^*$. It remains to show disjointness. 

Recall that $a_i \in R_1$ and $b_i \in R_2$ for all $i$ where $V_s \setminus \{o\} = R_1 \cup R_2$ and $R_1 \cap  R_2 = \emptyset$, and that $E(Y_1)=E(S_1) \cup E(S_2)$ where $E(S_1) \cap E(S_2) = \emptyset$ and $I=V(S_1) \cap V(S_2)$.

First we compare intra-supernode edges formed by \eqref{no_edge_t1_1}, \eqref{no_edge_t1_2}, \eqref{no_edge_t2_3}, \eqref{no_edge_1_1}, \eqref{no_edge_1_2}, and \eqref{no_edge_1_3}. Since the pair $(a_i, b_i)$ is unique for each $i = 2, \ldots, t_s$, \eqref{no_edge_t1_1} and \eqref{no_edge_t1_2} are disjoint. Both are disjoint from \eqref{no_edge_1_1} since $a_i, b_i \in V\setminus\{o\}$. Since $R_1 \cap R_2 = \emptyset$, \eqref{no_edge_1_2} and \eqref{no_edge_1_3} are disjoint. Further, since $a_i \in R_1$ and $b_i \in R_2$, then \eqref{no_edge_t1_1} is disjoint from \eqref{no_edge_1_2} and \eqref{no_edge_t1_2} is disjoint from \eqref{no_edge_1_3}. Finally, since each $\mathbf{Y}$ is a set of EDSTs \eqref{no_edge_t2_3} is disjoint from all of the other sets.

Lastly we compare the inter-supernode edges formed by \eqref{no_edge_t1_3}, \eqref{no_edge_t2_1}, \eqref{no_edge_t2_2}, \eqref{no_edge_1_4}, and \eqref{no_edge_1_5}.

 It suffices to compare \eqref{no_edge_t2_1} to \eqref{no_edge_1_5} and \eqref{no_edge_t2_2} to \eqref{no_edge_1_4} since the pairs $(c_i, d_i)$ were chosen to be unique  and since $R_1 \cap R_2 = \emptyset$. As $c_i, d_i \not \in V(S_1 \cap S_2)$, these sets must be disjoint.

    If $r_s < t_s$ and $r_n < t_n$, then $t_s$ and $t_n$ must be maximum, and 
    maximum cardinality of the construction follows from 
    Proposition \ref{prop:max_bounds_tr}.   \end{proof}

\begin{theorem}\label{th:depth_r_lt_t}
   \begin{enumerate}
       \item The depth of the trees in Construction \ref{construction:no_edge_t1_trees} is at most $4d_{s_i} + d_{n_1}$ for each $i=1, \ldots , t_s$.
       \item The depth of the trees in Construction \ref{construction:no_edge_t2_trees} is at most $2d_{s_1}d_{n_i}+d_{n_i} + d_{s_1}$ for each $i = 2, \ldots, t_n$. 
       \item The depth of the tree in Construction \ref{construction:no_edge_final_tree} is at most $3d_{n_1} + d_{s_1}.$
   \end{enumerate}
\end{theorem}

\begin{proof}

    \begin{enumerate}
        \item Set $s \in V_H$ to be a vertex in $S_1 \cap S_2$. Let $(a_i, s)$ be the root. Suppose $(x,y)$ is a vertex not within supernode $a_i$. To get to $(b_i,s)$ takes at most $2d_{s_1}$ hops since $a_i$ might not be the root of $\bar{X}_1$. From $(b_i, s)$, it takes at most $d_{n_1}$ hops to get to the correct node $(b_i, t)$ within this supernode that has a direct path to $(x,y)$. Finally, to get to $(x,y)$ takes at most $2d_{s_1}$ hops.

        Suppose $(x,y)$ is a vertex within supernode $a_i$. If $y$ is connected to $s$ in $S_1$, then it takes at most $2d_{n_1}$ hops to reach $(x,y)$. If $y$ is not connected to $s$, then we first must traverse from $(a_i, s)$ to $(a_i, t)$ for some $t$, then to $(b_i,u)$ for some $u$, where $t$ is chosen so that $(b_i,u)$ has a direct path back to $(x,y)$. This takes at most $d_{n_1} + 4d_{s_1}$ hops.

        \item The proof is the same as that of Theorem \ref{thm:depth_universal_2}(1).

        \item Let $(o, o')$ be the root of the tree in Construction \ref{construction:no_edge_final_tree}. Consider a vertex $(x,y)$. If this vertex is adjacent to no edges from the tree within supernode $x$, then there exists a vertex $(o,z)$ such that there is a length at most $d_{s_1}$ to $(x,y)$. Thus the total distance from the root is at most $d_{s_1} + d_{n_1}$. 

        If $(x,y)$ is contained within $S_1$ (or $S_2)$, then set $s \in S_1 \cap S_2$. There is a vertex $(o, z')$ so that there is a length at most $d_{s_1}$ path from $(o,z')$ to $(x,s)$. From $(x,s)$, there is a path of length at most $2d_{n_1}$ to $(x,y)$. Thus the total distance is at most $3d_{n_1} + d_{s_1}$.
    \end{enumerate}
\end{proof}

\renewcommand\arraystretch{1.5}
\begin{table*}[ht]
\noindent\adjustbox{max width=\textwidth}{
\begin{tabular}{|c|c|c|>{\columncolor[HTML]{F3F3F3}} c|c|c|c|c|c|>{\columncolor[HTML]{F3F3F3}} c|
>{\columncolor[HTML]{E8FFE8}}c |}
\Xhline{1.6pt}
\cellcolor[HTML]{C0C0C0}\textbf{Network} &
  \cellcolor[HTML]{C0C0C0}\textbf{Parameters} &
  \cellcolor[HTML]{C0C0C0}\textbf{\begin{tabular}[c]{@{}c@{}}Degree of\\ Network if\\ it is regular\end{tabular}} &
  \cellcolor[HTML]{C0C0C0}\textbf{\begin{tabular}[c]{@{}c@{}}Upper bound\\ on \# EDSTs\\($\tau$ bound)\end{tabular}} &
  \cellcolor[HTML]{C0C0C0}\boldmath{$t_s$} &
  \cellcolor[HTML]{C0C0C0}\boldmath{$r_s$} &
  \cellcolor[HTML]{C0C0C0}\boldmath{$t_n$} &
  \cellcolor[HTML]{C0C0C0}\boldmath{$r_n$} &
  \cellcolor[HTML]{C0C0C0}\textbf{Thm.} &
  \cellcolor[HTML]{C0C0C0}\textbf{\begin{tabular}[c]{@{}c@{}}\# EDSTs\\ constructed\\ in this paper\end{tabular}} &
  \cellcolor[HTML]{C0C0C0}\textbf{Max?} \\ \Xhline{1.6pt}
 &
  $q=4k-1$ &
  $6k$ &
  $3k$ &
  $2k$ &
  $6k+1$ &
   &
   &
  \ref{cor:r_ge_t} &
  $3k$ &
  Yes \\ \cline{2-6} \cline{9-11} 
 &
  $q=4k$ &
  $6k$ &
  $3k$ &
  $2k$ &
  $2k$ &
   &
   &
  \ref{th:r_eq_t} &
  $3k$ &
  Yes \\ \cline{2-6} \cline{9-11} 
\multirow{-3}{*}{\begin{tabular}[c]{@{}c@{}}\textbf{Slim Fly \cite{slim-fly}}\\ \boldmath{$K_{q,q}*C(q)$}\\ \textbf{\cite{graph_theory_bondy,geom_MMS_2004}}\end{tabular}} &
  $q=4k+1$ &
  $6k-1$ &
  $3k-1$ &
  $2k-1$ &
  $6k-2$ &
  \multirow{-3}{*}{$k$} &
  \multirow{-3}{*}{$k$} &
  \ref{cor:r_ge_t} &
  $3k-1$ &
  Yes \\ \Xhline{1.6pt}
 &
  \begin{tabular}[c]{@{}c@{}}$q=4\ell+1$\\ $a=4k+1$\end{tabular} &
  $6\ell+2k$ &
  $3\ell+k$ &
  $3\ell$ &
  $q^2+3\ell$ &
   &
   &
  \ref{cor:r_ge_t} &
  $3\ell+k$ &
  Yes \\ \cline{2-6} \cline{9-11} 
 &
  \begin{tabular}[c]{@{}c@{}}$q=4\ell$\\ $a=4k+1$\end{tabular} &
  $6\ell+2k$ &
  $3\ell+k$ &
  $3\ell$ &
  $3\ell$ &
   &
   &
  \ref{th:r_eq_t} &
  $3\ell+k$ &
  Yes \\ \cline{2-6} \cline{9-11} 
\multirow{-5}{*}{\begin{tabular}[c]{@{}c@{}}\textbf{Bundlefly \cite{bundlefly_2020}}\\ \boldmath{$H_q*QR(a)$}\\ \textbf{\cite{geom_MMS_2004}}\end{tabular}} &
  \begin{tabular}[c]{@{}c@{}}$q=4\ell-1$\\ $a=4k+1$\end{tabular} &
  $6\ell+2k-1$ &
  $3\ell+k-1$ &
  $3\ell-1$ &
  $q^2+3\ell-1$ &
  \multirow{-5}{*}{$k$} &
  \multirow{-5}{*}{$k$} &
  \ref{cor:r_ge_t} &
  $3\ell+k-1$ &
  Yes \\ \Xhline{1.6pt}
 &
  \begin{tabular}[c]{@{}c@{}}$q$ even\\ $a=4k+1$\end{tabular} &
   &
   &
  $\frac{q}{2}$ &
  $\frac{q(q+1)}{2}$ &
   &
   &
  \ref{cor:r_ge_t} &
  $\floor{\frac{q}{2}}+k$ &
  Yes \\ \cline{2-2} \cline{5-6} \cline{9-11} 
\multirow{-3}{*}{\begin{tabular}[c]{@{}c@{}}\textbf{PolarStar \cite{PolarStar_23}}\\ \boldmath{$ER_q*QR(a)$}\\ \textbf{\cite{erdosrenyi1962,geom_MMS_2004}}\end{tabular}} &
  \begin{tabular}[c]{@{}c@{}}$q$ odd\\ $a=4k+1$\end{tabular} &
  \multirow{-3}{*}{not regular} &
  \multirow{-3}{*}{$\floor{\frac{q}{2}}+k$} &
  $\frac{q+1}{2}$ &
  0 &
  \multirow{-3}{*}{$k$} &
  \multirow{-3}{*}{$k$} &
  \ref{cor:t_r_ge_or} &
  $\floor{\frac{q}{2}}+k$ &
  Yes \\ \Xhline{1.6pt}
 &
  \begin{tabular}[c]{@{}c@{}}$q$ even\\ $d=4m>0$\end{tabular} &
   &
   &
   &
   &
  $\frac{d}{2}$ &
  $\frac{d}{2}$ &
  \ref{cor:r_ge_t} &
  $\floor{\frac{q+d}{2}}$ &
  Yes \\ \cline{2-2} \cline{7-11} 
 &
  \begin{tabular}[c]{@{}c@{}}$q$ even\\ $d=4m+3$\end{tabular} &
   &
   &
  \multirow{-2}{*}{$\frac{q}{2}$} &
  \multirow{-2}{*}{$\frac{q(q+1)}{2}$} &
  $\frac{d-1}{2}$ &
  $\frac{3d+1}{2}$ &
  \ref{cor:r_ge_t} &
  $\floor{\frac{q+d}{2}}$ &
  Yes \\ \cline{2-2} \cline{5-11} 
 &
  \begin{tabular}[c]{@{}c@{}}$q$ odd\\ $d=4m>0$\end{tabular} &
   &
   &
   &
   &
  $\frac{d}{2}$ &
  $\frac{d}{2}$ &
  \ref{cor:t_r_ge_or} &
  $\floor{\frac{q+d}{2}}$ &
  Yes \\ \cline{2-2} \cline{7-11} 
\multirow{-7}{*}{\begin{tabular}[c]{@{}c@{}}\textbf{PolarStar \cite{PolarStar_23}}\\ \boldmath{$ER_q*IQ(d)$}\\
\textbf{\cite{erdosrenyi1962,PolarStar_23}}\end{tabular}} &
  \begin{tabular}[c]{@{}c@{}}$q$ odd\\ $d=4m+3$\end{tabular} &
  \multirow{-7}{*}{not regular} &
  \multirow{-7}{*}{$\floor{\frac{q+d}{2}}$} &
  \multirow{-2}{*}{$\frac{q+1}{2}$} &
  \multirow{-2}{*}{0} &
  $\frac{d-1}{2}$ &
  $\frac{3d+1}{2}$ &
  \ref{cor:t_r_ge_or} &
  $\floor{\frac{q+d}{2}}-1$ &
  \cellcolor[HTML]{FFE8E8}Maybe \\ \Xhline{1.6pt}
\end{tabular}
}
\caption{Comparison of the number of constructed EDSTs (column $10$) to the upper bound on EDSTs (column $4$) for several modern star-product networks. All but one of these networks are of maximum cardinality:  this is guaranteed if Theorem \ref{th:r_eq_t} applies, and is observed here when the constructed number meets the calculated upper bound from Proposition~\ref{prop:max_bounds_uv}. 
\\
The upper bounds are calculated in Appendix \ref{sec:star_product_edsts}.
The values for $t_i$ and $r_i$ are calculated in Appendix \ref{sec:appendix_factor_graphs}, and are summarized in Table~\ref{table:factor_graphs}.  We assume for the purposes of this table that all factor graphs have sets of EDSTs reaching the theoretical maximum cardinality, and know this in many cases \cite{paley_ham_2012,EDST_survey_2001,allreduce_PF_2023}.
}
\label{table:network_graphs}
\end{table*}

\section{Discussion}\label{sec:discussion}
There is a trade-off between a solution with a maximum number of EDSTs and a solution with low-depth EDSTs.

\begin{itemize}[itemsep=0pt,parsep=2pt,leftmargin=*]
    \item The Universal solution gives $t_s+t_n-2$ trees, with worst-case depth $\max(2d_{s_i} + d_{n_1},d_{s_1} + 2d_{n_i})$ where $d_{s_i}$ and $d_{n_1}$ are the depths of the
    EDSTs in factor graphs $G_s$ and $G_n$.
    \item The Maximum solutions give $t_s+t_n$ or $t_s+t_n-1$ trees, with worst-case depth $(d_{s_{1}} + 1)(2d_{n_{1}} + 2r_{s})+d_{s_{1}}$.
\end{itemize}
In the Maximum solution, we get one or two more trees, but its worst-case depth is quadratic in the 
depth of factor graph EDSTs rather than linear as in the Universal solution. 
While depth affects the latency of communication over the trees, 
cardinality affects the available parallelism and communication bandwidth.
Solution choice will depend on application needs.

Low-latency small-depth EDSTs may be better for routing, synchronization and broadcasts (for fault tolerance or recovery)~\cite{awerbuch1986reliable, mudigonda2010spain, besta2020fatpaths, petrini2001hardware,rabenseifner2000automatic}. Reducing hops 
traversed by packets also improves injection bandwidth per switch~\cite{slim-fly, polarfly_sc22}.
In contrast, more trees may be better for distributed machine learning applications, giving more parallelism and bandwidth~\cite{allreduce_PF_2023}.
In-network Allreduce is typically implemented over spanning trees in network~\cite{graham2016scalable, graham2020scalable, lakhotia2021accelerating, allreduce_PF_2023, de2024canary} and maximizing Allreduce bandwidth is crucial for distributed training~\cite{sergeev2018horovod, ben2019demystifying, li2020pytorch, akiba2017extremely,lakhotia2021accelerating}. 
For instance, the Ring Allreduce is widely used in host-based implementations despite its high depth and latency, as it avoids congestion and achieves high bandwidth~\cite{thakur2005optimization, sergeev2018horovod, ueno2018technologies, patarasuk2009bandwidth, verbraeken2020survey}.

The number of EDSTs constructed here is at worst $3$ fewer and typically no more than $1$ fewer than the maximum possible, by Theorem~\ref{th:one_less_than_max}. While this seems a small difference, it may compound in a multi-level star product when fewer than the maximum are constructed on multiple levels.
Fortunately, the percentage of loss decreases as the number of factor-graph trees (and the radix) increase. 
For example, Bundlefly is the star product of the star product $H_q$ and the Paley graph. The Maximum solution applied to $H_q$ then to Bundlefly gives $4$ more trees than the Universal. In a radix-$32$ Bundlefly, 
this gives $25\%$ more
trees, and a proportional speedup in bandwidth sensitive operations. On the other hand, in a radix-$256$ Bundlefly, the Maximum gives only $3.125\%$ more trees. 

We note that the Maximum solution does have a large number ($t_s-1$) of low-depth trees,  with depth similar to that of the Universal solution. These can be used to cater to latency-sensitive applications. 
Also, in large systems, multiple jobs and communicator groups may co-exist simultaneously. 
An individual communicator will span a subset of nodes in the cluster and in-network collective protocols like SHARP may allocate only a subtree(s) of an EDST(s) to any given communicator~\cite{graham2016scalable, graham2020scalable}.
The depth of subtrees may thus be much smaller than the entire EDST(s), and having more EDSTs can simplify congestion-free allocation of subtrees.

 \section{Related Work}
In 1961, Tutte~\cite{Tutte_graphdecomp_1961} and Nash-Williams~\cite{NashWilliams_edsts_1961} independently derived necessary and sufficient conditions for an undirected graph
to have $k$ EDSTs. There are recent comprehensive surveys on the existence and construction of EDSTs on general graphs \cite{EDST_survey_2001,sp_trees_survey_2023}. In \cite{EDST_survey_2001}, Palmer relates the maximum number of EDSTs in an undirected graph $\sigma(G)$ to its edge-connectivity $\lambda(G)$: $\sigma(G) \ge \floor{\frac{\lambda(G)}{2}}$. In \cite{marcote_2006}, Marcote et al. show a lower bound for the edge-connectivity of the star product in terms of the edge-connectivity and degrees of the factor graphs. 
Here, we give actual constructions with a guaranteed number of 
EDSTs in terms of factor-graph EDSTs and non-tree edges. 

There are several ways to find EDSTs in graphs.
Roskind and Tarjan introduced an algorithm that finds $k$ EDSTs for an arbitrary graph and a user specified $k$, in $\mathcal{O}(n^2k^2)$ time. 
It is a greedy algorithm on matroids that incrementally adds edges to an independent set of edges which represents $k$ edge-disjoint forests.
Since the greedy algorithm for a maximum independent set on a matroid is optimal~\cite{welsh2010matroid}, this algorithm will add maximum number of edges and output $k$ spanning trees if they exist. 
However, the incremental nature of the algorithm makes it hard to control the depth and other properties of EDSTs.

Another option is to 
use topological properties of the underlying network graph to
derive 
EDST constructions,
demonstrated on 
such networks as Hypercube \cite{barden_hypercube}, PolarFly \cite{allreduce_PF_2023} and the star network \cite{fragopoulou_fault_tolerance}.
The network construction mechanism (e.g., via a graph product) is an important topological property and 
has been used to construct a set of EDSTs of maximum or near-maximum size in Cartesian product graphs~\cite{Product_STs_2003}.

The star product generalizes the Cartesian product \cite{marcote_2006}, and was used to construct state-of-the-art diameter-3 network topologies Bundlefly~\cite{bundlefly_2020} and PolarStar~\cite{PolarStar_23}. Slim Fly~\cite{slim-fly} is another modern low-diameter topology that was not constructed as a star product but is recognized as one here. 
Application of the EDST constructions in this paper to these star product topologies shows excellent results, as seen in Table~\ref{table:network_graphs}.

\section{Conclusion and Future Work}
In this paper, we present constructions for maximum or near-maximum number of EDSTs in star products. We generalize the Cartesian-product results in \cite{Product_STs_2003} and further extend them by analyzing the depth of proposed EDST constructions. 

It is reasonable to think that some of the extensive known results on the Cartesian product ~\cite{cart_prod_networks_1997,Youssef1991CartesianPN} might also generalize to the star product, 
opening new possibilities in the analysis of such networks.
These include bisection bandwidth \cite{bisect_BW_cart_2014}, path diversity and fault diameter \cite{fault_diam_cart_2000}, routing algorithms \cite{Youssef1991CartesianPN}, deadlock avoidance \cite{deadlock_cart_1998}, and resource placement \cite{resource_placement_cart_2010}. Extension of these and other Cartesian network results will be an interesting and important area of future research for successful implementation of star-product networks.

\clearpage
\bibliographystyle{IEEEtran}
\bibliography{biblio}

\begin{thebibliography}{10}
\providecommand{\url}[1]{#1}
\csname url@samestyle\endcsname
\providecommand{\newblock}{\relax}
\providecommand{\bibinfo}[2]{#2}
\providecommand{\BIBentrySTDinterwordspacing}{\spaceskip=0pt\relax}
\providecommand{\BIBentryALTinterwordstretchfactor}{4}
\providecommand{\BIBentryALTinterwordspacing}{\spaceskip=\fontdimen2\font plus
\BIBentryALTinterwordstretchfactor\fontdimen3\font minus \fontdimen4\font\relax}
\providecommand{\BIBforeignlanguage}[2]{{%
\expandafter\ifx\csname l@#1\endcsname\relax
\typeout{** WARNING: IEEEtran.bst: No hyphenation pattern has been}%
\typeout{** loaded for the language `#1'. Using the pattern for}%
\typeout{** the default language instead.}%
\else
\language=\csname l@#1\endcsname
\fi
#2}}
\providecommand{\BIBdecl}{\relax}
\BIBdecl

\bibitem{slim-fly}
M.~Besta and T.~Hoefler, ``{Slim Fly: A Cost Effective Low-Diameter Network Topology},'' in \emph{Proceedings of the Conference on High Performance Computing Networking, Storage and Analysis}.\hskip 1em plus 0.5em minus 0.4em\relax New York, NY, USA: Association for Computing Machinery, 11 2014.

\bibitem{bundlefly_2020}
\BIBentryALTinterwordspacing
F.~Lei, D.~Dong, X.~Liao, and J.~Duato, ``Bundlefly: A low-diameter topology for multicore fiber,'' in \emph{Proceedings of the 34th ACM International Conference on Supercomputing}, ser. ICS '20.\hskip 1em plus 0.5em minus 0.4em\relax New York, NY, USA: Association for Computing Machinery, 2020. [Online]. Available: \url{https://doi.org/10.1145/3392717.3392747}
\BIBentrySTDinterwordspacing

\bibitem{PolarStar_23}
K.~Lakhotia, L.~Monroe, K.~Isham, M.~Besta, N.~Blach, T.~Hoefler, and F.~Petrini, ``Polar{S}tar: Expanding the scalability of diameter-3 networks,'' in \emph{Proceedings of the 36th ACM Symposium on Parallelism in Algorithms and Architectures}, ser. SPAA '24, 2024, pp. 345--357.

\bibitem{cart_prod_networks_1997}
K.~Day and A.-E. Al-Ayyoub, ``The cross product of interconnection networks,'' \emph{IEEE Transactions on Parallel and Distributed Systems}, vol.~8, no.~2, pp. 109--118, 1997.

\bibitem{Youssef1991CartesianPN}
\BIBentryALTinterwordspacing
A.~Youssef, ``Cartesian product networks,'' in \emph{International Conference on Parallel Processing}, 1991. [Online]. Available: \url{https://api.semanticscholar.org/CorpusID:8249681}
\BIBentrySTDinterwordspacing

\bibitem{HyperX_2009}
\BIBentryALTinterwordspacing
J.~H. Ahn, N.~Binkert, A.~Davis, M.~McLaren, and R.~S. Schreiber, ``Hyper{X}: Topology, routing, and packaging of efficient large-scale networks,'' in \emph{Proceedings of the Conference on High Performance Computing Networking, Storage and Analysis}, ser. SC '09.\hskip 1em plus 0.5em minus 0.4em\relax New York, NY, USA: Association for Computing Machinery, 2009. [Online]. Available: \url{https://doi.org/10.1145/1654059.1654101}
\BIBentrySTDinterwordspacing

\bibitem{bluegene_2011}
\BIBentryALTinterwordspacing
D.~Chen, N.~A. Eisley, P.~Heidelberger, R.~M. Senger, Y.~Sugawara, S.~Kumar, V.~Salapura, D.~L. Satterfield, B.~Steinmacher-Burow, and J.~J. Parker, ``The {IBM} {B}lue {G}ene/{Q} interconnection network and message unit,'' in \emph{Proceedings of 2011 International Conference for High Performance Computing, Networking, Storage and Analysis}, ser. SC '11.\hskip 1em plus 0.5em minus 0.4em\relax New York, NY, USA: Association for Computing Machinery, 2011. [Online]. Available: \url{https://doi.org/10.1145/2063384.2063419}
\BIBentrySTDinterwordspacing

\bibitem{lakhotia2021network}
K.~Lakhotia, F.~Petrini, R.~Kannan, and V.~Prasanna, ``In-network reductions on multi-dimensional hyperx,'' in \emph{2021 IEEE Symposium on High-Performance Interconnects (HOTI)}.\hskip 1em plus 0.5em minus 0.4em\relax IEEE, 2021, pp. 1--8.

\bibitem{graham2016scalable}
R.~L. Graham, D.~Bureddy, P.~Lui, H.~Rosenstock, G.~Shainer, G.~Bloch, D.~Goldenerg, M.~Dubman, S.~Kotchubievsky, V.~Koushnir \emph{et~al.}, ``Scalable hierarchical aggregation protocol ({SHA}r{P}): A hardware architecture for efficient data reduction,'' in \emph{2016 First International Workshop on Communication Optimizations in HPC (COMHPC)}.\hskip 1em plus 0.5em minus 0.4em\relax IEEE, 2016, pp. 1--10.

\bibitem{graham2020scalable}
R.~L. Graham, L.~Levi, D.~Burredy, G.~Bloch, G.~Shainer, D.~Cho, G.~Elias, D.~Klein, J.~Ladd, O.~Maor \emph{et~al.}, ``Scalable hierarchical aggregation and reduction protocol ({SHARP})\textsuperscript{\texttrademark} streaming-aggregation hardware design and evaluation,'' in \emph{High Performance Computing: 35th International Conference, ISC High Performance 2020, Frankfurt/Main, Germany, June 22--25, 2020, Proceedings 35}.\hskip 1em plus 0.5em minus 0.4em\relax Springer, 2020, pp. 41--59.

\bibitem{lakhotia2021accelerating}
K.~Lakhotia, F.~Petrini, R.~Kannan, and V.~Prasanna, ``Accelerating {A}llreduce with in-network reduction on {I}ntel {PIUMA},'' \emph{IEEE Micro}, vol.~42, no.~2, pp. 44--52, 2021.

\bibitem{allreduce_PF_2023}
\BIBentryALTinterwordspacing
K.~Lakhotia, K.~Isham, L.~Monroe, M.~Besta, T.~Hoefler, and F.~Petrini, ``In-network {A}llreduce with multiple spanning trees on {P}olar{F}ly,'' in \emph{Proceedings of the 35th ACM Symposium on Parallelism in Algorithms and Architectures}, ser. SPAA '23.\hskip 1em plus 0.5em minus 0.4em\relax New York, NY, USA: Association for Computing Machinery, 2023, p. 165–176. [Online]. Available: \url{https://doi.org/10.1145/3558481.3591073}
\BIBentrySTDinterwordspacing

\bibitem{de2024canary}
D.~De~Sensi, E.~C. Molero, S.~Di~Girolamo, L.~Vanbever, and T.~Hoefler, ``Canary: Congestion-aware in-network {A}llreduce using dynamic trees,'' \emph{Future Generation Computer Systems}, vol. 152, pp. 70--82, 2024.

\bibitem{fragopoulou_fault_tolerance}
P.~Fragopoulou and S.~G. Akl, ``Edge-disjoint spanning trees on the star network with applications to fault tolerance,'' \emph{IEEE Transactions on Computers}, vol.~45, no.~2, pp. 174--185, 1996.

\bibitem{awerbuch1986reliable}
B.~Awerbuch and S.~Event, ``Reliable broadcast protocols in unreliable networks,'' \emph{Networks}, vol.~16, no.~4, pp. 381--396, 1986.

\bibitem{petrini2001hardware}
F.~Petrini, S.~Coll, E.~Frachtenberg, and A.~Hoisie, ``Hardware-and software-based collective communication on the quadrics network,'' in \emph{Proceedings IEEE International Symposium on Network Computing and Applications. NCA 2001}.\hskip 1em plus 0.5em minus 0.4em\relax IEEE, 2001, pp. 24--35.

\bibitem{Product_STs_2003}
S.-C. Ku, B.-F. Wang, and T.-K. Hung, ``Constructing edge-disjoint spanning trees in product networks,'' \emph{IEEE Transactions on Parallel and Distributed Systems}, vol.~14, no.~3, pp. 213--221, 2003.

\bibitem{petrini2002quadrics}
F.~Petrini, W.-C. Feng, A.~Hoisie, S.~Coll, and E.~Frachtenberg, ``The quadrics network: High-performance clustering technology,'' \emph{IEEE Micro}, vol.~22, no.~1, pp. 46--57, 2002.

\bibitem{roskind_paper}
J.~Roskind and R.~E. Tarjan, ``A note on finding minimum-cost edge-disjoint spanning trees,'' \emph{Mathematics of Operations Research}, vol.~10, no.~4, pp. 701--708, 1985.

\bibitem{roskind_thesis}
J.~A. Roskind, ``Edge disjoint spanning trees and failure recovery in data communication networks,'' Ph.D. dissertation, Massachusetts Institute of Technology, 1983.

\bibitem{edst_ipdps_2025}
K.~Isham, L.~Monroe, K.~Lakhotia, A.~Dawkins, D.~Hwang, and A.~Kubicek, ``Edge-disjoint spanning trees on star-product networks,'' in \emph{2025 IEEE International Parallel and Distributed Processing Symposium (IPDPS)}, June 2025.

\bibitem{bermond82}
J.~Bermond, C.~Delorme, and G.~Farhi, ``{Large graphs with given degree and diameter III},'' \emph{Ann. of Discrete Math.}, vol.~13, pp. 23--32, 1982.

\bibitem{chimera_2014}
P.~I. Bunyk, E.~M. Hoskinson, M.~W. Johnson, E.~Tolkacheva, F.~Altomare, A.~J. Berkley, R.~Harris, J.~P. Hilton, T.~Lanting, A.~J. Przybysz, and J.~Whittaker, ``Architectural considerations in the design of a superconducting quantum annealing processor,'' \emph{IEEE Transactions on Applied Superconductivity}, vol.~24, no.~4, pp. 1--10, 2014.

\bibitem{chimera_2016}
\BIBentryALTinterwordspacing
T.~Boothby, A.~D. King, and A.~Roy, ``Fast clique minor generation in {C}himera qubit connectivity graphs,'' \emph{Quantum Information Processing}, vol.~15, no.~1, p. 495–508, jan 2016. [Online]. Available: \url{https://doi.org/10.1007/s11128-015-1150-6}
\BIBentrySTDinterwordspacing

\bibitem{bisect_BW_cart_2014}
J.~A. Aroca and A.~F. Anta, ``Bisection (band)width of product networks with application to data centers,'' \emph{IEEE Transactions on Parallel and Distributed Systems}, vol.~25, no.~3, pp. 570--580, 2014.

\bibitem{fault_diam_cart_2000}
K.~Day and A.-E. Al-Ayyoub, ``Minimal fault diameter for highly resilient product networks,'' \emph{IEEE Transactions on Parallel and Distributed Systems}, vol.~11, no.~9, pp. 926--930, 2000.

\bibitem{deadlock_cart_1998}
R.~Kr{\'a}{\v{l}}ovi{\v{c}}, B.~Rovan, P.~Ru{\v{z}}i{\v{c}}ka, and D.~{\v{S}}tefankovi{\v{c}}, ``Efficient deadlock-free multi-dimensional interval routing in interconnection networks,'' in \emph{Distributed Computing}, S.~Kutten, Ed.\hskip 1em plus 0.5em minus 0.4em\relax Berlin, Heidelberg: Springer Berlin Heidelberg, 1998, pp. 273--287.

\bibitem{resource_placement_cart_2010}
\BIBentryALTinterwordspacing
N.~Imani, H.~Sarbazi-Azad, and A.~Zomaya, ``Resource placement in {C}artesian product of networks,'' \emph{Journal of Parallel and Distributed Computing}, vol.~70, no.~5, pp. 481--495, 2010. [Online]. Available: \url{https://www.sciencedirect.com/science/article/pii/S0743731509001142}
\BIBentrySTDinterwordspacing

\bibitem{MMS_98}
\BIBentryALTinterwordspacing
B.~D. McKay, M.~Miller, and J.~Siráň, ``A note on large graphs of diameter two and given maximum degree,'' \emph{Journal of Combinatorial Theory, Series B}, vol.~74, no.~1, pp. 110--118, 1998. [Online]. Available: \url{https://www.sciencedirect.com/science/article/pii/S0095895698918287}
\BIBentrySTDinterwordspacing

\bibitem{erdosrenyi_paley_1963}
P.~Erd{\H o}s and A.~R\'enyi, ``Asymmetric graphs,'' \emph{Acta Mathematica Academiae Scientiarum Hungaricae}, vol.~14, pp. 295--315, 1963.

\bibitem{erdosrenyi1962}
------, ``On a problem in the theory of graphs,'' \emph{Publ. Math. Inst. Hungar. Acad. Sci.}, vol.~7A, pp. 623--641, 1962.

\bibitem{brown_1966}
W.~G. Brown, ``On graphs that do not contain a {T}homsen graph,'' \emph{Can. Math. Bull.}, vol.~9, no.~3, p. 281–285, 1966.

\bibitem{petersen_1886}
\BIBentryALTinterwordspacing
A.~B. Kempe, ``A memoir on the theory of mathematical form,'' \emph{Philosophical Transactions of the Royal Society of London}, vol. 177, pp. 1--70, 1886. [Online]. Available: \url{http://www.jstor.org/stable/109477}
\BIBentrySTDinterwordspacing

\bibitem{petersen_1898}
J.~Petersen, ``Sur le théorème de tait,'' \emph{L'Intermédiaire des Mathématiciens}, vol.~5, pp. 225--227, 1898.

\bibitem{pegasus_2020}
\BIBentryALTinterwordspacing
K.~Boothby, P.~I. Bunyk, J.~Raymond, and A.~Roy, ``Next-generation topology of {D}-wave quantum processors,'' \emph{arXiv: Quantum Physics}, 2020. [Online]. Available: \url{https://api.semanticscholar.org/CorpusID:201076106}
\BIBentrySTDinterwordspacing

\bibitem{zephyr_2021}
\BIBentryALTinterwordspacing
K.~Boothby, A.~D. King, and J.~Raymond, ``Zephyr topology of {D}-wave quantum processors: Technical report,'' 2021. [Online]. Available: \url{https://api.semanticscholar.org/CorpusID:250352637}
\BIBentrySTDinterwordspacing

\bibitem{graph_theory_bondy}
J.~A. Bondy and U.~S.~R. Murty, \emph{Graph Theory}.\hskip 1em plus 0.5em minus 0.4em\relax London: Springer, 2008.

\bibitem{geom_MMS_2004}
\BIBentryALTinterwordspacing
P.~R. Hafner, ``Geometric realisation of the graphs of {M}c{K}ay–{M}iller–{Š}iráň,'' \emph{Journal of Combinatorial Theory, Series B}, vol.~90, no.~2, pp. 223--232, 2004. [Online]. Available: \url{https://www.sciencedirect.com/science/article/pii/S0095895603000868}
\BIBentrySTDinterwordspacing

\bibitem{paley_ham_2012}
\BIBentryALTinterwordspacing
B.~Alspach, D.~Bryant, and D.~Dyer, ``Paley graphs have {H}amilton decompositions,'' \emph{Discrete Mathematics}, vol. 312, no.~1, pp. 113--118, 2012, algebraic Graph Theory — A Volume Dedicated to Gert Sabidussi on the Occasion of His 80th Birthday. [Online]. Available: \url{https://www.sciencedirect.com/science/article/pii/S0012365X11002603}
\BIBentrySTDinterwordspacing

\bibitem{EDST_survey_2001}
\BIBentryALTinterwordspacing
E.~Palmer, ``On the spanning tree packing number of a graph: a survey,'' \emph{Discrete Mathematics}, vol. 230, no.~1, pp. 13--21, 2001. [Online]. Available: \url{https://www.sciencedirect.com/science/article/pii/S0012365X00000662}
\BIBentrySTDinterwordspacing

\bibitem{mudigonda2010spain}
J.~Mudigonda, P.~Yalagandula, M.~Al-Fares, and J.~C. Mogul, ``{SPAIN}: {COTS} data-center ethernet for multipathing over arbitrary topologies.'' in \emph{NSDI}, vol.~10, 2010, pp. 18--18.

\bibitem{besta2020fatpaths}
M.~Besta, M.~Schneider, M.~Konieczny, K.~Cynk, E.~Henriksson, S.~Di~Girolamo, A.~Singla, and T.~Hoefler, ``Fatpaths: Routing in supercomputers and data centers when shortest paths fall short,'' in \emph{SC20: International Conference for High Performance Computing, Networking, Storage and Analysis}.\hskip 1em plus 0.5em minus 0.4em\relax IEEE, 2020, pp. 1--18.

\bibitem{rabenseifner2000automatic}
R.~Rabenseifner, ``Automatic {MPI} counter profiling,'' in \emph{42nd {CUG} {C}onference}, 2000, pp. 396--405.

\bibitem{polarfly_sc22}
K.~Lakhotia, M.~Besta, L.~Monroe, K.~Isham, P.~Iff, T.~Hoefler, and F.~Petrini, ``Polar{F}ly: A cost-effective and flexible low-diameter topology,'' in \emph{Proceedings of the International Conference on High Performance Computing, Networking, Storage and Analysis}, IEEE.\hskip 1em plus 0.5em minus 0.4em\relax New York, NY, USA: Association for Computing Machinery, 2022, pp. 1--15.

\bibitem{sergeev2018horovod}
A.~Sergeev and M.~Del~Balso, ``Horovod: fast and easy distributed deep learning in {T}ensor{F}low,'' \emph{arXiv preprint arXiv:1802.05799}, 2018.

\bibitem{ben2019demystifying}
T.~Ben-Nun and T.~Hoefler, ``Demystifying parallel and distributed deep learning: An in-depth concurrency analysis,'' \emph{ACM Computing Surveys (CSUR)}, vol.~52, no.~4, pp. 1--43, 2019.

\bibitem{li2020pytorch}
S.~Li, Y.~Zhao, R.~Varma, O.~Salpekar, P.~Noordhuis, T.~Li, A.~Paszke, J.~Smith, B.~Vaughan, P.~Damania \emph{et~al.}, ``Py{T}orch distributed: Experiences on accelerating data parallel training,'' \emph{arXiv preprint arXiv:2006.15704}, 2020.

\bibitem{akiba2017extremely}
T.~Akiba, S.~Suzuki, and K.~Fukuda, ``Extremely large minibatch {SGD}: Training {R}es{N}et-50 on {I}mage{N}et in 15 minutes,'' \emph{arXiv preprint arXiv:1711.04325}, 2017.

\bibitem{thakur2005optimization}
R.~Thakur, R.~Rabenseifner, and W.~Gropp, ``Optimization of collective communication operations in {MPICH},'' \emph{The International Journal of High Performance Computing Applications}, vol.~19, no.~1, pp. 49--66, 2005.

\bibitem{ueno2018technologies}
Y.~Ueno and K.~Fukuda, ``Technologies behind distributed deep learning: {A}llreduce,'' \emph{online] https://tech. preferred. jp/en/blog/technologies-behinddistributed-deep-learning-allreduce/[Aug. 2023]}, 2018.

\bibitem{patarasuk2009bandwidth}
P.~Patarasuk and X.~Yuan, ``Bandwidth optimal all-reduce algorithms for clusters of workstations,'' \emph{Journal of Parallel and Distributed Computing}, vol.~69, no.~2, pp. 117--124, 2009.

\bibitem{verbraeken2020survey}
J.~Verbraeken, M.~Wolting, J.~Katzy, J.~Kloppenburg, T.~Verbelen, and J.~S. Rellermeyer, ``A survey on distributed machine learning,'' \emph{Acm computing surveys (csur)}, vol.~53, no.~2, pp. 1--33, 2020.

\bibitem{Tutte_graphdecomp_1961}
\BIBentryALTinterwordspacing
W.~T. Tutte, ``On the problem of decomposing a graph into n connected factors,'' \emph{Journal of The London Mathematical Society-second Series}, pp. 221--230, 1961. [Online]. Available: \url{https://api.semanticscholar.org/CorpusID:121586383}
\BIBentrySTDinterwordspacing

\bibitem{NashWilliams_edsts_1961}
\BIBentryALTinterwordspacing
C.~S. J.~A. Nash-Williams, ``Edge-disjoint spanning trees of finite graphs,'' \emph{Journal of The London Mathematical Society-second Series}, vol.~36, pp. 445--450, 1961. [Online]. Available: \url{https://api.semanticscholar.org/CorpusID:121193791}
\BIBentrySTDinterwordspacing

\bibitem{sp_trees_survey_2023}
\BIBentryALTinterwordspacing
B.~Cheng, D.~Wang, and J.~Fan, ``Independent spanning trees in networks: A survey,'' \emph{ACM Comput. Surv.}, vol.~55, no. 14s, jul 2023. [Online]. Available: \url{https://doi.org/10.1145/3591110}
\BIBentrySTDinterwordspacing

\bibitem{marcote_2006}
\BIBentryALTinterwordspacing
X.~Marcote, C.~Balbuena, P.~García-Vázquez, and J.~Valenzuela, ``Highly connected star product graphs,'' \emph{Electronic Notes in Discrete Mathematics}, vol.~26, pp. 91--96, 2006, combinatorics 2006. [Online]. Available: \url{https://www.sciencedirect.com/science/article/pii/S1571065306001107}
\BIBentrySTDinterwordspacing

\bibitem{welsh2010matroid}
D.~J. Welsh, \emph{Matroid theory}.\hskip 1em plus 0.5em minus 0.4em\relax Courier Corporation, 2010.

\bibitem{barden_hypercube}
B.~Barden, R.~Libeskind-Hadas, J.~Davis, and W.~Williams, ``On edge-disjoint spanning trees in hypercubes,'' \emph{Information Processing Letters}, vol.~70, no.~1, pp. 13--16, 1999.

\bibitem{bipart_EDST_2010}
\BIBentryALTinterwordspacing
S.~Li, W.~Li, and X.~Li, ``The generalized connectivity of complete bipartite graphs,'' \emph{Ars Comb.}, vol. 104, pp. 65--79, 2010. [Online]. Available: \url{https://api.semanticscholar.org/CorpusID:14240977}
\BIBentrySTDinterwordspacing

\end{thebibliography}
\clearpage
\begin{appendices}
\section{Proofs of a Theorem on Upper Bounds on the Number of EDSTs}\label{sec:appendix_proofs}
In this appendix, we provide a proof for Proposition~\ref{prop:max_bounds_tr} from Section~\ref{sec:upper_bounds}, which establishes linearity of our constructions in terms of the upper bound on the number of EDSTs in the factor graphs. 

First, we recall that 1) $G$ is simple, 2) $|V_G|\ge 2$ and 3) the canonical upper bound on the number of EDSTs in a graph is
\begin{equation}\label{eq:edsts_appendix}
    \tau = \floor{\frac{|E|}{|V|-1}}.    
\end{equation}
\begin{lemma}\label{lemma:ubE}
    Let $G^* = G_s*G_n$ be a star product, with $G^*$ simple. Let $\tau_s$ and $\tau_n$ be the canonical upper bounds on the number of EDSTs possible in $G_s$ and $G_n$, and let $\rho_s$ and $\rho_n$ be the number of unused edges that would be in $G_s$ and $G_n$ if the $\tau_s$ and $\tau_n$ EDSTs were constructed. Then the number of edges in $G*$ is $$|E| \le (\tau_s+\tau_n)(|V|-1) + \lambda,$$ where
    $$
        \lambda = |V_n|\left(\frac{1}{2}+\rho_s-\tau_s\right) 
        \ +|V_s|\left(\frac{1}{2}+\rho_n-\tau_n\right).
    $$
\end{lemma}
\begin{proof}
    We have that 
    \begin{align}
        |E_s| &= \tau_s(|V_s|-1)+\rho_s, \text{ and}\nonumber\\
        |E_n| &= \tau_n(|V_n|-1)+\rho_n, \text{ where}\nonumber\\
        \rho_s &< |V_s|-1 \text{ and } \rho_n < |V_n|-1.\label{eq:bounds_R}
    \end{align}
    since $\tau_s$ and $\tau_n$ are upper bounds as per Equation~(\ref{eq:edsts_appendix}).
    So 
    \begin{align*}
        |E|
        &=\tau_s[(|V_s|-1)+\rho_s]|V_n| + [\tau_n(|V_n|-1)+\rho_n]|V_s|\\
        &= (\tau_s+\tau_n)(|V|-1)+L,
    \end{align*}
    where
    \begin{align}
        L
        &=\tau_s+\tau_n +(\rho_s-\tau_s)|V_n|+(\rho_n-\tau_n)|V_s|\nonumber\\
        &\le \frac{|V_s|}{2}+\frac{|V_n|}{2}+(\rho_s-\tau_s)|V_n| \ +(\rho_n-\tau_n)|V_s|\nonumber\\
        &=|V_n|\left(\frac{1}{2}+\rho_s-\tau_s\right) 
        \ +|V_s|\left(\frac{1}{2}+\rho_n-\tau_n\right).\label{eq:U}
    \end{align}
We obtain the inequality in the second line from the fact that
\begin{equation}\label{eq:tau_ub}
    |E|\le \frac{|V|(|V-1|)}{2}, \text{ which implies that } \tau_i \le 
        \frac{|V|}{2}.
\end{equation}
If the factor graphs are complete, equality in (\ref{eq:U}) is attained, and we have a tight upper bound $\lambda$ on $L$: 
\begin{equation}\label{eq:lambda}
    \lambda = |V_n|\left(\frac{1}{2}+\rho_s-\tau_s\right) 
        \ +|V_s|\left(\frac{1}{2}+\rho_n-\tau_n\right),
\end{equation}   
and $|E| \le (\tau_s+\tau_n)(|V|-1)+\lambda.$
\end{proof}
\maxboundstr*
\begin{proof}
    By Equation~(\ref{eq:edsts_appendix}) and Lemma~\ref{lemma:ubE}, a (possibly loose) upper bound $\tau$ is     \begin{equation}\label{eq:tau}
        \tau = \tau_s+\tau_n + \floor{\frac{\lambda}{|V|-1}},
    \end{equation}
where
    $$
        \lambda = |V_n|\left(\frac{1}{2}+\rho_s-\tau_s\right) 
        \ +|V_s|\left(\frac{1}{2}+\rho_n-\tau_n\right).
    $$
    To evaluate $\lambda$ in terms of $|V|-1$, we use the following inequality, derived from Inequality~(\ref{eq:bounds_R}) and the facts that $\tau_i\ge 0$ and that $\rho_i \le \tau_i-1$ when $\rho_i<\tau_i$.

    \begin{numcases}{\frac{1}{2}+\rho_i-\tau_i \le}
           |V_i|-\frac{3}{2}, &\text{if $\rho_i \ge \tau_i$},\label{eq:bounds_mid_ge}\\ 
           -\frac{1}{2}, &\text{if $\rho_i < \tau_i$}.\label{eq:bounds_mid_lt}
    \end{numcases}
    We then consider each set of conditions on $\tau_i$ and $\rho_i$.
    \setcounter{paragraph}{0}
    {\parindent0pt
    \newline
    \paragraph{Both $\rho_s\ge \tau_s$ and $\rho_n\ge \tau_n$} We use Equation~(\ref{eq:lambda}) and Inequality~(\ref{eq:bounds_mid_ge}) to get
    \begin{align*}
        \lambda &\le |V_n|\left(|V_s|-\frac{3}{2}\right)+|V_s|\left(|V_n|-\frac{3}{2}\right)\\
        &= 2|V| -\frac{3(|V_n|+|V_s|)}{2}.
    \end{align*}
    We have assumed that $|V_s|\ge 2$ and $|V_n|\ge 2$.
    When one of $|V_s|$ or $|V_n|=2$, and the other is $2$ or $3$, 
    $$
    2\le 2|V|-\frac{3(|V_s|+|V_n|)}{2}< |V|-1,
    $$ 
    which does not leave enough edges for another spanning tree, so the upper bound in those few examples is $\tau_s+\tau_n.$
    
    For all other $|V_s|$ and $|V_n| \ge 2,$ 
    $$
    |V|-1\le 2|V|-\frac{3(|V_s|+|V_n|)}{2}< 2(|V|-1),
    $$
    which leaves enough edges for another spanning tree, but not enough for two, so $\tau_s+\tau_n+1$ is an upper bound on the number of EDSTs, and serves as a general upper bound in this case.
    \newline
    \paragraph{$\rho_s=\tau_s$ and $\rho_n=\tau_n$} 
    By Equation~(\ref{eq:lambda}), 
    $$
    \lambda = \frac{|V_n|}{2}+\frac{|V_s|}{2}.
    $$
    Both $|V_n|$ and $|V_s| \ge 2,$ so $$2 \le \frac{|V_n|}{2}+\frac{|V_s|}{2} < |V|-1,$$ and the upper bound on EDSTs in this case is $\tau_s+\tau_n$.
    \newline
    \paragraph{$\rho_s\ge \tau_s$ and $\rho_n< \tau_n$ (or $\rho_s< \tau_s$ and $\rho_n\ge \tau_n$)} Without loss of generality, we consider the case $\rho_s\ge \tau_s$ and $\rho_n< \tau_n$. We use Equation~(\ref{eq:lambda}) and Inequalities~(\ref{eq:bounds_mid_ge}) and ~(\ref{eq:bounds_mid_lt}) to get
    \begin{align*}
        \lambda &= |V_n|\left(|V_s|-\frac{3}{2}\right)-\frac{|V_s|}{2}=|V|-\frac{3|V_n|+|V_s|}{2}.      
    \end{align*}
    Since both $|V_n|$ and $|V_s|\ge 2$, we have that 
    $$
        0 \le |V|-\frac{3|V_n|+|V_s|}{2} < |V|-1
    $$
    so there are not enough edges for another spanning tree, and $\tau_s+\tau_n$ is an upper bound for the number of EDSTs that can be constructed in this case.
    \newline
    \paragraph{Both $\rho_s< \tau_s$ and $\rho_n< \tau_n$}     
    By Equation~(\ref{eq:lambda}) and Inequality~(\ref{eq:bounds_mid_lt}), we have 
    $$
        \lambda = -\left(\frac{|V_n|}{2}+\frac{|V_s|}{2}\right).
    $$
    }
    Again, both $|V_n|$ and $|V_s| \ge 2,$ so 
    $$
    -(|V|-1) < -\left(\frac{|V_n|}{2}+\frac{|V_s|}{2}\right) \le -2,
    $$
    This means that there must be at least $1$ fewer spanning tree than the $\tau_s+\tau_n$ we started with, so an upper bound on EDSTs here is $\tau_s+\tau_n-1$.
\end{proof}

\section{Bounds on the Number of Factor Graph EDSTs}
In this appendix, we explicitly calculate upper bounds on the number of EDSTs for factor graphs of the star products discussed in this paper. We use these bounds to establish that our constructions on all but one examined graph produce a maximum number of EDSTs. 
\subsection{Useful Theorems}\label{sec:appendix_general_star_prod_ths}
We begin with two theorems from \cite{bermond82} giving the degree and number of vertices and edges in a star product, in terms of the degrees and number of vertices and edges of its factor graphs. We restate Property~\ref{property:star_props} from Section~\ref{sec:defs}.
\props*
\begin{property}\cite{bermond82} \label{property:star_reg_degree}
    If the maximum degrees in $G_s$ and $G_n$ are $d_s$ and $d_n$, the degree of $G^*$ is $d \leq d_s + d_n$. In particular, if both $G_s$ and $G_n$ are regular, the degree of $G^*$ is $d_s + d_n$.\label{th:star_props_app_deg_appendix}
\end{property}

The star-product upper-bound calculations in Section~\ref{sec:star_product_edsts} entail a fair amount of combinatorial complexity when calculated directly from Equation~\eqref{eq:edsts}:
\begin{equation*}
    \floor{\frac{|E|}{|V|-1}}.   
\end{equation*}
Proposition \ref{prop:max_bounds_uv} and Corollary \ref{cor:sptree_identity_regular} from Section~\ref{sec:upper_bounds}, restated here, considerably simplify these calculations for the star products discussed in this paper. 

Proposition \ref{prop:max_bounds_uv} is used when the graph has the property that the number of its edges $|E|$ shares a factor with the number of its vertices $|V|$. In particular, this is true when at least one of the factor graphs has this property. In that case, by Property~\ref{property:star_props}, the common factor drops out of both the numerator and the denominator in the division $\frac{|E|}{|V|}$. 
We observe that this is true when at least one of the factor graphs are regular. 
Thus, Proposition \ref{prop:max_bounds_uv} applies to Slim Fly, Bundlefly and PolarStar.

\ubsimple*
\begin{proof}
    $G$ is simple, so $m\le \frac{|V|-1}{2}$. There are at most $\floor{\frac{|E|}{|V|-1}}$ EDSTs in $G$. We have that $|E|=m(|V|-1)+(m+c)$. 

    If $m+c<|V|-1$, then the upper bound $\tau$ on the number of EDSTs is $m$.    
    If $m+c\ge |V|-1$, we have that 
    \begin{align*}
        m+c &\le m+(|V|-1) \\
        &\le \frac{|V|-1}{2}+(|V|-1)\\
        &< 2(|V|-1).
    \end{align*}
    So $|E| < (m+2)(|V| - 1)$ and the upper bound $\tau$ is $m+1$.
     
\end{proof}
Corollary~\ref{cor:sptree_identity_regular} is useful when both factor graphs are regular, by Proposition~\ref{prop:reg_star}. In this case, Property~\ref{property:star_reg_degree} gives a simple formula for the degree of the star product in terms of the degrees of the factor graphs, and the upper bound is very easily calculated using Corollary~\ref{cor:sptree_identity_regular}. 
In particular, Corollary~\ref{cor:sptree_identity_regular} applies to Slim Fly and Bundlefly, but not to PolarStar, since its structure graph PolarFly is not regular (when self-loops are not included in its edge set).
\ubregular*

\begin{proposition}\label{prop:reg_star}
    Let $G^*=G_s*G_n$ be a star product. Then $G^*$ is regular if and only if both $G_s$ and $G_n$ are regular.
\end{proposition}
\begin{proof}
    Follows by the definition of the star product.
\end{proof}

\subsection{Factor Graph Parameters}\label{sec:appendix_factor_graphs}

\renewcommand{\arraystretch}{2.5}
\begin{table}[ht]
\noindent\adjustbox{max width=\columnwidth}{
\begin{tabular}{|c|l|c|c|c|c|c|}
\Xhline{1.6pt}
\rowcolor[HTML]{C0C0C0} 
\textbf{Graph} &
  \textbf{Parameters} &
  \textbf{$|V|$} &
  \textbf{$|E|$} &
  \textbf{ST edges} &
  \textbf{$t$} &
  \textbf{$r$} \\ \Xhline{1.6pt}
 &
  $q=4k+1$ &
  $4k+1$ &
  $k(4k+1)$ &
  $4k$ &
   &
   \\ \cline{2-5}
 &
  $q=4k$ &
  $4k$ &
  $4k^2$ &
  $4k-1$ &
   &
   \\ \cline{2-5}
\multirow{-3}{*}{\boldmath{$C(q)$}} &
  $q=4k-1$ &
  $4k-1$ &
  $k(4k-1)$ &
  $4k-2$ &
  \multirow{-3}{*}{$k$} &
  \multirow{-3}{*}{$k$} \\ \Xhline{1.6pt}
 &
  $q=4k+1$ &
   &
   &
   &
  $2k$ &
  $6k+1$ \\ \cline{2-2} \cline{6-7} 
 &
  $q=4k$ &
   &
   &
   &
  $2k$ &
  $6k$ \\ \cline{2-2} \cline{6-7} 
\multirow{-3}{*}{\boldmath{$K_{q,q}$}} &
  $q=4k-1$ &
  \multirow{-3}{*}{$2q$} &
  \multirow{-3}{*}{$q^2$} &
  \multirow{-3}{*}{$2q-1$} &
  $2k-1$ &
  $6k-2$ \\ \Xhline{1.6pt}
 &
  $q=4k+1$ &
   &
  $\frac{q^2(3q-1)}{2}$ &
   &
  $3k$ &
  $q^2+3k$ \\ \cline{2-2} \cline{4-4} \cline{6-7} 
 &
  $q=4k$ &
   &
  $\frac{3q^3}{2}$ &
   &
  $3k$ &
  $3k$ \\ \cline{2-2} \cline{4-4} \cline{6-7} 
\multirow{-3}{*}{\boldmath{$H_q$}} &
  $q=4k-1$ &
  \multirow{-3}{*}{$2q^2$} &
  $\frac{q^2(3q-1)}{2}$ &
  \multirow{-3}{*}{$2q^2-1$} &
  $3k-1$ &
  $q^2+3k-1$ \\ \Xhline{1.6pt}
 &
  $m=2\ell$ &
  $2\ell$ &
  $\ell(2\ell-1)$ &
  $2\ell-1$ &
   &
  $0$ \\ \cline{2-5} \cline{7-7} 
\multirow{-2}{*}{\boldmath{$K_m$}} &
  $m=2\ell+1$ &
  $2\ell+1$ &
  $\ell(2\ell+1)$ &
  $2\ell$ &
  \multirow{-2}{*}{$\ell$} &
  $\ell$ \\ \Xhline{1.6pt}
 &
  $d$ even &
   &
   &
   &
  $\frac{d}{2}$ &
  $\frac{d}{2}$ \\ \cline{2-2} \cline{6-7} 
\multirow{-2}{*}{\boldmath{$BDF(d)$}} &
  $d$ odd &
  \multirow{-2}{*}{$2d$} &
  \multirow{-2}{*}{$d^2$} &
  \multirow{-2}{*}{$2d-1$} &
  $\frac{d-1}{2}$ &
  $\frac{3d-1}{2}$ \\ \Xhline{1.6pt}
 &
  $d=4m$ &
   &
   &
   &
  $\frac{d}{2}$ &
  $\frac{d}{2}$ \\ \cline{2-2} \cline{6-7} 
\multirow{-2}{*}{\boldmath{$IQ(d)$}} &
  $d=4m+3$ &
  \multirow{-2}{*}{$2d+2$} &
  \multirow{-2}{*}{$d(d+1)$} &
  \multirow{-2}{*}{$2d+1$} &
  $\frac{d-1}{2}$ &
  $\frac{3d+1}{2}$ \\ \Xhline{1.6pt}
 &
  $q$ even &
   &
   &
   &
  $\frac{q}{2}$ &
  $\frac{q(q+1)}{2}$ \\ \cline{2-2} \cline{6-7} 
\multirow{-2}{*}{\boldmath{$ER_q$}} &
  $q$ odd &
  \multirow{-2}{*}{$q^2+q+1$} &
  \multirow{-2}{*}{$\frac{q(q+1)^2}{2}$} &
  \multirow{-2}{*}{$q(q+1)$} &
  $\frac{q+1}{2}$ &
  $0$ \\ \Xhline{1.6pt}
\end{tabular}
}
\caption{This table gives parameters and several statistics on the factor graphs that we use throughout this paper, as calculated in this section. In particular, it counts $t$, the theoretical upper bound on the number of edge-disjoint spanning trees in the graph, by dividing the number of edges by the edges in a spanning tree, and counts $r$, the edges remaining that are not part of these trees. }
\label{table:factor_graphs}
\end{table}

In this section, we calculate parameters on the common factor graphs that are seen throughout this paper. 
These factor graph calculations are straightforward using Equation \eqref{eq:edsts}. The results in this section are included in Table~\ref{table:network_graphs}.
\subsubsection{C(a), a Cayley Graph}
The parameter $a$ may be $4k+1$, $4k$ or $4k-1$. In each of these cases, $t=k$ and $r=k$.
\paragraph[C(4k+1) = QR(4k+1) : the Paley Graph.]{C($4k+1$) = $QR(4k+1)$ : the Paley Graph.}
The Paley graph of prime power order $q=4k+1$ has

\begin{align*}
    |V|&=q, \text{ and}\\
    |E|&=\frac{q(q-1)}{4}.
\end{align*}

Each spanning tree has $q-1$ edges, so the maximum number of EDSTs in a Paley graph is
$$
t=\floor{\frac{q(q-1)}{4(q-1)}}=\frac{q-1}{4} = k.
$$ 
By \cite{paley_ham_2012}, a Paley graph may be decomposed into Hamiltonian cycles, so the bound is attained.
The number of remaining non-tree edges is
$$r=\frac{q(q-1)}{4}-\frac{(q-1)^2}{4} = \frac{q-1}{4} = k.$$
The degree of $QR(4k+1)$ is $2k$. 
\paragraph[C(4k) : an MMS Supernode.]{C($4k$) : an MMS Supernode.}
The McKay-Miller-Širáň supernode of prime power order $q=4k$ has 

\begin{align*}
    |V|&=q, \text{ and}\\
    |E|&=\frac{q^2}{4}.
\end{align*}

Each spanning tree has $q-1$ edges, so the maximum number of EDSTs in this supernode graph is
$$t=\floor{\frac{q^2}{4(q-1)}}=\frac{q}{4} = k.$$ The number of remaining non-tree edges is
$$r=\frac{q^2}{4}-\frac{q}{4}(q-1) = \frac{q}{4} = k.$$
The degree of $C(4k)$ is $2k$. 
\paragraph[C(4k-1) : an MMS Supernode.]{C($4k-1$) : an MMS Supernode.}
The McKay-Miller-Širáň supernode of prime power order $q=4k-1$ has 

\begin{align*}
    |V|&=q, \text{ and}\\
    |E|&=\frac{q(q+1)}{4}.
\end{align*}

Each spanning tree has $q-1$ edges, so the maximum number of EDSTs in this supernode graph is
$$t=\floor{\frac{q(q+1)}{4(q-1)}}=\frac{q+1}{4}=k.$$ The number of remaining non-tree edges is
$$r=\frac{q(q+1)}{4}-\frac{q+1}{4}\cdot (q-1) = \frac{q+1}{4}=k.$$
The degree of $C(4k-1)$ is $2k$. 
\subsubsection[K(q,q) : the Bipartite Graph]{$K_{q,q}$ : the Bipartite Graph}
The bipartite $(q,q)$ graph, where $q>2$ is a prime power, has 
\begin{align*}
    |V|&=2q, \text{ and}\\
    |E|&=q^2.
\end{align*}

Each spanning tree has $2q-1$ edges, so the maximum number of EDSTs in a Paley graph is
\[t=\floor{\frac{q^2}{2q-1}}=
\begin{cases*}
    \frac{q-1}{2} = 2k & if $q=4k+1,$ \\
    \frac{q}{2} = 2k & if $q=4k,$ \\
    \frac{q-1}{2} = 2k-1 & if $q=4k-1.$ 
\end{cases*}
\]
The actual number of EDSTs of $K_{q,q}$ was shown to be $\floor{\frac{q^2}{2q-1}}$ by Li et al. in \cite{bipart_EDST_2010}, so the upper bound is met.
The number of remaining non-tree edges is
\[r=q^2-t\cdot (2q-1) =
\begin{cases*}
    \frac{3q-1}{2} = 6k+1 & if $q=4k+1,$ \\
    \frac{q}{2} = 2k & if $q=4k,$ \\
    \frac{3q-1}{2} = 6k-2 & if $q=4k-1.$ 
\end{cases*}
\]
The degree of $K_{q,q}$ is $q$. 
\subsubsection[H(q), the McKay-Miller-Širáň Graph]{$H_q$, the McKay-Miller-Širáň Graph} \label{sec:MMS_tr}
$H_q$ is not only a factor graph for Bundlefly, but is a star product in its own right. 

Slim Fly is the McKay-Miller-Širáň graph $H_q$ \cite{MMS_98, geom_MMS_2004}: the star product $K_{q,q} * \text{C}(q)$, where $q$ is a prime power of the form $4k+\delta$, with $\delta \in \{1,0,-1\}$. 

The supernode C($q$) is the Paley graph $QR(q)$ when $q=4k+1$, and is the $H_q$ supernode discussed in \cite{geom_MMS_2004} when $q=4k$ or $q=4k-1$. 
In each of these cases, $|V_s|=2q$, $|E_s|=q^2$, $|V_n|=q$, and $|E_n|=kq$, so

\begin{align*}
    |V|&=2q^2, \text{ and}\\
    |E|&=q^2(2k+q).
\end{align*}

We thus derive the following identities: the maximum possible number of spanning trees $t$ in $H_q$ is  
\[t=
\begin{cases*}
    3k & if $q=4k+1,$ \\
    3k & if $q=4k,$ \\
    3k -1 & if $q=4k-1.$
\end{cases*}
\]
and the number of remaining edges $r$ is
\[r=
\begin{cases*}
    q^2+3k & if $q=4k+1,$ \\
    3k & if $q=4k,$ \\
    q^2+3k -1 & if $q=4k-1.$
\end{cases*}
\]

The degree of $H_q$ is $\frac{3q-\delta}{2}$. 
We discuss the upper bound on EDSTs for Slim Fly in Section~\ref{sec:slimfly}.

\subsubsection[K(m) : the Complete Graph on m Vertices]{$K_m$ : the Complete Graph on $m$ Vertices}
The complete graph on $m$ vertices has 

\begin{align*}
    |V|&=m, \text{ and}\\
    |E|&=\frac{m(m-1)}{2}.
\end{align*}

Each spanning tree has $m-1$ edges, so the maximum number of EDSTs in a Paley graph is
\[t=\floor{\frac{m(m-1)}{2(m-1)}}=\floor{\frac{m}{2}}=
\begin{cases*}
    \frac{m}{2} & if $m$ is even, \\
    \frac{m-1}{2} & if $m$ is odd.
\end{cases*}
\]
The the number of remaining non-tree edges is
\[r= \frac{m(m-1)}{2} - t(m-1) =
\begin{cases*}
    \ 0 & if $m$ is even, \\
    \frac{m-1}{2} & if $m$ is odd.
\end{cases*}
\]
The degree of $K_m$ is $m-1$. 
\subsubsection[BDF(d): the Bermond-Delorme-Farhi Graph]{$BDF(d)$: the Bermond-Delorme-Farhi Graph \cite{bermond82}}
The degree-$d$ Bermond-Delorme-Farhi (BDF) graph  \cite{bermond82} has 

\begin{align*}
    |V|&=2d, \text{ and}\\
    |E|&=d^2.
\end{align*}

Each spanning tree has $2d-1$ edges, so the maximum number of EDSTs in a BDF graph is
\[t=\floor{\frac{d^2}{(2d-1)}}=\floor{\frac{d}{2}}=
\begin{cases*}
    \frac{d}{2} =m & if $d=2m$, \\
    \frac{d-1}{2} =m & if $d=2m+1$.
\end{cases*}
\]
The number of remaining non-tree edges is
\[r= d^2 - t(2d-1) =
\begin{cases*}
    \frac{d}{2}=m & if $d=2m$, \\
    \frac{3d-1}{2}=3m+1 & if $d=2m+1$.
\end{cases*}
\]
\subsubsection[IQ(d): the Inductive-Quad Graph]{$IQ(d)$: the Inductive-Quad Graph \cite{PolarStar_23}}
The Inductive-Quad graph \cite{PolarStar_23} of degree $d = 4m \text{ or } d=4m+3$ has 
\begin{align*}
    |V|&=2d+2, \text{ and}\\
    |E|&=d(d+1).
\end{align*}

Each spanning tree has $2d+1$ edges, so the maximum number of EDSTs in an Inductive-Quad graph is
\[t=\floor{\frac{d(d+1)}{(2d+1)}}=\floor{\frac{d}{2}}=
\begin{cases*}
    \frac{d}{2} = 2m & if $d=4m$, \\
    \frac{d-1}{2} = 2m+1 & if $d=4m+3$.
\end{cases*}
\]
The number of remaining non-tree edges is
\[r= d(d-1) - t(2d+1) =
\begin{cases*}
    \frac{d}{2} = 2m & if $d=4m$, \\
    \frac{3d+1}{2} = 6m+5 & if $d=4m+3$.
\end{cases*}
\]
\subsubsection[ERq: the Erd\H os-R\'enyi Polarity Graph  ]{$ER_q$: the Erd\H os-R\'enyi Polarity Graph  \cite{erdosrenyi1962, brown_1966}}
ER$_q$, with $q$ a prime power \cite{erdosrenyi1962, brown_1966}, has 
\begin{align*}
    |V|&=q^2+q+1, \text{ and}\\
    |E|&=\frac{q(q+1)^2}{2}.
\end{align*}

Each spanning tree has $q(q+1)$ edges, so the maximum number of EDSTs in ER$_q$ is
\[t=\floor{\frac{q+1}{2}}=
\begin{cases*}
    \frac{q}{2} = \ell & if $q=2\ell$, \\
    \frac{q+1}{2} = \ell+1 & if $q=2\ell+1$.
\end{cases*}
\]
This theoretical bound is attained up to $q=128$ with edge-disjoint Hamiltonian paths, shown with a construction in \cite{allreduce_PF_2023}. We conjecture that this bound is always attained.
The total number of remaining non-tree edges is
\[r= \frac{q(q+1)^2}{2} - tq(q+1) =
\begin{cases*}
    \frac{q(q+1)}{2} = \ell(2\ell+1) & if $q=2\ell$, \\
    \ 0 & if $q=2\ell+1$.
\end{cases*}
\]
The degree of $ER_q$ is $q+1$ \cite{erdosrenyi1962}.

\subsection{Upper Bounds on the Number of EDSTs in  Star Products}\label{sec:star_product_edsts}
In this section, we calculate the upper bounds of the number of spanning trees on a given graph, and check to see if the results for these graphs from Theorem~\ref{th:r_eq_t} and  Corollaries \ref{cor:r_ge_t} and \ref{cor:t_r_ge_or} meet these upper bounds. In most cases, they do. A summary of the results is given in Table~\ref{table:network_graphs}.

Our strategy throughout this section will be to use Proposition \ref{prop:max_bounds_uv} and Corollary \ref{cor:sptree_identity_regular} to greatly simplify upper-bound calculations. Corollary~\ref{cor:sptree_identity_regular} applies to Slim Fly and Bundlefly, since both graphs (and their factor graphs) are regular. Proposition \ref{prop:max_bounds_uv} may be used for PolarStar, since its structure graph is not regular but its supernode is.

\subsubsection{Slim Fly}\label{sec:slimfly_appendix}
Slim Fly is the McKay-Miller-Širáň graph $H_q$ \cite{MMS_98, geom_MMS_2004}: the star product $K_{q,q} * \text{C}(q)$, where $q$ is a prime power of the form $4k+\delta$, with $\delta \in \{1,0,-1\}$. 

The degree of $H_q$ \cite{geom_MMS_2004} is 
\begin{align*}
    d &=\floor{\frac{3q-\delta}{2}}\\
    &= 
\begin{cases*}
    6k & if $\delta \in \{0,1\}$, \\
    6k-1 & if $\delta = -1$.
\end{cases*}
\end{align*}
Both $K_{q,q}$ and C$(q)$ are regular, so $H_q$ is also regular, by Proposition~\ref{prop:reg_star}. So by Corollary~\ref{cor:sptree_identity_regular}, the upper bound on the number of EDSTs is
\[\floor{\frac{d}{2}} = \floor{\frac{3q-\delta}{4}} =
\begin{cases*}
    3k & if $\delta \in \{0,1\}$, \\
    3k-1 & if $\delta = -1$.
\end{cases*}
\]

\subsubsection{Bundlefly}
Bundlefly is the star product $H_q*QR(a)$, where $q = 4\ell +\delta$ is a prime power, and $QR(a)$ is the Paley graph where $a$ is a prime power with $a=4k+1$ for some integer $k$.  Both are regular graphs, so Bundlefly is also regular, by Proposition~\ref{prop:reg_star}, and the degree of Bundlefly is the sum of the degrees of $H_q$ and $QR(a)$:
\begin{align*}
    d = 
    \begin{cases*}
    6\ell + 2k & if $\delta \in \{0,1\}$, \\
    6\ell-1 + 2k & if $\delta = -1$.
\end{cases*}
\end{align*}
Since Bundlefly is regular, Corollary~\ref{cor:sptree_identity_regular} applies, and the upper bound on the number of EDSTs in Bundlefly is
$$
\begin{cases*}
    3\ell + k & if $\delta \in \{0,1\}$, \\
    3\ell + k -1 & if $\delta = -1$.
\end{cases*}
$$

\subsubsection{PolarStar}

PolarStar is the star product $ER_q*G'$ where $q$ is an even or odd prime power. $G'$ is either the Paley graph QR($u$) with $u=4k+1$, or the Inductive Quad IQ($d$), with $d$ an integer where $d=4m$ or $d=4m+3$. In both cases,
\begin{align*}
    |V_s|&=q^2+q+1 \text{, and}\\
    |E_s| &=\frac{q(q+1)^2}{2}.
\end{align*}
\emph{Case 1:} $G' = QR(u)$, the Paley Graph.
\newline
Here, $|V_n|=u$ and $|E_n|=ku$. Note that $u$ is odd. The number of vertices is $$|V|=u(q^2+q+1)$$ and is odd, and the number of edges is
\begin{align*}
    |E| &= ku(q^2+q+1) + \frac{uq(q+1)^2}{2} \\
    &= k|V| + \frac{uq(q^2+q+1)}{2} + \frac{uq^2}{2}\\
    &= k|V| + \frac{q}{2}\cdot |V| + \frac{uq^2}{2}\\
    &= \left(\frac{q}{2}+k\right)|V| + \frac{uq^2}{2}.
\end{align*}
If $q$ is even, $m=\frac{q}{2}+k$ and $r=\frac{uq^2}{2}$, and both are integers.
We then have that $m+1 = \frac{q}{2}+k+1< u(\frac{q^2}{2}+q+1) = |V|-r.$ 

If $q$ is odd,  
\begin{align*}
    |E| &= \left(\frac{q}{2}+k \right)|V| + \frac{uq^2}{2} \\
    &= \left(\frac{q-1}{2}+k \right)\cdot |V| + \frac{uq^2+|V|}{2}\\
    &= \left(\frac{q-1}{2}+k \right)\cdot |V| + \frac{u(2q^2+q+1)}{2}.
\end{align*}
Here, $m=\frac{q-1}{2}+k$ and $r=\frac{u(2q^2+q+1)}{2}$. Both are integers, since $q$ is odd.
We then have that
\begin{align*}
m+1 &= \frac{q+1}{2}+k\\
&< \frac{q+1}{2}+2k(q+1)\\ &= \frac{q+1}{2}+\frac{(u-1)(q+1)}{2} \\
&= \frac{u}{2}(q+1) \\
&= |V|-r.
\end{align*} 

So in either case, $q$ odd or even, by Proposition \ref{prop:max_bounds_uv}, the largest possible number of EDSTs in this PolarStar is $t=m = \floor{\frac{q}{2}}+k.$
\newline
\newline
\emph{Case 2:} {$G' = IQ(d)$, the Inductive Quad Graph.}

In the $IQ$ graph, $d$ must be $4k$ or $4k+3$ for some $k$. Here, $|V_n|=2d+2,$  and $|E_n|=d(d+1)$. The number of vertices is
$$
|V|=(q^2+q+1)(2d+2) = 2(q^2+q+1)(d+1).
$$
and is even. The number of edges is
\begin{align*} 
|E|
&=(q^2+q+1)d(d+1)+\frac{(2d+2)q(q+1)^2}{2} \\ 
&=\left(\frac{d}{2}(q^2+q+1)+ \frac{q}{2}(q^2+q+1)+\frac{q^2}{2}\right)(2d+2) \\ 
&=\frac{d+q}{2}(q^2+q+1)(2d+2)+\frac{q^2}{2}(2d+2)\\
&=\frac{d+q}{2}\cdot |V|+\frac{q^2}{2}(2d+2).
\end{align*}
If both $q$ and $d$ are even or if both $q$ and $d$ are odd, $m=\frac{d+q}{2}$ is an integer, and $r=\frac{q^2}{2}(2d+2)$ is an integer regardless of the values of $q$ and $d$. We then have that 
$$
    m+1 = \frac{d+q+2}{2} < \left(\frac{q^2}{2}+q+1\right)(2d+2) = |V|-r.
$$
If exactly one of $q$ or $d$ is odd,
$$
    |E|=\frac{d+q-1}{2}\cdot |V|+ \left(\frac{q^2}{2}(2d+2)+\frac{|V|}{2}\right).
$$
So here, $m=\frac{d+q-1}{2}$ is an integer, and $r=\frac{q^2}{2}(2d+2)+\frac{v}{2}$ is also an integer, since $|V|$ is even. We then have that 
$$
    m+1 = \frac{d+q+1}{2} < (q+1)(d+1) = |V|-r.
$$

So for all parities of $q$ and $d$, by Proposition \ref{prop:max_bounds_uv}, the largest possible number of EDSTs in PolarStar is $t=m = \floor{\frac{q+d}{2}}.$

\section{The Chimera Graph}\label{sec:appendix_chimera}
\subsection{Definition of the Chimera Graph}
The \emph{Chimera graph} was used in production in the first D-Wave adiabatic quantum computers~\cite{chimera_2016}. The $n$-dimensional Chimera ($C_n$) is also a building block of D-Wave's Pegasus~\cite{pegasus_2020} and upcoming Zephyr~\cite{zephyr_2021} topologies. 
\begin{figure}[!ht]
    \centering
    \begin{subfigure}[t]{.45\linewidth}
      \centering
      \includegraphics[width=.5\linewidth]{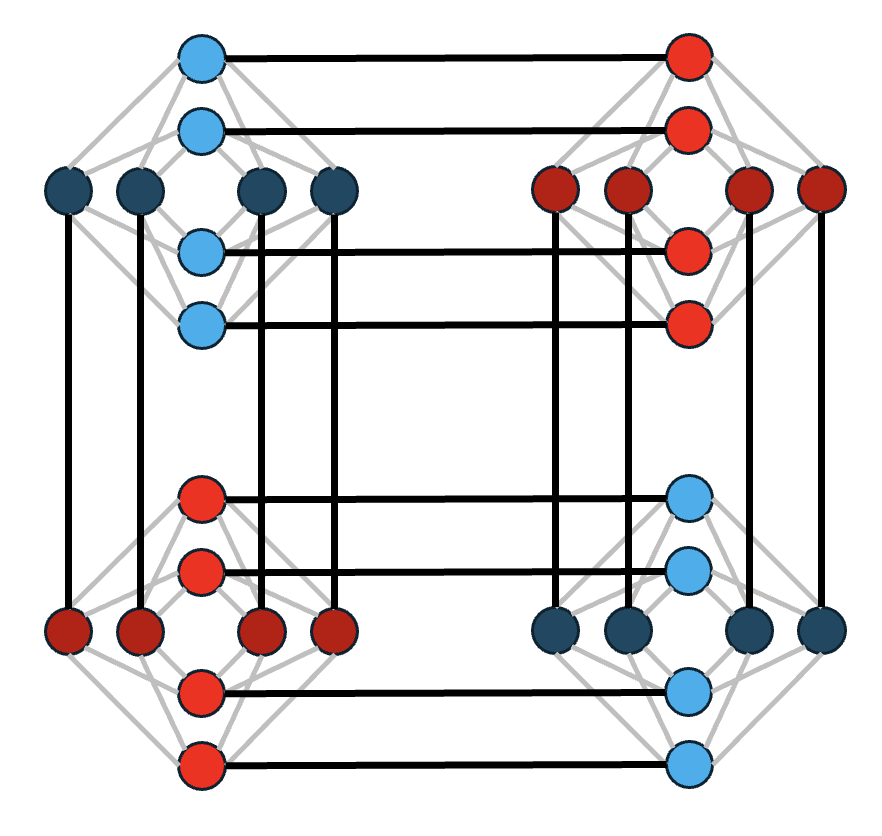}
    \caption{The Chimera graph is $C_2 = P_2*(K_{4,4}\cup K_{4,4})$. It is composed of four  $K_{4,4}$ graphs with grey edges, shown here in a cross configuration, joined horizontally and vertically by black edges. }
    \label{fig:chimera_star}
    \end{subfigure}\hfill
    \begin{subfigure}[t]{.45\linewidth}
      \centering      \includegraphics[width=.5\linewidth]{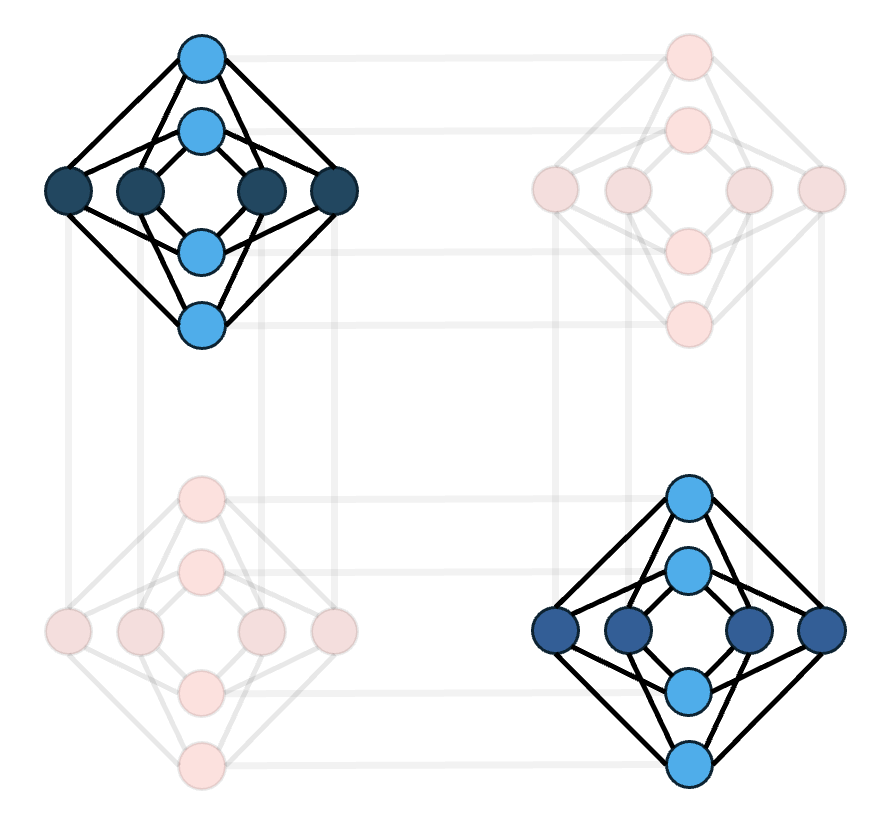}
      \caption{The disconnected $C_2$ supernode is $K_{4,4} \cup K_{4,4}$. The blue  $K_{4,4}$ instances are one supernode, and the red the other. The structure graph is just $P_2$, the two-vertex path. Thus, $C_2=P_2 * (K_{4,4} \cup K_{4,4})$.}
    \label{fig:chimera_supernode}    \end{subfigure}
\caption{The Chimera $C_2$ graph is the base graph for all Chimera graphs, and is also used in Pegasus and Zephyr. Vertically oriented vertices in $K_{4,4}$ are connected horizontally to corresponding vertices in the horizontally neighboring $K_{4,4}$, and and horizontally oriented vertices are similarly connected to corresponding vertices in the vertically neighboring $K_{4,4}$.}
\label{fig:chimera}
\vspace{-1em}
\end{figure}

The Chimera graph is especially interesting as a connected star product with disconnected supernodes. This can not happen in a Cartesian product: if the supernodes are disconnected there, so must the Cartesian product itself be disconnected. This is because the $f_{(x,y)}$ must be the identity, joining each disconnected component in one supernode to the same component in another. This illustrates the importance of careful choice of the $f_{(x,y)}$ when the supernode is disconnected. 

The $C_2$ graph, shown in Figure~\ref{fig:chimera}, is the basis for all Chimeras. $C_2$ is a $P_2 * (K_{4,4} \cup K_{4,4}).$ $P_2,$ the $2$-path, is the structure graph. $(K_{4,4} \cup K_{4,4})$ is the supernode, with $K_{4,4}$ the bipartite graph on $4$ vertices. $(K_{4,4} \cup K_{4,4})$ is the kitty-corner $K_{4,4}$ instances, unconnected to each other.

An internal portion of a Chimera graph is shown in Figure~\ref{fig:chim_internal}. Chimera $C_2$ cells are reversed in color moving from one cell horizontally or vertically to the next cell. In this illustration, blue $K_{4,4}$ are joined to each other vertically and horizontally, as are red $K_{4,4}$. 

The $C_{2(2m+1)}$ instance is a star product. The supernodes are the sets of all blue nodes with the edges joining blue nodes to other blue nodes, and all red nodes with the edges joining red nodes to other red nodes. The structure graph is a $P_2$. 

More formally, $C_{2(2m+1)}$ is a $P_2*S$, where
$$
S= \left(\bigcup_{2m^2}  C_2 \bigcup_{4m} \widehat{{K}_{4,4}} \bigcup_2 K_{4,4}\right),
$$
with $\widehat{{K}_{4,4}}$ a graph made of two copies of $K_{4,4}$, with the vertices of one of the partitions of the first $K_{4,4}$ joined bijectively to the vertices of one of the partitions of the other $K_{4,4}$. These components are shown in Figure~\ref{fig:chim_components}, and an example $C_{2\cdot(2\cdot1+1)}=C_6$ is shown in Figure~\ref{fig:chim_c6}.
\begin{figure}[ht]
\centering
\includegraphics[width=.6\columnwidth]{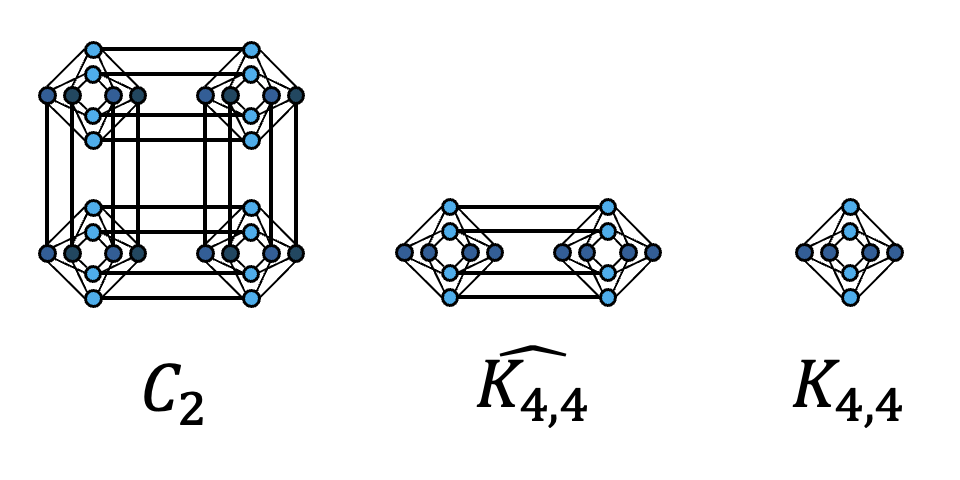}
\caption{Representative components of a $C_{2(2k+1)}$ supernode. }
\label{fig:chim_components}
\end{figure}
\begin{figure}[ht]
\centering
\includegraphics[width=.65\columnwidth]{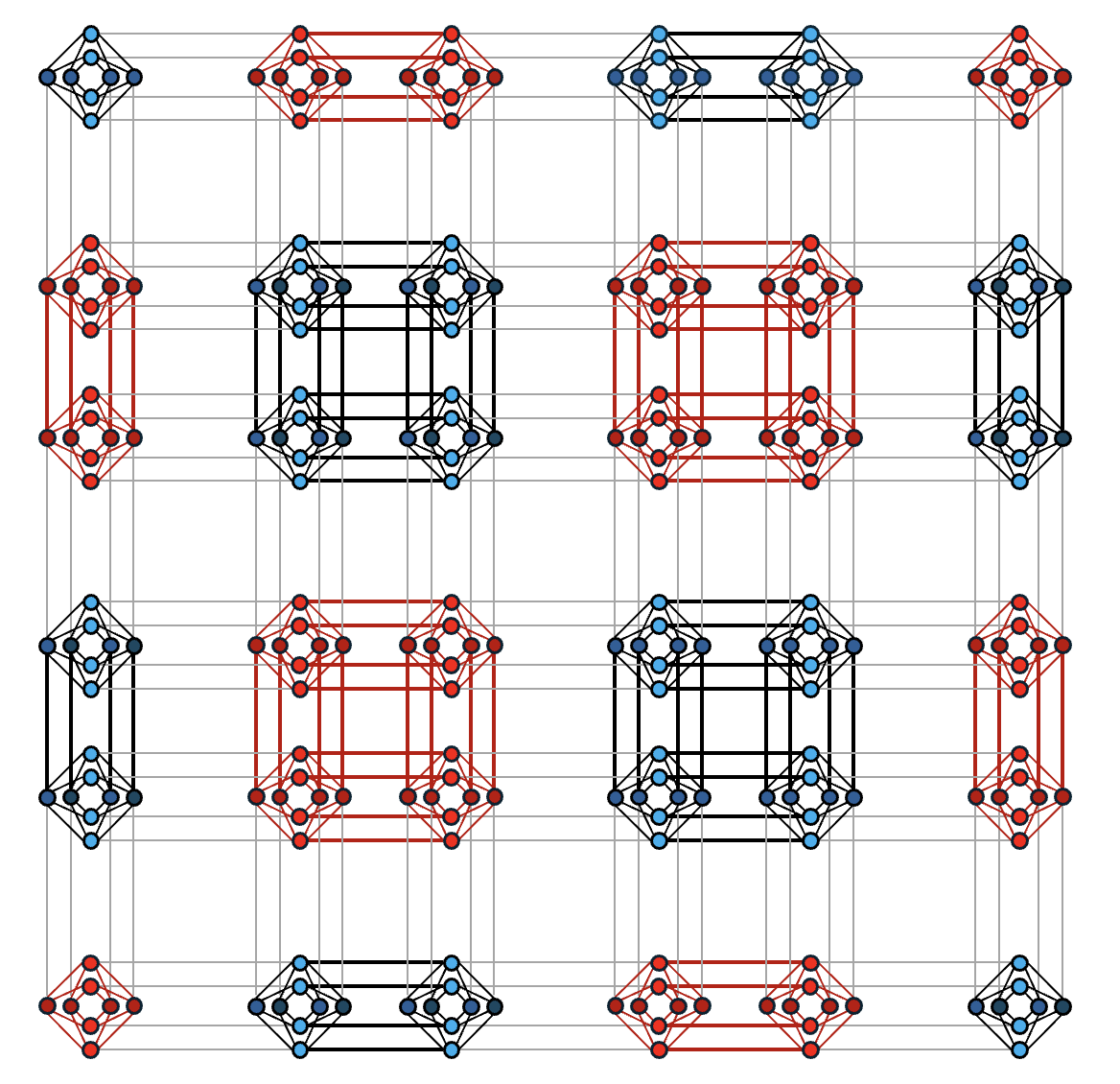}
\caption{The Chimera $C_6$ graph. One of its disconnected supernodes is shown here in blue, and the other in red. Edges linking the supernodes are in grey.
Any $C_{2(2k+1)}$ has two disconnected supernodes, each having $2(k+1)^2$ components. Representative components are shown in Figure~\ref{fig:chim_components}. }
\label{fig:chim_c6}
\end{figure}

$C_{2(2m)}$ is not a star product, since the ``supernodes'' defined as in the $C_{2(2m+1)}$ case are not isomorphic. This can be observed in Figure~\ref{fig:chim_internal}, where the blue supernode of this $C_{4}$ would be $C_2 \cup K_{4,4}\cup K_{4,4}\cup K_{4,4}\cup K_{4,4}$, and the red supernode would be a union of four joined $K_{4,4}$ pairs. Thus, the blue and red supernodes would not be isomorphic.

However, one could construct a Chimera-like $TC_{2(2m)}$ of the form 
$$
    P_2*\left(\bigcup_{2m^2}  C_2\right)
$$
by forming a torus out of the Chimera $C_{2(2m)}$ plane in the usual way of forming a torus from a plane, constructing $C_2$ instances by joining $\widehat {K_{4,4}}$ instances on the $4$ sides of the plane to the corresponding $\widehat {K_{4,4}}$ on the opposite side, and by joining $K_{4,4}$ instances on the corners of the plane. This may be seen in Figure~\ref{fig:chim_internal}, which would construct a Toric-Chimera $TC_4$ graph.

This toric construction would have the advantage of producing a regular graph where each vertex has $6$ neighbors, in contrast to the standard Chimera, where all of the vertices in the $K_{4,4}$ in the plane corners and half the vertices in the $\widehat {K_{4,4}}$ instances along the plane sides have only $5$ neighbors. D-Wave did not pursue a toric Chimera variant, but instead moved on to the Pegasus graph.
\begin{figure}[ht]
\centering
\includegraphics[width=.5\columnwidth]{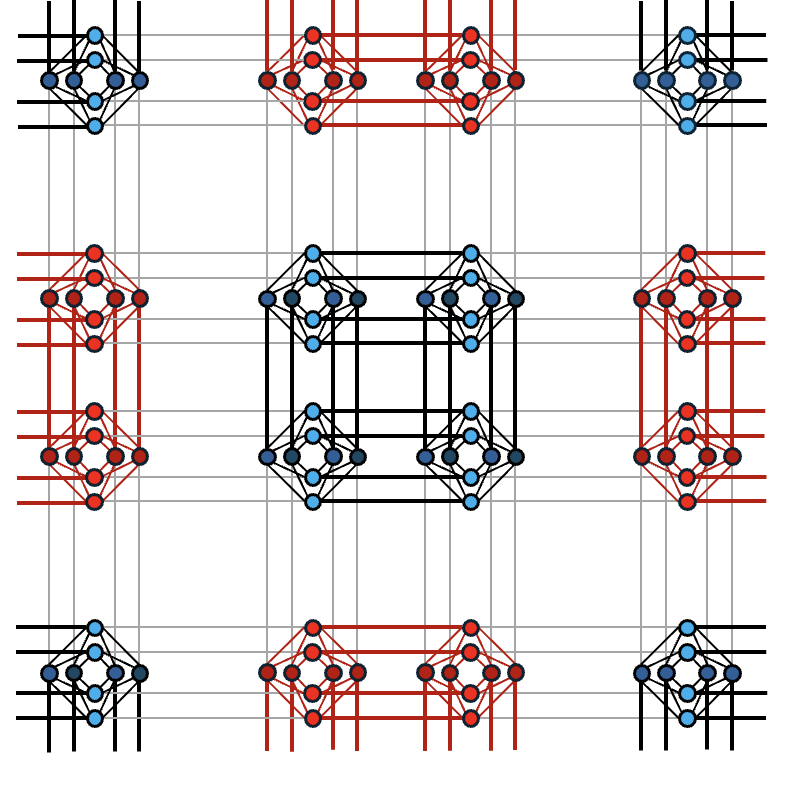}
\caption{An internal portion of a Chimera graph. Dark blue edges are those included in the blue supernode, and dark red edges are included in the red supernode.  Edges linking the supernodes are in grey. This may be used to produce a Toric-Chimera $TC_{2(2k)}$ graph by constructing a torus in the usual way from a plane.}
\label{fig:chim_internal}
\end{figure}

\subsubsection{Upper Bounds on the Number of EDSTs in the Chimera Graph}
\emph{Case 1:} $C_{2(2k+1)}$ is a $P_2*U$
$$
U= \left(\bigcup_{2k^2}  C_2 \bigcup_{4k} \widehat{{K}_{4,4}} \bigcup_2 K_{4,4}\right),
$$
with $\widehat{{K}_{4,4}}$ a graph made of two copies of $K_{4,4}$, with the vertices of one of the partitions of the first $K_{4,4}$ joined bijectively to the vertices of one of the partitions of the other $K_{4,4}$. $U$ is disconnected with $2k^2+4k+2=2(k+1)^2$ components.
We have that 
\begin{align*}
    |V_s| &= 2,\\
    |E_s| &= 1,\\
    |V_n| &= 32(2k^2)+16(4k) +8\cdot 2,\\
    &= 16(2k+1)^2,\\
    |E_n| &= 80(2k^2)+36(4k)+16(2)\\
    &= 16(2k+1)(5k+2)
\end{align*}
So, by Property~\ref{property:star_props},
\begin{align*}
    |V| &= 32(2k+1)^2\\
    |E| &= 2\cdot 16(5k+2)(2k+1) + 16(2k+1)^2 \\
    &= 2\cdot 32(2k+1)^2+32k(2k+1)) + 16(2k+1)^2 \\
    &= 2\cdot |V|+32(2k+1)\left(2k+\frac{1}{2}\right) 
\end{align*}
So we have $m=2$, and $c=32(2k+1)(2k+\frac{1}{2})<32(2k+1)^2=|V|$. Clearly, both $m$ and $c$ are positive integers, and
\begin{align*}
m+c
&=2+32(2k+1)(2k+\frac{1}{2})\\
&= 32(2k+1)^2 - 16(2k+1)+2 \\
&= |V| - (32k +14) \\
&< |V| - 1,
\end{align*}
so we may apply Proposition \ref{prop:max_bounds_uv} to find that the largest possible number of EDSTs in $C_{2(2k+1)}$ is 2.  
\newline\newline
\emph{Case 2:} The proposed $TC_{2(2k)}$ is a $P_2*U$, where
$$
    U=\left(\bigcup_{2k^2}  C_2\right)
$$
$U$ is disconnected with $2k^2$ components. We have that 
\begin{align*}
    |V_s| &= 2,\\
    |E_s| &= 1,\\
    |V_n| &= 32(2k^2), \\
    |E_n| &= 80 (2k^2).
\end{align*}
So, by Property~\ref{property:star_props},
\begin{align*}
    |V| &= 2*32(2k^2)=128k^2\\
    |E| &= 2 \cdot 80(2k^2) + 32(2k^2)= 384k^2\\
    &= 3\cdot |V|. 
\end{align*}
Here, $m=3$, and $c=0<v$, so $m+c=3<|V|-1$. Again applying  Proposition \ref{prop:max_bounds_uv}, the largest possible number of EDSTs in $TC_{2(2k)}$ is 3. Permitting the ``looping around'' of the torus-like structure gives a possible extra spanning tree in its set of EDSTs.

\end{appendices}
\end{document}